\documentclass[sn-mathphys,Numbered]{sn-jnl}% Math and Physical Sciences Reference Style
\usepackage{graphicx}%
\usepackage{multirow}%
\usepackage{amsmath,amssymb,amsfonts}%
\usepackage{amsthm}%
\usepackage{mathrsfs}%
\usepackage[title]{appendix}%
\usepackage{xcolor}%
\usepackage{textcomp}%
\usepackage{manyfoot}%
\usepackage{booktabs}%
\usepackage{algorithm}%
\usepackage{algorithmicx}%
\usepackage{algpseudocode}%
\usepackage{listings}%
\usepackage{enumerate}

\usepackage{array}

\newcommand{\real}{\mathbb{R}}

\newcommand{\mat}[1] {\mathtt{#1}}
\newcommand{\supp} {\mathtt{supp}}
\newcommand{\core}{\mathtt{core}}

\newcommand{\iprod}[1]  {\bigl\langle #1 \,\bigr\rangle}
\newcommand{\norm}[1]   {\| #1 \|}
\newcommand{\cancel}[1] {}
\theoremstyle{thmstyleone}%
\newtheorem{theorem}{Theorem}%  meant for continuous numbers
\newtheorem{lemma}{Lemma}

\theoremstyle{thmstyletwo}%

\theoremstyle{thmstylethree}%
\newtheorem{definition}{Definition}%

\begin{document}

\title[Article Title]{On Non-Negative Quadratic Programming \\
	in Geometric Optimization}


\author*[12]{\fnm{Siu Wing} \sur{Cheng}}\email{scheng@cse.ust.hk}
\equalcont{These authors contributed equally to this work.}

\author[13]{\fnm{Man Ting} \sur{Wong}}\email{mtwongaf@connect.ust.hk}
\equalcont{These authors contributed equally to this work.}

\affil[1]{\orgdiv{Department of Computer Science and Engineering},
	 \orgname{\\The Hong Kong University of Science and Technology}, \orgaddress{\\ \city{Hong Kong}, \country{China}}}
\affil[2]{ORCID: { 0000-0002-3557-9935}}
\affil[3]{ORCID: { 0000-0002-5682-0003}}


\begin{center}
	\Large On Non-Negative Quadratic Programming in Geometric Optimization\\
	\vspace{1cm}
	\large Siu Wing Cheng\footnote{Corresponding author. E-mail: scheng@cse.ust.hk. ORCID: { 0000-0002-3557-9935}}\footnotemark[3] and Man Ting Wong\footnote{Contributing author. Email: mtwongaf@connect.ust.hk. ORCID: { 0000-0002-5682-0003}}\footnote[3]{These authors contributed equally to this work.}\\
	\vspace{0.5cm}
	Department of Computer Science and Engineering,\\
	The Hong Kong University of Science and Technology, \\
	Hong Kong, China.\\
	\vspace{0.5cm}

\end{center}

%	\vspace{1cm}
%\begin{center}
%	\large Siu Wing Cheng and Man Ting Wong
%	
%\end{center}
%	\vspace{1cm}

%%================================%%
%% Sample for structured abstract %%
%%================================%%



\cancel{

\begin{center}
	\textbf{\emph{Statements and Declarations}}
 	%\paragraph{Statements and Declarations} 
\end{center}

	This work was supported by Research Grants Council, Hong Kong, China (project no.~16203718). The authors have no competing interests to declare that are relevant to the content of this article. The datasets generated during and/or analysed during the current study are available from the corresponding author on reasonable request.
	
}

\begin{center}
	\vspace{0.25cm}
	\textbf{\emph{Abstract}}
\end{center}

	We study a method to solve several non-negative quadratic programming problems in geometric optimization by applying a numerical solver iteratively. 
	We implemented the method by using {\tt quadprog} of MATLAB as a subroutine.  In comparison with a single call of {\tt quadprog}, we obtain a 10-fold speedup for two proximity graph problems in $\real^d$ on some public data sets, and a 2-fold or more speedup for deblurring some gray-scale space and thermal images via non-negative least square. In the image deblurring experiments, the iterative method compares favorably with other software that can solve non-negative least square, including FISTA with backtracking, SBB, FNNLS, and {\tt lsqnonneg} of MATLAB.   We try to explain the observed efficiency by proving that, under a certain assumption, the number of iterations needed to reduce the gap between the current solution value and the optimum by a factor $e$ is roughly bounded by the square root of the number of variables.  We checked the assumption in our experiments on proximity graphs and confirmed that the assumption holds.


\section{Introduction}

\setcounter{page}{2}

We study a method to solve several \emph{non-negative convex quadratic programming problems} (NNQ problems):~minimize $\mat{x}^t\mat{A}^t\mat{Ax} + \mat{a}^t\mat{x}$~subject to~$\mat{Bx} \geq \mat{b}$ and $\mat{x} \geq 0_\nu$. It applies a numerical solver on the constrained subproblems iteratively until the optimal solution is obtained. 
Given a vector $\mat{x}$, we use $(\mat{x})_i$ to denote its $i$-th coordinate.  We use $1_\nu$ and $0_\nu$ to denote $\nu$-dimensional vectors that consist of 1s and 0s, respectively.
%We implemented the method by calling {\sf quadprog} of MATLAB.  In comparison with a single call of {\sf quadprog}, the iterative method is much more efficient when the intermediate solutions are very sparse compared with the number of variables, say the number of positive variables is no more than 20\% of all variables.

The first problem is fitting a {\em proximity graph} to data points in $\real^d$, which finds applications in classification, regression, and clustering~\cite{daitch2009,jebara2009,kalofolias2016,ZHL2014}.  
%It is desirable to have a sparse graph connecting the data points that reflect their proximity, and subsequent processing will be faster if there are fewer edges in the graph. 
Daitch~et~al.~\cite{daitch2009} defined a proximity graph via solving an NNQ problem.  The unknown is a vector $\mat{x} \in \real^{n(n-1)/2}$, specifying the edge weights in the complete graph on the input points $\mat{p}_1,\ldots,\mat{p}_n$.  Let $(\mat{x})_{i \vartriangle j}$ denote the coordinate of $\mat{x}$ that stores the weight of the edge $\mat{p}_i\mat{p}_j$.  Edge weights must be non-negative; each point must have a total weight of at least 1 over its incident edges; these constraints can be modelled as $\mat{Bx} \geq 1_n$ and  $\mat{x} \geq 0_{n(n-1)/2}$, where $\mat{B} \in \real^{n \times n(n-1)/2}$ is the incidence matrix for the complete graph, and $1_n$ is the $n$-dimensional vector with all coordinates equal to 1.  The objective is to minimize $\sum_{i=1}^n \bigl\|\sum_{j \in [n]\setminus\{i\}} (\mat{x})_{i\vartriangle j}(\mat{p}_i - \mat{p}_j)\bigr\|^2$, which can be written as $\mat{x}^t\mat{A}^t\mat{Ax}$ for some matrix $\mat{A} \in \real^{dn \times n(n-1)/2}$.
%$\mat{A}$ resembles the incidence matrix in the sense that every non-zero entry of the incidence matrix gives rise to $d$ non-zero entries in the same column in $\mat{A}$.   
%The edges with positive weights form the output graph.  
We call this the DKSG problem.   The one-dimensional case can be solved in $O(n^2)$ time~\cite{CCLR21}.

We also study another proximity graph in $\real^d$ defined by Zhang~et~al.~\cite{ZHL2014} via minimizing  $\frac{1}{d}\mat{b}^t\mat{x} + \frac{\mu}{2}\norm{\mat{Ux}-1_n}^2 + \frac{\rho}{2}\norm{\mat{x}}^2$ for some constants $\mu,\rho \geq 0$.   The vector $\mat{x} \in \real^{n(n-1)/2}$ stores the unknown edge weights.  For all distinct $i, j \in [n]$, the coordinate $(\mat{b})_{i\vartriangle j}$ of $\mat{b}$ is equal to $\norm{\mat{p}_i-\mat{p}_j}^2$.  The matrix $\mat{U}$ is the incidence matrix for the complete graph.  The objective function can be written as $\mat{x}^t\mat{A}^t\mat{A}\mat{x} + \mat{a}^t\mat{x}$, where $\mat{A}^t = \bigl[ \frac{\mu}{2}\mat{U}^t \,\,\, \frac{\rho}{2}\mat{I}_{n(n-1)/2} \bigr]$ and $\mat{a} = \frac{1}{d}\mat{b} - \mu\mat{U}^t1_n$.  The only  constraints are $\mat{x} \geq 0_{n(n-1)/2}$.  We call this  the ZHLG problem.  

The third problem is to deblur some mildly sparse gray-scale images~\cite{HNV00,LB91}. The pixels in an image form a vector $\mat{x}$ of gray scale values in $\real^d$; the blurring can be modelled by the action of a matrix $\mat{A} \in \real^{d \times d}$ that depends on the particular point spread function adopted~\cite{LB91}; and $\mat{Ax}$ is the blurred image.  Working backward, given the matrix $\mat{A}$ and an observed blurred image~$\mat{b}$, we recover the image by finding the unknown $\mat{x} \geq 0_d$ that minimizes $\norm{\mat{Ax}-\mat{b}}^2$, which is a \emph{non-negative least square problem} (NNLS)~\cite{KSD13,LH95}.  The problem is often ill-posed, i.e., $\mat{A}$ is not invertible.

\subsection{Related work}  
Random sampling has been applied to select a subset of violated constraints in an iterative scheme to solve linear programming and its generalizations (which includes convex quadratic programming)~\cite{C1995,DF1989,G1995,MSW1996,SW1992}.  For linear programming in $\real^d$ with $n$ constraints, Clarkson's Las Vegas method~\cite{C1995} starts with a small working set of constraints and runs in iterations; it draws a random sample of the violated constraints at the beginning of each iteration, adds them to the working set of constraints, and recursively solves the linear programming problem on the expanded working set. %The expected running time is $O(d^2 n)+(\log n)O(d)^{d/2+O(1)}+O(d^4 \sqrt{n} \log n)$.  
This algorithm solves approximately $d$ subproblems in expectation, with $O(d\sqrt{n})$ constraints in each subproblem. In addition, Clarkson~\cite{C1995} provides an iterative reweighing algorithm for linear programming that solves $O(d \log n)$ subproblems in expectation with $O(d^2)$ constraints in each subproblem. A mixed algorithm is proposed in~\cite{C1995} by following the first algorithm but solving every subproblem using the iterative reweighing algorithm. All algorithms run in the RAM model and the expected running time of the mixed algorithm is $O(d^2 n)+(d^2 \log n)O(d)^{d/2+O(1)}+O(d^4 \sqrt{n} \log n)$. %The expected running time of the mixed algorithm is linear in $n$ and exponential in $d$.
Subsequent research by Sharir and Welzl~\cite{SW1992} and Matou\v{s}ek~et~al.~\cite{MSW1996} improved the expected running time to $O(d^2 n + e^{O(\sqrt{d\ln d})})$ in the RAM model which is subexponential in $d$. Brise and G\"{a}rtner~\cite{BG2011} showed that Clarkson's algorithm works for the larger class of violator spaces~\cite{GMRS2008}. 

Quadratic programming problems are convex when the Hessian is positive (semi-)definite. When the Hessian is indefinite, the problem becomes non-convex and NP-hard. For non-convex quadratic programming, the goal is to find a local minimum instead of the global minimum.

When all input numbers are rational numbers, one can work outside the RAM model to obtain polynomial-time algorithms for convex quadratic programming. Kozlov~et~al.~\cite{KTK1979} described a polynomial-time time algorithm for convex quadratic programming that is based on the ellipsoid method~\cite{Bland1980TheEM,KHACHIYAN198053}. Kapoor and Vaidya~\cite{KV1986} showed that convex quadratic programming with $n$ constraints and $m$ variables can be solved in $O((n+m)^{3.67} (\log L)(\log (n+m))L)$ arithmetic operations, where $L$ is bounded by the number of bits in the input. Monteiro and Adler~\cite{monteiro89} showed that the running time can be improved to $O(m^3 L)$, provided that $n = O(m)$.  

A simplex-based exact convex quadratic programming solver was developed by G\"{a}rtner and Sch\"{o}nherr. It works well when the constraint matrix is dense and there are few variables or constraints~\cite{GS2000}.  However, it has been noted that this solver is inefficient in high dimensions due to the use of arbitrary-precision arithmetic~\cite{FGK2003}. Generally, iterative methods are preferred for large problems because they preserve the sparsity of the Hessian matrix~\cite{Yang1991}. 

%Our method is related to solving large-scale quadratic programming by the active set method. The active set method can be viewed as turning the active inequality constraints into equality constraints, resulting in a simplified subproblem. Large-scale convex or non-convex quadratic programming algorithms often assume box constraints, i.e., each variable has an upper bound and lower bound; however, these bounds may not necessarily be finite.

Our method is related to solving large-scale quadratic programming by the active set method. The active set method can be viewed as turning the active inequality constraints into equality constraints, resulting in a simplified subproblem. Large-scale quadratic programming algorithms often assume box constraints, i.e., each variable has an upper bound and lower bound; however, these bounds may not necessarily be finite. In the subsequent discussion, we will provide a brief overview of the existing active-set algorithms for large-scale quadratic programming found in the literature.

\paragraph{Box constraints only} Coleman and Hulbert~\cite{Coleman1989} proposed an active set algorithm that works for both convex and non-convex quadratic programming problems. The algorithm is designed for the case that there are only box constraints with finite bounds. Therefore, having a variable in the active set means requiring that variable to be at its lower or upper bound. As in~\cite{bertsekas1982, CM87, Conn1988}, their method identifies several variables to be added to the active set. Given a search direction, the projected search direction is a projection of the search direction that sets all its coordinates that are in the active set to zero. Consequently, following the projected search direction would not alter the variables in the active set. This descent process begins by following the projected search direction until a coordinate hits its lower or upper bound. The algorithm updates the active set by adding the variable that just hits its lower or upper bound. This process is repeated until the projected search direction is no longer a descent direction or a local minimum is reached along the projected search direction. If the solution is not optimal, the algorithm uses the gradient to identify a set of variables in the active set such that, by allowing these variables to change their values, the objective function value will decrease. In each iteration, the algorithm releases this set of variables from the active set. Thus allows for the addition and removal of multiple variables from the active set in each iteration, guaranteeing that the active set does not experience monotonic growth. This is considered advantageous because active set algorithms often suffer from only being able to add or delete a single constraint from the active set per iteration.

Yang and Tolle~\cite{Yang1991} proposed an active set algorithm for convex quadratic programming with box constraints that can remove more than one box constraint from the active set in each iteration. The descent in each iteration is performed using the conjugate gradient method~\cite{Hestenes1952}, which ignores the box constraints and generates a descent path that either stops at the optimum or at the intersection with the boundary of the feasible region. In the latter case, some variables reach their lower or upper bounds. The algorithm maintains a subset $S$ of such variables and the bounds reached by them that were recorded in previous iterations. In the main loop of the algorithm, the first step is to compute the set $B$ of variables that are at their lower or upper bounds such that the objective function value is not smaller at any feasible solution obtained by adjusting the variables in $B$. Then, the algorithm removes one constraint at a time from $S$ until the shrunken $S$ induces a descent direction. If $S$ becomes empty and there is no feasible descent direction, then the solution is optimal. If there exists a descent direction, then $S \cup B$ forms the new active set. The algorithm solves the constrained problem induced by this new active set using the conjugate gradient method~\cite{Hestenes1952} which we already mentioned earlier. Empirically, the algorithm performs better for problems that have a small set of free variables at optimality, and for problems with only non-negativity constraints.

Mor\'e and Toraldo~\cite{More1991} proposed an algorithm for strictly convex quadratic programs subject to box constraints. The box constraints induce the faces of a hyper-rectangle, which form the boundaries of the feasible region. Furthermore, an active set identifies a face of the hyper-rectangle. The main loop iterates two procedures.  The first procedure performs gradient descent until the optimal solution is reached or the gradient descent path is intercepted by a facet $F$ of the feasible region.  In the first case, the algorithm terminates.  In the latter case, the affine subspace of $F$ is defined by an active set $S$ of variables at their upper or lower bounds, and the second procedure is invoked.  It runs the conjugate gradient method on the original problem constrained by the active set $S$; it either finds the optimum of this constrained problem in the interior of $F$, or descends to a boundary face of $F$. In either case, the algorithm proceeds to the next iteration of the main loop.

Friedlander et al.~\cite{Friedlander1995} proposed an algorithm for convex quadratic programs subject to box constraints. The structure of the algorithm is similar to~\cite{More1991}. The aglorithm performs gradient descent until the descent path is intercepted by a facet $F$ of the feasible region. To solve the original problem constrained on the facet $F$, instead of using the conjugate gradient method, it employs the Barzilai-Borwein method to find the optimum of this constrained problem. The advantage of the Barzilai-Borwein method is that it requires less memory storage and fewer computations than the conjugate gradient method.

\paragraph{Box constraints and general linear constraints} Gould~\cite{Gould1991} proposed an active set algorithm for non-convex quadratic programming problems subject to box constraints and linear constraints. The algorithm adds or removes  at most one constraint from the active set per iteration. In each iteration, the algorithm constructs a system of linear equations induced by the active set, solves it, and adds this solution to current feasible solution of the quadratic program. The algorithm maintains an LU factorization of the coefficient matrix to exploit the sparsity of the system of equations. This approach allows for efficient updates to the factorization as the active set changes.

Bartlett and Biegler~\cite{bartlett2006} proposed an algorithm for convex quadratic programming subject to box constraints and linear inequality constraints. The algorithm maintains a dual feasible solution and a primal infeasible solution in the optimization process. In each iteration, the algorithm adds a violated primal constraint to the active set. Then, the Schur complement method~\cite{Gill1987ASM} is used to solve the optimisation problem constrained by the active set. If a dual variable reaches zero in the update, the corresponding constraint is removed from the active set.

\paragraph{Our contributions}  The method that we study greedily selects a subset of variables as free variables---the non-free variables are fixed at zero---and calls the solver to solve a constrained NNQ problem in each iteration.  At the end of each iteration, the set of free variables is updated in two ways.  First, all free variables that are equal to zero are made non-free in the next iteration.  Second, a subset of the non-free variables that violate the dual feasibility constraint the most are set free in the next iteration.  However, if there are too few variables that violate the dual feasibility constraint or the number of iterations exceeds a predefined threshold, we simply set free all variables that violate the dual feasibility constraint.

A version of the above greedy selection of free variables in an iterative scheme was already used in~\cite{daitch2009} for solving the DKSG problem because it took too long for a single call of a quadratic programming solver to finish.  However, it is only briefly mentioned in~\cite{daitch2009} that a small number of the most negative gradient coordinates are selected to be freed.  No further details are given.  We found that the selection of the number of variables to be freed is crucial. The right answer is not some constant number of variables or a constant fraction of them.  Also, one cannot always turn free variables that happen to be zero in an iteration into non-free variables in the next iteration as suggested in~\cite{daitch2009}; otherwise, the algorithm may not terminate in some cases.

Our method also uses the gradient to screen for descent directions, allowing us to release more than one variable at each iteration. However, instead of setting free all variables that have a non-zero gradient from the active set as in~\cite{BG2011, Coleman1989, Friedlander1995, More1991}, we choose a parameter that enables us to remove just the right amount of variables from the active set so that our subproblem does not become too large. Heuristically, this approach is faster than freeing all variables with a non-zero gradient or freeing only  one variable, as in~\cite{bartlett2006, Gould1991}. Similar to the method presented in~\cite{Yang1991}, we remove variables from the free variable set, enabling us to keep our subproblem small.  This is in contrast to the non-decreasing size of the free variable set, as in~\cite{BG2011, C1995}. Thus, our algorithm performs well when solution sparsity is observed.

Clarkson's linear programming algorithm~\cite{C1995} is designed for cases when the number of variables is much less than the number of given constraints. However, this is not the case for some other problems. For example, consider the DSKG problem with $n$ input points. Although the ambient space is $\mathbb{R}^d$, there are $\Theta(n^2)$ variables and $\Theta(n^2)$ constraints. By design, the number of constraints in a subproblem of Clarkson's algorithm is proportional to the product of the number of variables and the square root of the total number of constraints. In our context, the subproblem has $\Theta(n^3)$ constraints. Consequently, solving such a subproblem is not easier than solving the whole problem. Therefore, Clarkson's algorithm cannot be directly applied to our proximity graph and NNLS problems, as in these problems, the number of constraints and variables are proportional to each other.

We implemented the iterative method by using {\tt quadprog} of MATLAB as a subroutine.  In comparison with a single call of {\tt quadprog},   we obtain a 10-fold speedup for DKSG and ZHLG on some public data sets; we also obtain a 2-fold or more speedup for deblurring some gray-scale space and thermal images via NNLS. In the image deblurring experiments, the iterative method compares favorably with other software that can solve NNLS, including FISTA with backtracking~\cite{fistabt,BT09}, SBB~\cite{KSD13}, FNNLS~\cite{LH95}, and {\tt lsqnonneg} of MATLAB.  We emphasize that we do not claim a solution for image deblurring as there are many issues that we do not address; we only seek to demonstrate the potential of the iterative scheme.

The iterative method gains efficiency when the solutions for the intermediate constrained subproblems are sparse and the numerical solver runs faster in such cases.  Due to the overhead in setting up the problem to be solved in each iteration, a substantial fraction of the variables in an intermediate solution should be zero in order that the saving outweighs the overhead.

We try to offer a theoretical explanation of the observed efficiency.  Let $f : \real^\nu \rightarrow \real$ denote the objective function of the NNQ problem at hand.  Let $\langle \mat{x},\mat{y} \rangle$ denote the inner product of two vectors $\mat{x}$ and $\mat{y}$.  A unit direction $\mat{n} \in \real^\nu$ is a \emph{descent direction} from a feasible solution $\mat{x}$ if $\langle \nabla f(\mat{x}), \mat{n} \rangle <  0$ and $\mat{x} + s\mat{n}$ lies in the feasible region for some $s > 0$. 
% We offer a partial theoretical analysis.  
Let $\mat{x}_r$ be the solution produced in the $(r-1)$-th iteration.  Let $\mat{x}_*$ be the optimal solution. Using elementary vector calculus, one can show that $\frac{f(\mat{x}_{r+1})-f(\mat{x}_*)}{f(\mat{x}_r)-f(\mat{x}_*)} \leq 1 - \frac{1}{2}\cdot\frac{\norm{\mat{x}_r-\mat{y}}}{\norm{\mat{x}_r-\mat{x}_*}} \cdot \frac{\langle \nabla f(\mat{x}_r), \mat{n}\rangle}{\langle \nabla f(\mat{x}_r),\mat{n}_*\rangle}$,
%= 1 - \frac{1}{2}\cdot\frac{\norm{\mat{x}_r-\mat{y}}}{\norm{\mat{x}_r-\mat{x}_*}} \cdot \frac{\cos\angle(\nabla f(\mat{x}_r), \mat{n})}{\cos \angle(\nabla f(\mat{x}_r),\mat{n}_*)}$, 
where $\mat{y}$ is minimum in direction $\mat{n}$ from $\mat{x}_r$, and $\mat{n}_*$ is the unit vector $(\mat{x}_*-\mat{x}_r)/\norm{\mat{x}_*-\mat{x}_r}$.  We prove that there exists a descent direction $\mat{n}_r$ from $\mat{x}_r$ for which only a few variables need to be freed and $\frac{\langle \nabla f(\mat{x}_r), \mat{n}_r\rangle}{\langle \nabla f(\mat{x}_r),\mat{n}_*\rangle}$ is roughly bounded from below by the reciprocal of the square root of the number of variables.  Therefore, if we assume that the distance between $\mat{x}_r$ and the minimum $\mat{y}_r$ in direction $\mat{n}_r$ is at least $\frac{1}{\lambda}\norm{\mat{x}_r-\mat{x}_*}$ for all $r$, where $\lambda$ is some fixed value, then $f(\mat{x}_{r+i}) - f(\mat{x}_*) \leq e^{-1} (f(\mat{x}_r)-f(\mat{x}_*))$ for some $i$ that is roughly bounded by  $\lambda$ times the square root of the number of variables.  %For DKSG, we also need a generic assumption that will be explained later.  
%This gives a geometric-like convergence if this assumption holds.  
That is, the gap from $f(\mat{x}_*)$ will be reduced by a factor $e$ in at most $i$ iterations.  It follows that SolveNNQ will converge quickly to the optimum under our assumption.  The geometric-like convergence still holds even if the inequality $\norm{\mat{x}_r-\mat{y}_r} \geq \frac{1}{\lambda} \norm{\mat{x}_r-\mat{x}_*}$ is satisfied in only a constant fraction of any sequence of consecutive iterations (the iterations in that constant fraction need not be consecutive).

We ran experiments on the DKSG and ZHLG problems to check the assumption of our convergence analysis. We checked the ratios $\frac{\norm{\mat{x}_r-\mat{y}}}{\norm{\mat{x}_r-\mat{x}_*}}$ using the descent direction $\frac{\mat{x}_{r+1}-\mat{x}_r}{\norm{\mat{x}_{r+1}-\mat{x}_r}}$ taken by our algorithm as $\mat{n}$. Therefore, $\mat{y}$ is $\mat{x}_{r+1}$. Our experiments show that  $\frac{\norm{\mat{x}_r-\mat{y}}}{\norm{\mat{x}_r-\mat{x}_*}}\ge \frac{1}{100}$ in every iteration, with the exception of the outliers from the DKSG problem. However, the convergence still holds despite the outliers. Most of the runs achieve better results, as shown in Figures~\ref{fg:Rs}(a) and (b) in Section~\ref{sec:experiments}. 
For the ZHLG problem, we observe that $\frac{\norm{\mat{x}_r-\mat{y}}}{\norm{\mat{x}_r-\mat{x}_*}}\ge \frac{1}{2}$ holds in every iteration. For the DKSG problem, the majority of the ratios are greater than or equal to $\frac{1}{2}$, with only a few exceptions that are less than $\frac{1}{2}$.

\section{Algorithm}
\label{sec:alg}

\paragraph{Notation}  We use uppercase and lowercase letters in typewriter font to denote matrices and vectors, respectively.  The inner product of $\mat{x}$ and $\mat{y}$ is written as $\langle \mat{x},\mat{y} \rangle$ or $\mat{x}^t\mat{y}$.  We use $0_m$ to denote the $m$-dimensional zero vector and $1_m$ the $m$-dimensional vector with all coordinates equal to 1.  

The \emph{span} of a set of vectors $W$ is the linear subspace spanned by them; we denote it by $\mathtt{span}(W)$.  In $\real^\nu$, for $i \in [\nu]$, define the vector $\mat{e}_i$ to have a value 1 in the $i$-th coordinate and zero at other coordinates.  Therefore, $\mat{e}_1,\mat{e}_2,\ldots,\mat{e}_{\nu}$ form an orthonormal basis of $\real^\nu$.  

A \emph{conical combination} of a set of vectors $\{\mat{w}_1,\mat{w}_2,\ldots, \mat{w}_\nu\}$ is $\sum_{i=1}^\nu c_i\mat{w}_i$ for some non-negative real values $c_1, c_2, \ldots, c_\nu$.   For example, the set of all conical combinations of $\{\mat{e}_1,\mat{e}_2,\ldots,\mat{e}_\nu\}$ form the positive quadrant of $\real^\nu$.

Given a vector $\mat{x}$, we denote the $i$-th coordinate of $\mat{x}$ by putting a pair of brackets around $\mat{x}$ and attaching a subscript~$i$ to the right bracket, i.e., $(\mat{x})_i$.  This notation should not be confused with a vector that is annotated with a subscript.  For example, $(\mat{e}_i)_i = 1$ and $(\mat{e}_i)_j = 0$ for all $j \not= i$.  As another example, $\mat{x}_r$ means the intermediate solution vector that we obtain in the $(r-1)$-th iteration, but $(\mat{x}_r)_i$ refers to the $i$-th coordinate of $\mat{x}_r$.  The vector inside the brackets can be the result of some operations; for example, $(\mat{M}\mat{x})_i$ refers to the $i$-th coordinate of the vector formed by multiplying the matrix $\mat{M}$ with the vector $\mat{x}$.  We use $\supp(\mat{x})$ to denote $\bigl\{i  : (\mat{x})_i \not= 0  \bigr\}$.

\begin{figure}[b]
	\centering
	%\vspace*{20pt}
	\begin{tabular}{rcccrcc}
		$\min$ & $f(\mat{x})$ & & & & $\displaystyle \max_{\mat{u},\tilde{\mat{u}},\mat{v}} $ & $\displaystyle \min_{\mat{x}} \,\, g(\mat{u},\tilde{\mat{u}}, \mat{v},\mat{x})$ \\ 
		s.t. & $\mat{Bx} \geq \mat{b}$, $\mat{Cx} = \mat{c}$,  & & & & s.t. & $\mat{u} \geq 0_\kappa$, \\
		& $\forall \, i \not\in S_r, \, (\mat{x})_i \geq 0$, & & & & & $\forall \, i \not\in S_r, \,(\mat{v})_i \geq 0$. \\ 
		& $\forall \, i \in S_r, \, (\mat{x})_i = 0$.  \\ [0.25em]
		\multicolumn{3}{c}{(a)~Constrained NNQ} & \hspace*{.3in} &
		\multicolumn{3}{c}{(b)~Lagrange dual} \\
	\end{tabular}
	\caption{The constrained NNQ and its Lagrange dual problem.}
	\label{fg:NNM}
\end{figure}

\paragraph{Lagrange dual}  Let the objective function be denoted by $f : \real^\nu \rightarrow \real$.  Assume that the constraints of the original NNQ problem are $\mat{Bx} \geq \mat{b}$, $\mat{Cx} = \mat{c}$, and $\mat{x} \geq 0_\nu$ for some matrices $\mat{B}$ and $\mat{C}$.   Let $\kappa$ be the number of rows of $\mat{B}$.   The number of rows of $\mat{C}$ will not matter in our discussion.

For $r = 2, 3, \ldots$, the $(r-1)$-th iteration of the method solves a constrained version of the NNQ problem in which a subset of variables are fixed at zero. Figure~\ref{fg:NNM}(a) shows this constrained NNQ problem.   We use $S_r$ to denote the set of indices of these variables that are fixed at zero, and we call $S_r$ an \emph{active set}.  For convenience, we say that the variables whose indices belong to $S_r$ are the variables in the active set $S_r$.  We call the variables outside $S_r$ the \emph{free variables}.

We are interested in the Lagrange dual problem as it will tell us how to update the active set $S_r$ when proceeding to the next iteration~\cite{BV2004}.  Let $\mat{u}$ and $\tilde{\mat{u}}$ be the vectors of Lagrange multipliers for the constraints $\mat{Bx} \geq \mat{b}$ and $\mat{Cx} = \mat{c}$, respectively.  The dimensions of $\mat{u}$ and $\tilde{\mat{u}}$ are equal to number of rows of $\mat{B}$ and $\mat{C}$, respectively.  As a result, $\mat{u} \in \real^\kappa$.  In the dual problem, the coordinates of $\mat{u}$ are required to be non-negative as the inequality signs in $\mat{Bx} \geq \mat{b}$ are greater than or equal to, whereas the coordinates of $\tilde{\mat{u}}$ are unrestricted due to the equality signs in $\mat{Cx} = \mat{c}$.  Similarly, let $\mat{v}$ be the vector of Lagrange multipliers that correspond to the constraints $(\mat{x})_i \geq 0$ for $i \not\in S_r$ and $(\mat{x})_i = 0$ for $ i \in S_r$.  So $(\mat{v})_i$ is required to be non-negative for $i \not\in S_r$ whereas $(\mat{v})_i$ is unrestricted for $i \in S_r$.  The Lagrange dual function is $g(\mat{u},\tilde{\mat{u}},\mat{v},\mat{x})  = f(\mat{x}) + \mat{u}^t(\mat{b}-\mat{Bx}) + \tilde{\mat{u}}^t(\mat{c} - \mat{Cx}) - \mat{v}^t\,\mat{x}$.  The Lagrange dual problem is shown in Figure~\ref{fg:NNM}(b).   Note that $\mat{x}$ is unrestricted in the dual problem.

\paragraph{KKT conditions and candidate free variable}  By duality, the optimal value of the Lagrange dual problem is always a lower bound of the optimal value of the primal problem.  Since all constraints in Figure~\ref{fg:NNM}(a) are affine, strong duality holds in this case which implies that the optimal value of the Lagrange dual problem is equal to the optimal value of the primal problem~\cite[Section~5.2.3]{BV2004}.

Let $\mat{x}_r$ denote the optimal solution of the constrained NNQ problem in Figure~\ref{fg:NNM}(a).  That is, $\mat{x}_r$ is the optimal solution in the $(r-1)$-th iteration.  Let $(\mat{u}_r,\tilde{\mat{u}}_r,\mat{v}_r,\mat{x}_r)$ be the corresponding solution of the Lagrange dual problem.  For convenience, we say that $\mat{x}_r$ is the optimal solution of $f$ constrained by $S_r$, and that $(\mat{u}_r,\tilde{\mat{u}}_r,\mat{v}_r,\mat{x}_r)$ is the optimal solution of $g$ constrained by $S_r$.  

The Karush-Kuhn-Tucker conditions, or KKT conditions for short, are the sufficient and necessary conditions for $(\mat{u}_r,\tilde{\mat{u}}_r,\mat{v}_r,\mat{x}_r)$ to be an optimal solution for the Lagrange dual problem~\cite{BV2004}.  There are four of them:
\begin{itemize}
	\item Critical point:~$\displaystyle \frac{\partial g(\mat{x}_r)}{\partial \mat{x}}= \nabla f(\mat{x}_r) - \mat{B}^t\mat{u}_r - \mat{C}^t\tilde{\mat{u}}_r - \mat{v}_r = 0_\nu$.
	
	\item Primal feasibility:  $\mat{Bx}_r \geq \mat{b}$, $\mat{Cx}_r = \mat{c}$, $\mat{x}_r \geq 0_\nu$, and $\forall \, i \in S_r$, $(\mat{x}_r)_i = 0$.
	
	\item Dual feasibility: $\mat{u}_r \geq 0_\kappa$ and $\forall\, i \not\in S_r$, $(\mat{v}_r)_i \geq 0$.
	
	\item Complementary slackness: for all $i \in [\kappa]$, $(\mat{u}_r)_i \cdot (\mat{b} - \mat{Bx}_r)_i = 0$, and for all $i \in [\nu]$, $(\mat{v}_r)_i \cdot (\mat{x}_r)_i = 0$.
	
\end{itemize}

%By strong duality, the optimal value of $g$ constrained by $S_r$ is equal to the optimal value of $f$ constrained by $S_r$.  
If $\mat{x}_r$ is not an optimal solution for the original NNQ problem,
%We check $\nabla f(\mat{x}_r)$ to decide whether $(\mat{u}_r,\tilde{\mat{u}}_r,\mat{x}_r)$ is an optimal solution for the original NNQ problem.  Suppose not.
%The easy case is that $\nabla f(\mat{x}_r) = 0_\nu$, but this is usually not attainable due to the primal feasibility constraint.  In this case, 
%Optimality of the NNQ constrained by $S_r$ means that $\nabla f(\mat{x}_r)$ is orthogonal to the primal feasibility constraints that are tight.  
there must exist some $j\in S_r$ such  that if $j$ is removed from $S_r$, the KKT conditions are violated and hence $(\mat{u}_r,\tilde{\mat{u}}_r,\mat{v}_r,\mat{x}_r)$ is not an optimal solution for $g$ constrained by $S_r \setminus \{j\}$.  Among the KKT conditions, only the dual feasibility can be violated by removing $j$ from $S_r$; in case of violation, we  must have $(\mat{v}_r)_j < 0$.  Therefore, for all $j \in S_r$, if $(\mat{v}_r)_j < 0$, then $j$ is a candidate to be removed from the active set.  In other words, $(\mat{x})_j$ is a candidate variable to be set free in the next iteration.

\paragraph{Algorithm specification}  We use a trick that is important for efficiency.   For $r = 2, 3, \ldots$, before calling the solver to compute $\mat{x}_r$ in the $(r-1)$-th iteration, we eliminate the variables in the active set $S_r$ because they are set to zeros anyway.   Since there can be many variables in $S_r$ especially when $r$ is small, this step allows the solver to potentially run a lot faster.  Therefore, the solver cannot return any Lagrange multiplier $(\mat{v}_r)_j$ that corresponds to the variable $(\mat{x})_j$ in the active set $S_r$ because $(\mat{x})_j$ is eliminated before the solver is called.  This seems to be a problem because our key task at the end of the $(r-1)$-th iteration is to determine an appropriate subset of $j \in S_r$ such that $(\mat{v}_r)_j <  0$.   Fortunately, the first KKT condition tells us that $\mat{v}_r = \nabla f(\mat{x}_r) - \mat{B}^t\mat{u}_r  - \mat{C}^t\tilde{\mat{u}}_r$, so we can compute all entries in $\mat{v}_r$.  In other words, the optimal solution of the Lagrange dual problem constrained by $S_r$ is fully characterized by $(\mat{u}_r,\tilde{\mat{u}}_r,\mat{x}_r)$.

\begin{algorithm}
	\caption{SolveNNQ}
	\label{alg:1}
	\begin{algorithmic}[1]
		\State $\tau$ is an integer threshold much less than $\nu$
		\State $\beta_0$ is a threshold on $|E_r|$ and $\beta_0 > \tau$
		\State $\beta_1$ is a threshold on the number of iterations
		\State{compute an initial active set $S_1$ and an optimal solution $(\mat{u}_1,\tilde{\mat{u}}_1,\mat{x}_1)$ of $g$ constrained by $S_1$}
		\State $r \leftarrow 1$
		\Loop %\label{alg:loop-start}																							
		\State {$E_r \leftarrow$ the sequence of indices $i \in S_r$ such that $(\mat{v}_r)_i = (\nabla f(\mat{x}_r) - \mat{B}^t\mat{u}_r - \mat{C}^t\tilde{\mat{u}}_r)_i < 0$ and $E_r$ is sorted in non-decreasing order of $(\mat{v}_r)_i$}  				\label{alg:gradient}  \label{alg:loop-start}
		\If {$E_r = \emptyset$}
		
		\Return $\mat{x}_r$
		%{\RETURN $(\mat{u}_r,\tilde{\mat{u}}_r,\mat{x}_r)$}
		%\ELSIF{$|E_r| < \kappa$} 
		%{\STATE $E_r \leftarrow E$}
		%{\STATE $S_{r+1} \leftarrow S_r \setminus E_r$}
		\Else
		\If{$|E_r| < \beta_0$ or $r > \beta_1$} 						\label{alg:exceed-start}
		\State $S_{r+1} \leftarrow S_r \setminus E_r$				 \label{alg:exceed-end}
		\Else
		\State let $E'_r$ be subset of the first $\tau$ indices in $E_r$   \label{alg:free-start}
		\State $S_{r+1} \leftarrow [\nu] \setminus \bigl(\supp(\mat{x}_r) \cup E'_r\bigr)$  \label{alg:free}
		\EndIf
		\EndIf
		\State $(\mat{u}_{r+1},\tilde{\mat{u}}_{r+1},\mat{x}_{r+1}) \leftarrow$ optimal solution of $g$ constrained by the active set $S_{r+1}$						  \label{alg:qp}
		\State $r \leftarrow r + 1$  \label{alg:loop-end}
		\EndLoop
	\end{algorithmic}
\end{algorithm}

SolveNNQ in Algorithm~\ref{alg:1} describes the iterative method.  In line~\ref{alg:free} of SolveNNQ, we kick out all free variables that are zero in the current solution $\mat{x}_r$, i.e., if $i \not\in S_r$ and $i \not\in \supp(\mat{x}_r)$, move $i$ to $S_{r+1}$.  This step keeps the number of free variables small so that the next call of the solver will run fast.  When the computation is near its end, moving a free variable wrongly to the active set can be costly; it is preferable to shrink the active set to stay on course for the optimal solution $\mat{x}_*$.  Therefore, we have a threshold $\beta_0$ on $|E_r|$ and set $S_{r+1} \leftarrow S_r \setminus E_r$ in line~\ref{alg:exceed-end} if $|E_r| < \beta_0$.   The threshold $\beta_0$ gives us control on the number of free variables.  In some rare cases when we are near~$\mat{x}_*$, the algorithm may alternate between excluding a primal variable from the active set in one iteration and moving the same variable into the active set in the next iteration.  The algorithm may not terminate. Therefore, we have another threshold $\beta_1$ on the number of iterations beyond which the active set will be shrunk monotonically in line~\ref{alg:exceed-end}, thereby guaranteeing convergence.  The threshold $\beta_1$ is rarely exceeded in our experiments; setting $S_{r+1} \leftarrow S_r \setminus E_r$ when $|E_r| < \beta_0$ also helps to keep the the algorithm from exceeding the threshold $\beta_1$.  In our experiments, we set $\tau = 4\ln^2 \nu$, $\beta_0 = 3\tau$, and $\beta_1 = 15$.

Using elementary vector calculus, one can prove that $\frac{f(\mat{x}_{r+1}) - f(\mat{x}_*)}{f(\mat{x}_r) - f(\mat{x}_*)} \leq 1- 
\frac{1}{2}\cdot\frac{\norm{\mat{x}_r-\mat{y}}}{\norm{\mat{x}_r-\mat{x}_*}} \cdot 
\frac{\langle \nabla f(\mat{x}_r),\mat{n} \rangle}{\langle \nabla f(\mat{x}_r),\mat{n}_* \rangle}$ as stated in Lemma~\ref{lem:bound2} below, where $\mat{n}$ is any unit descent direction from $\mat{x}_r$, $\mat{y}$ is the minimum in direction $\mat{n}$ from $\mat{x}_r$, and $\mat{n}_*$ is the unit vector $(\mat{x}_*-\mat{x}_r)/\norm{\mat{x}_*-\mat{x}_r}$.  

In Sections~\ref{sec:zhlg} and~\ref{sec:second}, for the DKSG, ZHLG, and NNLS problems, we prove that there exists a descent direction $\mat{n}_r$ from $\mat{x}_r$ that involves only free variables in $[\nu] \setminus S_{r+1}$ and $\frac{\langle \nabla f(\mat{x}_r), \mat{n}_r\rangle}{\langle \nabla f(\mat{x}_r),\mat{n}_*\rangle}$ is roughly bounded from below by the square root of the number of variables.  Therefore, if we assume that the distance between $\mat{x}_r$ and the minimum in direction $\mat{n}_r$ is at least $\frac{1}{\lambda}\norm{\mat{x}_r-\mat{x}_*}$ for all $r$, where $\lambda$ is some fixed value, then $f(\mat{x}_{r+i}) - f(\mat{x}_*) \leq e^{-1} (f(\mat{x}_r)-f(\mat{x}_*))$ for some $i$ that is roughly bounded by  $\lambda$ times the square root of the number of variables.  It follows that if $\beta_1$ is not exceeded, SolveNNQ will converge efficiently to the optimum under our assumption.  The geometric-like convergence still holds even if the inequality $\norm{\mat{x}_r-\mat{y}_r} \geq \frac{1}{\lambda} \norm{\mat{x}_r-\mat{x}_*}$ is satisfied in only a constant fraction of any sequence of consecutive iterations (the iterations in that constant fraction need not be consecutive).  If $\beta_1$ is exceeded, Lemma~\ref{lem:SolveNNQ} below shows that the number of remaining iterations is at most $|\supp(\mat{x}_*)|$; however, we  have no control on the number of free variables.    

In the following, we derive Lemma~\ref{lem:bound2} for convex quadratic functions. Property (i) of  Lemma~\ref{lem:bound2} is closely related to the first-order condition~\cite{BV2004} for convex functions. For completeness, we include an adaptation of it to the convex quadratic functions. For simplicity in representation, we assume that the optimal solution is unique without loss of generality. However, the correctness of our proofs does not require the optimal solution to be unique.

\begin{lemma}
	\label{lem:bound2}
	%Let $\mat{x}_r$ be the optimal solution for $f$ constrained by some active set such that $f(\mat{x}_r) > f(\mat{x}_*)$.   
	Let $\mat{n} \in \real^\nu$ be any unit descent direction from $\mat{x}_r$.  Let $\mat{y} = \mat{x}_r + \alpha\mat{n}$ for some $\alpha > 0$ be the feasible point that minimizes $f$ on the ray from $\mat{x}_r$ in direction $\mat{n}$.  Let $S$ be an active set that is disjoint from $\supp(\mat{x}_r) \cup \supp(\mat{n})$.  Let $\mat{x}_{r+1}$ be the optimal solution for $f$ constrained by $S$.  Then, $\frac{f(\mat{x}_{r+1}) - f(\mat{x}_*)}{f(\mat{x}_r) - f(\mat{x}_*)} \leq 1- 
	\frac{1}{2}\cdot\frac{\norm{\mat{x}_r-\mat{y}}}{\norm{\mat{x}_r-\mat{x}_*}} \cdot 
	\frac{\langle \nabla f(\mat{x}_r),\mat{n} \rangle}{\langle \nabla f(\mat{x}_r),\mat{n}_* \rangle}$, where $\mat{n}_*$ is the unit vector $\frac{\mat{x}_*-\mat{x}_r}{\norm{\mat{x}_*-\mat{x}_r}}$.
\end{lemma}
\begin{proof}
	We prove two properties: (i)~$f(\mat{x}_r) - f(\mat{y}) \leq -\norm{\mat{x}_r-\mat{y}} \cdot \langle \nabla f(\mat{x}_r),\mat{n} \rangle$, and (ii)~$f(\mat{x}_r) - f(\mat{y}) \geq -\frac{1}{2}\norm{\mat{x}_r-\mat{y}} \cdot \langle \nabla f(\mat{x}_r),\mat{n} \rangle$.
	
	Consider (i).  For all $s \in [0,1]$, define $\mat{y}_s = \mat{x}_r + s(\mat{y}-\mat{x}_r)$.  By the chain rule, we have $\frac{\partial f}{\partial s} = \bigl\langle \frac{\partial f}{\partial \mat{y}_s},\frac{\partial \mat{y}_s}{\partial s} \bigr\rangle
	= \bigl\langle \nabla f(\mat{y}_s), \, \mat{y} -\mat{x}_r \bigr\rangle$.
	We integrate along a linear movement from $\mat{x}_r$ to $\mat{y}$.  Using the fact that $\nabla f(\mat{y}_s) = 2\mat{A}^t\mat{A}(\mat{x}_r + s(\mat{y}-\mat{x}_r)) + \mat{a} = \nabla f(\mat{x}_r) + 2s\mat{A}^t\mat{A}(\mat{y}-\mat{x}_r)$, we obtain
	\begin{eqnarray*}
		f(\mat{y}) & = & f(\mat{x}_r) + \int_{0}^1 \langle \nabla f(\mat{y}_s),  \mat{y} -\mat{x}_r \rangle \, \mathtt{d}s \nonumber \\
		& = & f(\mat{x}_r) + \int_0^1 \langle \nabla f(\mat{x}_r), \, \mat{y}-\mat{x}_r \rangle \, \mathtt{d}s + 
		\int_0^1 2s\iprod{\mat{A}^t\mat{A}(\mat{y}-\mat{x}_r), \, \mat{y}-\mat{x}_r} \, \mathtt{d}s \nonumber \\
		& = & f(\mat{x}_r) + \Bigl[\langle \nabla f(\mat{x}_r), \, \mat{y}-\mat{x}_r \rangle \cdot s\Bigr]^1_0 + 
		\Bigl[\norm{\mat{A}(\mat{y}-\mat{x}_r)}^2 \cdot s^2 \Bigr]^1_0 \\
		& = & f(\mat{x}_r) + \langle \nabla f(\mat{x}_r), \, \mat{y}-\mat{x}_r \rangle + \norm{\mat{A}(\mat{y}-\mat{x}_r)}^2.  
	\end{eqnarray*}
	It follows immediately that
	\begin{eqnarray}
		f(\mat{x}_r) - f(\mat{y}) & = & -\langle \nabla f(\mat{x}_r), \, \mat{y}-\mat{x}_r \rangle - \norm{\mat{A}(\mat{y}-\mat{x}_r)}^2 \label{eq:0} \\
		& \leq & -\langle \nabla f(\mat{x}_r), \, \mat{y}-\mat{x}_r \rangle  \;\; = \;\; -\norm{\mat{x}_r-\mat{y}} \cdot \langle \nabla f(\mat{x}_r),\mat{n} \rangle. \nonumber
	\end{eqnarray}
	This completes the proof of (i).
	
	Consider (ii).   Using $\mat{y}_s = \mat{x}_r + s(\mat{y}-\mat{x}_r) = \mat{y} + (1-s)(\mat{x}_r - \mat{y})$ and the fact that $\nabla f(\mat{y}_s) = 2\mat{A}^t\mat{A}(\mat{y} + (1-s)(\mat{x}_r-\mat{y})) + \mat{a} = \nabla f(\mat{y}) + 2(1-s)\mat{A}^t\mat{A}(\mat{x}_r-\mat{y})$,  we carry out a symmetric derivation:
	\begin{eqnarray}
		f(\mat{x}_r) & = & f(\mat{y}) + \int_1^0\langle \nabla f(\mat{y}_s),  \mat{y} -\mat{x}_r \rangle\, \mathtt{d}s \nonumber \\
		& = & f(\mat{y}) + \int_1^0 \langle \nabla f(\mat{y}), \, \mat{y}-\mat{x}_r \rangle \, \mathtt{d}s + 
		\int_1^0  2(1-s)\langle \mat{A}^t\mat{A}(\mat{x}_r-\mat{y}), \, \mat{y} - \mat{x}_r \rangle \, \mathtt{d}s \nonumber \\
		& = & f(\mat{y}) + \Bigl[\langle \nabla f(\mat{y}), \, \mat{y}-\mat{x}_r \rangle \cdot s \Bigr]^0_1 + 
		\Bigl[ \norm{\mat{A}(\mat{y}-\mat{x}_r)}^2 \cdot (1-s)^2\Bigr]^0_1 \nonumber \\
		& = & f(\mat{y}) - \langle \nabla f(\mat{y}), \, \mat{y} - \mat{x}_r \rangle + \norm{\mat{A}(\mat{y}-\mat{x}_r)}^2.   \label{eq:0-1}
	\end{eqnarray}
	Combining \eqref{eq:0} and \eqref{eq:0-1} gives $f(\mat{x}_r) - f(\mat{y}) = -\frac{1}{2}\langle \nabla f(\mat{x}_r),\mat{y}-\mat{x}_r \rangle - \frac{1}{2}\langle \nabla f(\mat{y}),\mat{y}-\mat{x}_r \rangle$.  As $\mat{y}$ is the feasible point that minimizes the value of $f$ on the ray from $\mat{x}_r$ in direction $\mat{n}$, we have $\langle \nabla f(\mat{y}),\mat{y}-\mat{x}_r \rangle \leq 0$.  Therefore, $f(\mat{x}_r) - f(\mat{y}) \geq  -\frac{1}{2}\langle \nabla f(\mat{x}_r),\mat{y}-\mat{x}_r \rangle = -\frac{1}{2}\norm{\mat{x}_r-\mat{y}} \cdot \langle \nabla f(\mat{x}_r),\mat{n} \rangle$, completing the proof of~(ii).
	
	By applying (i) to the descent direction $\mat{x}_*-\mat{x}_r$, we get $f(\mat{x}_r) - f(\mat{x}_*) \leq  -\langle \nabla f(\mat{x}_r),\mat{x}_*-\mat{x}_r \rangle$.   Since $S$ is disjoint from $\supp(\mat{x}_*) \cup \supp(\mat{n})$, the feasible points on the ray from $\mat{x}_r$ in direction $\mat{n}$ are also feasible with respect to $S$.  Therefore, $f(\mat{x}_r) - f(\mat{x}_{r+1}) \geq f(\mat{x}_r) - f(\mat{y})$.  By applying (ii) to the descent direction $\mat{n}$, we get $f(\mat{x}_r) - f(\mat{x}_{r+1}) \geq f(\mat{x}_r) - f(\mat{y}) \geq -\frac{1}{2}\norm{\mat{x}_r-\mat{y}} \cdot \langle \nabla f(\mat{x}_r),\mat{n} \rangle$.  Dividing the lower  bound for $f(\mat{x}_r) - f(\mat{x}_{r+1})$ by the upper bound for $f(\mat{x}_r) - f(\mat{x}_*)$ proves that $\frac{f(\mat{x}_{r}) - f(\mat{x}_{r+1})}{f(\mat{x}_r) - f(\mat{x}_*)} \geq \frac{1}{2}\cdot\frac{\norm{\mat{x}_r-\mat{y}}}{\norm{\mat{x}_r-\mat{x}_*}} \cdot 
	\frac{\langle \nabla f(\mat{x}_r),\mat{n} \rangle}{\langle \nabla f(\mat{x}_r),\mat{n}_* \rangle}$.  Hence, $\frac{f(\mat{x}_{r+1}) - f(\mat{x}_*)}{f(\mat{x}_r) - f(\mat{x}_*)} = 1 - \frac{f(\mat{x}_{r}) - f(\mat{x}_{r+1})}{f(\mat{x}_r) - f(\mat{x}_*)} \leq 1 - \frac{1}{2}\cdot\frac{\norm{\mat{x}_r-\mat{y}}}{\norm{\mat{x}_r-\mat{x}_*}} \cdot 
	\frac{\langle \nabla f(\mat{x}_r),\mat{n} \rangle}{\langle \nabla f(\mat{x}_r),\mat{n}_* \rangle}$. 
\end{proof}

\cancel{
	\setlength\extrarowheight{3pt}
	\begin{table}
		\centering
		\begin{tabular}{|c|c||c|c|}
			\hline
			& Upper bound of $\frac{f(\mat{x}_{r+1})-f(\mat{x}_*)}{f(\mat{x}_r)-f(\mat{x}_*)}$ 
			&& Upper bound of  $\frac{f(\mat{x}_{r+1})-f(\mat{x}_*)}{f(\mat{x}_r)-f(\mat{x}_*)}$  \\ [3pt]
			\hline 
			DKSG & $ 1 -  \frac{1}{6\sqrt{2}n} \cdot \frac{\norm{\mat{x}_r - \mat{y}_r}}{\norm{\mat{x}_r - \mat{x}_*}}$ &
			NNLS & $ 1 - \frac{1}{2\sqrt{2n\ln n}} \cdot \frac{\norm{\mat{x}_r - \mat{y}_r}}{\norm{\mat{x}_r - \mat{x}_*}}$ \\ [3pt]
			\hline
			ZHLG & $ 1 - \frac{1}{2\sqrt{2}n\sqrt{\ln n}} \cdot \frac{\norm{\mat{x}_r - \mat{y}_r}}{\norm{\mat{x}_r - \mat{x}_*}}$ &
			PD & $1 - \frac{1}{2\sqrt{2(m+n)}} \cdot \frac{\norm{\mat{x}_r - \mat{y}_r}}{\norm{\mat{x}_r - \mat{x}_*}}$ \\ [3pt]
			\hline
			MEB & $1 - \frac{1}{2\sqrt{2n}} \cdot \frac{\norm{\mat{x}_r - \mat{y}_r}}{\norm{\mat{x}_r - \mat{x}_*}}$ & & \\ [3pt]
			\hline
		\end{tabular}
		\caption{Upper bounds of $\frac{f(\mat{x}_{r+1})-f(\mat{x}_*)}{f(\mat{x}_r)-f(\mat{x}_*)}$ for DKSG, ZHLG, MEB, NNLS, and PD.}
		\label{tb:factor}
	\end{table}
	
	\setlength\extrarowheight{0pt}
}

\begin{lemma}
	\label{lem:SolveNNQ}
	If $\beta_1$ is exceeded, the lines~\ref{alg:loop-start}--\ref{alg:loop-end} of {\em SolveNNQ} will be %lines~\ref{alg:loop-start}--\ref{alg:loop-end}
	 iterated at most $|\supp(\mat{x}_*)|$ more times before $\mat{x}_*$ is found.
\end{lemma}
\begin{proof}
	By the KKT conditions, $\nabla f(\mat{x}_r) = \mat{B}^t\mat{u}_r + \mat{C}^t\tilde{\mat{u}}_r + \mat{v}_r$, so $\langle \nabla f(\mat{x}_r),\mat{x}_* - \mat{x}_r \rangle = \langle \mat{B}^t\mat{u}_r, \mat{x}_* - \mat{x}_r \rangle + \langle \mat{C}^t\tilde{\mat{u}}_r, \mat{x}_* - \mat{x}_r \rangle + \langle \mat{v}_r,\mat{x}_* - \mat{x}_r \rangle$.   By the complementary slackness, we have $\mat{u}_r^t(\mat{Bx}_r^{}-\mat{b}) = 0$ and $\mat{v}_r^t\mat{x}_r^{} = 0$.  The former equation implies that $\mat{u}_r^t\mat{Bx}_r^{} = \mat{u}^t_r\mat{b}$.  By primal feasibility, $\mat{Cx}_* = \mat{Cx}_r = \mat{c}$ which implies that $\langle \mat{C}^t\tilde{\mat{u}}_r, \mat{x}_* - \mat{x}_r \rangle = \tilde{\mat{u}}_r^t (\mat{Cx}_* - \mat{Cx}_r) = 0$.  As a result,
	\begin{equation}
		\langle \nabla f(\mat{x}_r),\mat{x}_* - \mat{x}_r \rangle = \mat{u}_r^t(\mat{B}\mat{x}_* - \mat{b}) + \mat{v}_r^t\mat{x}_*.
		\label{eq:-1-2}
	\end{equation}
	By primal and dual feasibilities, $\mat{Bx}_* \geq \mat{b}$ and $\mat{u}_r \geq 0_\kappa$.  So $\mat{u}_r^t(\mat{Bx}_* - \mat{b}) \geq 0$.   
	
	We claim that $(\mat{v}_r)_i < 0$ for some $i \in \supp(\mat{x}_*) \cap S_r$.  Assume to the contrary that $(\mat{v}_r)_i \geq 0$ for all $i \in \supp(\mat{x}_*) \cap S_r$.  Then, $(\mat{v}_r)_i \geq 0$ for all $i \in \supp(\mat{x}_*)$, which implies that $\mat{v}_r^t\mat{x}_* \geq 0$.  Substituting $\mat{u}_r^t(\mat{Bx}_* - \mat{b}) \geq 0$ and $\mat{v}_r^t\mat{x}_* \geq 0$ into \eqref{eq:-1-2} gives $\iprod{\nabla f(\mat{x}_r), \mat{x}_* - \mat{x}_r} \geq 0$.   However, since $\mat{x}_r$ is not optimal, by the convexity of $f$ and the feasible region of the NNQ problem, $\mat{x}_* - \mat{x}_r$ is a descent direction from $\mat{x}_r$, contradicting the relation $\langle \nabla f(\mat{x}_r), \mat{x}_* - \mat{x}_r \rangle \geq 0$.  This proves our claim.
	
	By our claim, there must exist some index $i \in \supp(\mat{x}_*) \cap S_r$ such that $(\mat{v}_r)_i < 0$.   This index $i$ must be included in $E_r$, which means that $i \not\in S_{r+1}$ because $S_{r+1} = S_r \setminus E_r$.   As a result, each iteration of lines~\ref{alg:exceed-start} and~\ref{alg:exceed-end} sets free at least one more element of $\supp(\mat{x}_*)$.  After at most $|\supp(\mat{x}_*)|$ more iterations, all elements of $\supp(\mat{x}_*)$ are free variables, which means that the solver will find $\mat{x}_*$ in at most $|\supp(\mat{x}_*)|$ more iterations.
\end{proof}

\section{Experimental results}
\label{sec:experiments}

In this section, we present our experimental results on the DKSG, ZHLG, and image deblurring problems.  Our machine configuration is: Intel Core 7-9700K 3.6Hz cpu, 3600 Mhz ram, 8~cores, 8 logical processors.  We use MATLAB version R2020b.  We use the option of {\tt quadprog} that runs the interior-point-convex algorithm.  

As explained in Section~\ref{sec:alg}, we eliminate the variables in the active set in each iteration in order to speed up the solver.  One implementation is to extract the remaining rows and columns of the constraint matrix; however, doing so incurs a large amount of data movements.   Instead, we zero out the rows and columns corresponding to the variables in the active set.    %With the version of {\tt quadprog} that we use, it is unnecessary to set the upper bounds of the non-free variables to zero.  Setting these bounds leads to a very minute difference in performance, if any.  

We determine that $\tau = 4\ln^2 \nu$ is a good setting for the DKSG and ZHLG data sets.  If we increase $\ln^2 \nu$ to $\ln^3\nu$, the number of iterations is substantially reduced, but the running time of each iteration is larger; there is little difference in the overall efficiency.  If we decrease $\ln^2 \nu$ to $\ln \nu$, the number of iterations increases so much that the overall efficiency decreases.   We also set $\tau = 4\ln^2 \nu$ in the experiments with  image deblurring without checking whether this is the best for them.  It helps to verify the robustness of this choice of $\tau$.  

We set $\beta_0$ to be $3\tau$ and $\beta_1$ to be~15.  Since we often observe a geometric reduction in the gap between the current objective function value and the optimum from one iteration to the next, we rarely encounter a case in our experiments that the threshold $\beta_1$ is exceeded.

\subsection{DKSG and ZHLG}

For both the DKSG and ZHLG problems, we ran experiments to compare the performance of SolveNNQ with a single call of {\tt quadprog} on the UCI Machine Learning Repository data sets Iris, HCV, Ionosphere, and AAL~\cite{UCI}.  The data sets IRIS and Ionosphere were also used in the experiments in the paper of Daitch~et~al.~\cite{daitch2009} that proposed the DKSG problem.  For each data set, we sampled uniformly at random a number of rows that represent the number of points $n$ and a number of attributes that represent the number of dimensions $d$.   There are two parameters $\mu$ and $\rho$ in the objective function of the ZHLG problem.  According to the results in~\cite{ZHL2014}, we set $\mu = 16$ and $\rho = 2$.

\begin{table}
	\centering
	{\footnotesize
		\begin{tabular}{|c|c|c|>{\bfseries}c|>{\bfseries}c|c|>{\bfseries}c|>{\bfseries}c|}
			\hline
		 & \multicolumn{7}{c|}{Ionosphere}                                                                 \\ \hline
		&     & \multicolumn{3}{c|}{DKSG}                   & \multicolumn{3}{c|}{ZHLG}                   \\ \hline
		$n$ & $d$ &\tt quad & ours   & nnz  & \tt quad & ours   & nnz  \\ \hline		
		70  & 2   & 4.57s                       & 1.84s  & 44\% & 2.51s                       & 1.06s  & 47\% \\
		    & 4   & 5.78s                       & 2.25s  & 47\% & 2.20s                       & 0.97s  & 49\% \\
		& 10  & 4.65s                       & 1.55s  & 47\% & 1.86s                       & 1.67s  & 59\% \\ \hline
		100 & 2   & 26.67s                      & 5.41s  & 33\% & 16.31s                      & 3.35s  & 36\% \\
		& 4   & 43.60s                      & 7.06s  & 35\% & 14.87s                      & 2.77s  & 36\% \\
		& 10  & 28.53s                      & 5.46s  & 35\% & 13.34s                      & 5.42s  & 44\% \\ \hline
		130 & 2   & 192.41s                     & 11.56s & 27\% & 144.27s                     & 5.53s  & 28\% \\
		& 4   & 261.54s                     & 13.13s & 27\% & 141.70s                     & 6.23s  & 29\% \\
		& 10  & 208.24s                     & 12.02s & 27\% & 134.95s                     & 12.43s & 36\% \\ \hline
		160 & 2   & 660.25s                     & 25.39s & 23\% & 614.71s                     & 9.14s  & 24\% \\
		& 4   & 888.06s                     & 27.20s & 23\% & 622.58s                     & 10.74s & 24\% \\
		& 10  & 979.92s                     & 22.03s & 22\% & 598.55s                     & 29.97s & 31\% \\ \hline	
	    \end{tabular}
		}
		\caption{Running time comparison of the Ionosphere data set.  The data in each row is the average of the corresponding data
		from five runs; nnz is the percentage of the average number of non-zeros in the final solution.  We use {\tt quad} to denote {\tt quadprog}.}
	\label{tb:1a}

\end{table}

\begin{table}
	\centering
	{\footnotesize
		\begin{tabular}{|c|c|c|>{\bfseries}c|>{\bfseries}c|c|>{\bfseries}c|>{\bfseries}c|}
			\hline
			    & \multicolumn{7}{c|}{AAL}                                                                        \\ \hline
			&     & \multicolumn{3}{c|}{DKSG}                   & \multicolumn{3}{c|}{ZHLG}                   \\ \hline
			$n$ & $d$ & \tt quad & ours   & nnz  & \tt quad & ours   & nnz  \\ \hline
			70  & 2   & 3.99s                       & 1.72s  & 44\% & 2.36s                       & 1.19s  & 52\% \\
			& 4   & 4.64s                       & 1.94s  & 46\% & 2.32s                       & 1.27s  & 50\% \\
			& 10  & 4.89s                       & 1.80s  & 48\% & 1.93s                       & 1.69s  & 58\% \\ \hline
			100 & 2   & 22.91s                      & 5.30s  & 34\% & 15.25s                      & 3.24s  & 37\% \\
			& 4   & 38.63s                      & 7.81s  & 37\% & 13.38s                      & 2.68s  & 35\% \\
			& 10  & 36.58s                      & 6.00s  & 35\% & 13.34s                      & 4.59s  & 43\% \\ \hline
			130 & 2   & 175.43s                     & 15.81s & 27\% & 145.83s                     & 6.58s  & 28\% \\
			& 4   & 221.80s                     & 15.89s & 29\% & 142.01s                     & 5.88s  & 30\% \\
			& 10  & 257.41s                     & 14.60s & 28\% & 135.08s                     & 10.03s & 35\% \\ \hline
			160 & 2   & 550.44s                     & 26.06s & 24\% & 613.45s                     & 15.01s & 26\% \\
			& 4   & 864.51s                     & 31.71s & 24\% & 610.71s                     & 14.62s & 26\% \\
			& 10  & 984.85s                     & 20.86s & 22\% & 601.50s                     & 22.42s & 30\% \\ \hline
		\end{tabular}
	}
	\caption{Running time comparison of the AAL data set.  The data in each row is the average of the corresponding data
		from five runs; nnz is the percentage of the average number of non-zeros in the final solution.  We use {\tt quad} to denote {\tt quadprog}.}
	\label{tb:1b}
\end{table}

\begin{table}
	\centering
	{\footnotesize
		\begin{tabular}{|c|c|c|>{\bfseries}c|>{\bfseries}c|c|>{\bfseries}c|>{\bfseries}c|}
			\hline
			& \multicolumn{7}{c|}{IRIS}                                                                      \\ \hline
			&     & \multicolumn{3}{c|}{DKSG}                   & \multicolumn{3}{c|}{ZHLG}                  \\ \hline
			$n$      & $d$ & \tt quad & ours   & nnz  & \tt quad & ours  & nnz  \\ \hline
			70       & 2   & 5.61s                       & 2.00s  & 43\% & 2.76s                       & 1.21s & 47\% \\
			& 3   & 5.66s                       & 2.09s  & 44\% & 2.61s                       & 1.15s & 48\% \\
			& 4   & 4.88s                       & 2.10s  & 45\% & 2.23s                       & 1.18s & 50\% \\ \hline
			100      & 2   & 43.87s                      & 6.55s  & 33\% & 17.63s                      & 2.13s & 34\% \\
			& 3   & 37.56s                      & 5.92s  & 33\% & 15.92s                      & 2.64s & 36\% \\
			& 4   & 42.69s                      & 5.95s  & 33\% & 15.24s                      & 2.95s & 37\% \\ \hline
			130      & 2   & 241.59s                     & 13.29s & 26\% & 148.34s                     & 5.62s & 29\% \\
			& 3   & 199.22s                     & 12.20s & 27\% & 150.54s                     & 5.34s & 28\% \\
			& 4   & 256.28s                     & 13.96s & 27\% & 151.87s                     & 6.20s & 30\% \\ \hline
			150 & 2   & 504.21s                     & 20.15s & 23\% & 414.79s                     & 8.38s & 25\% \\
			  & 3   & 510.55s                     & 19.81s & 24\% & 398.06s                     & 7.80s & 26\% \\
			& 4   & 565.89s                     & 21.04s & 24\% & 390.47s                     & 7.36s & 26\% \\ \hline
		\end{tabular}
	}
	\caption{Running time comparison of the IRIS data set.  The data in each row is the average of the corresponding data
		from five runs; nnz is the percentage of the average number of non-zeros in the final solution.  We use {\tt quad} to denote {\tt quadprog}.}
	\label{tb:2a}
\end{table}

\begin{table}
	\centering
	{\footnotesize
		\begin{tabular}{|c|c|c|>{\bfseries}c|>{\bfseries}c|c|>{\bfseries}c|>{\bfseries}c|}
			\hline
			  & \multicolumn{7}{c|}{HCV}                                                                        \\ \hline
			&     & \multicolumn{3}{c|}{DKSG}                   & \multicolumn{3}{c|}{ZHLG}                   \\ \hline
			$n$      & $d$ & \tt quad & ours   & nnz  & \tt quad & ours   & nnz  \\ \hline
			70       & 2   & 3.68s                       & 1.16s  & 43\% & 2.65s                       & 0.84s  & 45\% \\
			& 4   & 5.22s                       & 1.71s  & 44\% & 2.13s                       & 0.88s  & 45\% \\
			& 10  & 4.55s                       & 1.35s  & 43\% & 1.88s                       & 2.06s  & 62\% \\ \hline
			100      & 2   & 13.02s                      & 3.50s  & 32\% & 15.20s                      & 2.00s  & 32\% \\
			& 4   & 39.50s                      & 4.84s  & 33\% & 13.03s                      & 2.35s  & 32\% \\
			& 10  & 29.79s                      & 3.65s  & 31\% & 12.78s                      & 8.32s  & 48\% \\ \hline
			130      & 2   & 108.10s                     & 6.90s  & 28\% & 141.20s                     & 3.07s  & 24\% \\
			& 4   & 180.42s                     & 12.05s & 27\% & 138.05s                     & 3.78s  & 24\% \\
			& 10  & 184.51s                     & 8.16s  & 24\% & 133.91s                     & 21.17s & 38\% \\ \hline
			160 & 2   & 551.43s                     & 17.32s & 24\% & 613.29s                     & 5.22s  & 20\% \\
			  & 4   & 868.60s                     & 23.93s & 22\% & 613.04s                     & 8.10s  & 21\% \\
			& 10  & 633.97s                     & 14.48s & 20\% & 588.64s                     & 45.33s & 32\% \\ \hline
		\end{tabular}
	}
	\caption{Running time comparison of the HCV data set.  The data in each row is the average of the corresponding data
		from five runs; nnz is the percentage of the average number of non-zeros in the final solution.  We use {\tt quad} to denote {\tt quadprog}.}
	\label{tb:2b}
\end{table}

Tables~\ref{tb:1a},~\ref{tb:1b},~\ref{tb:2a}, and~\ref{tb:2b} show some average running times of SolveNNQ and {\tt quadprog}.   When $n \geq 130$ and $n \geq 13d$, SolveNNQ is often 10 times or more faster.  For the same dimension $d$, the column nnz shows that the percentage of non-zero solution coordinates decreases as $n$ increases, that is, the sparsity increases as $n$ increases.  Correspondingly, the speedup of SolveNNQ over a single call of {\tt quadprog} also increases as reflected by the running times.

\begin{figure}
	\centering
	\begin{tabular}{ccc}
		\includegraphics[scale=0.5]{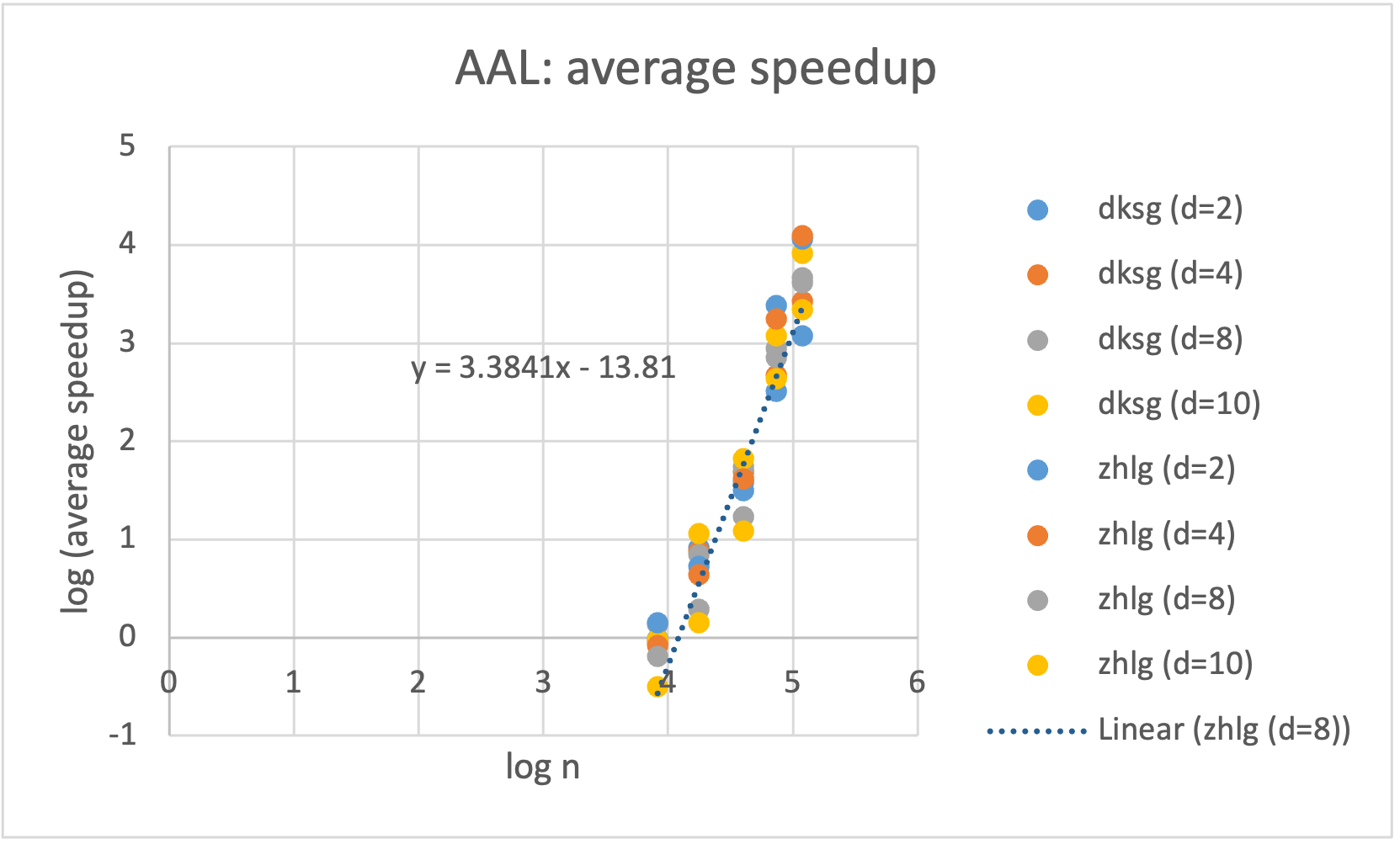} & & 
		\includegraphics[scale=0.5]{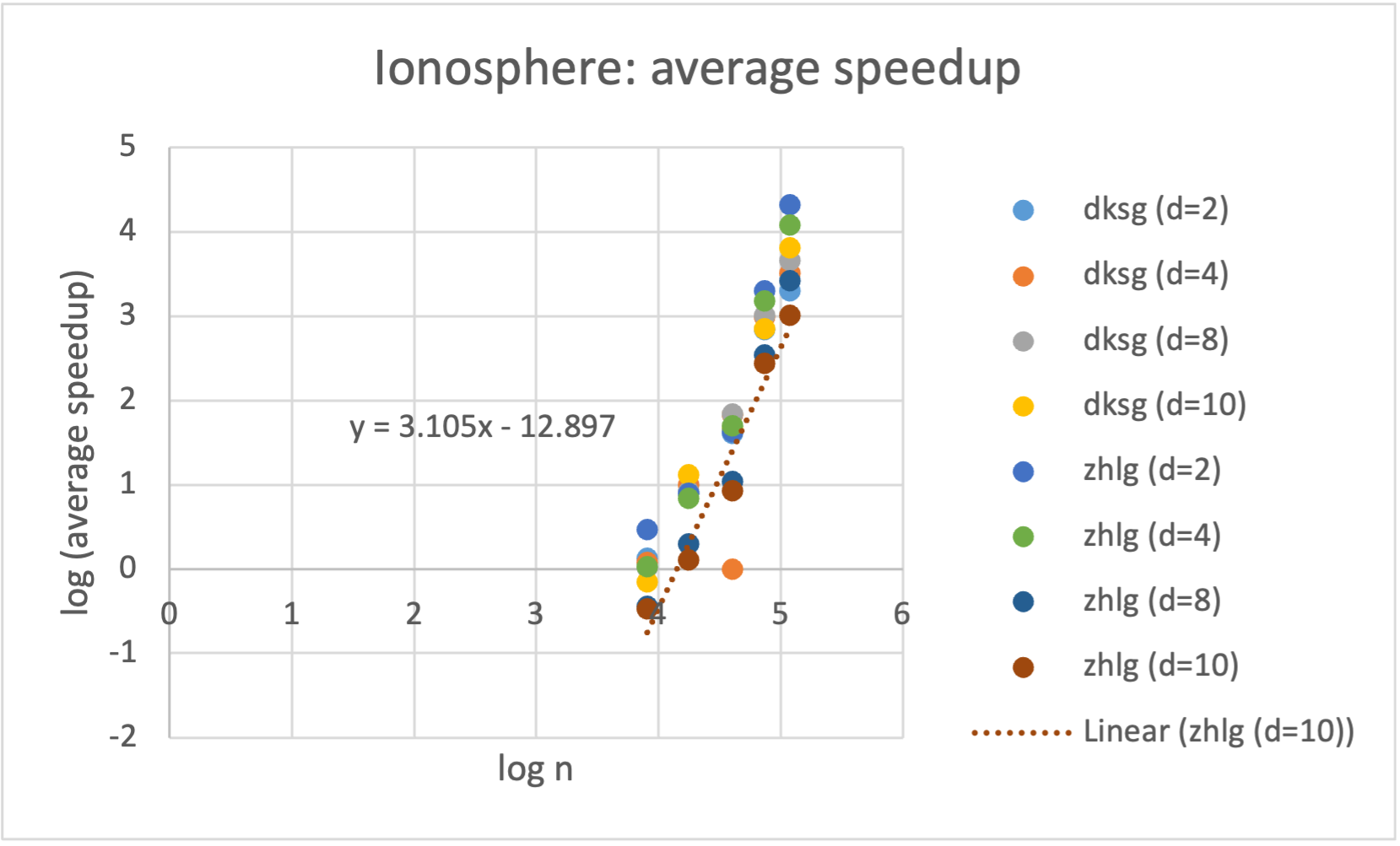} \\
		\includegraphics[scale=0.5]{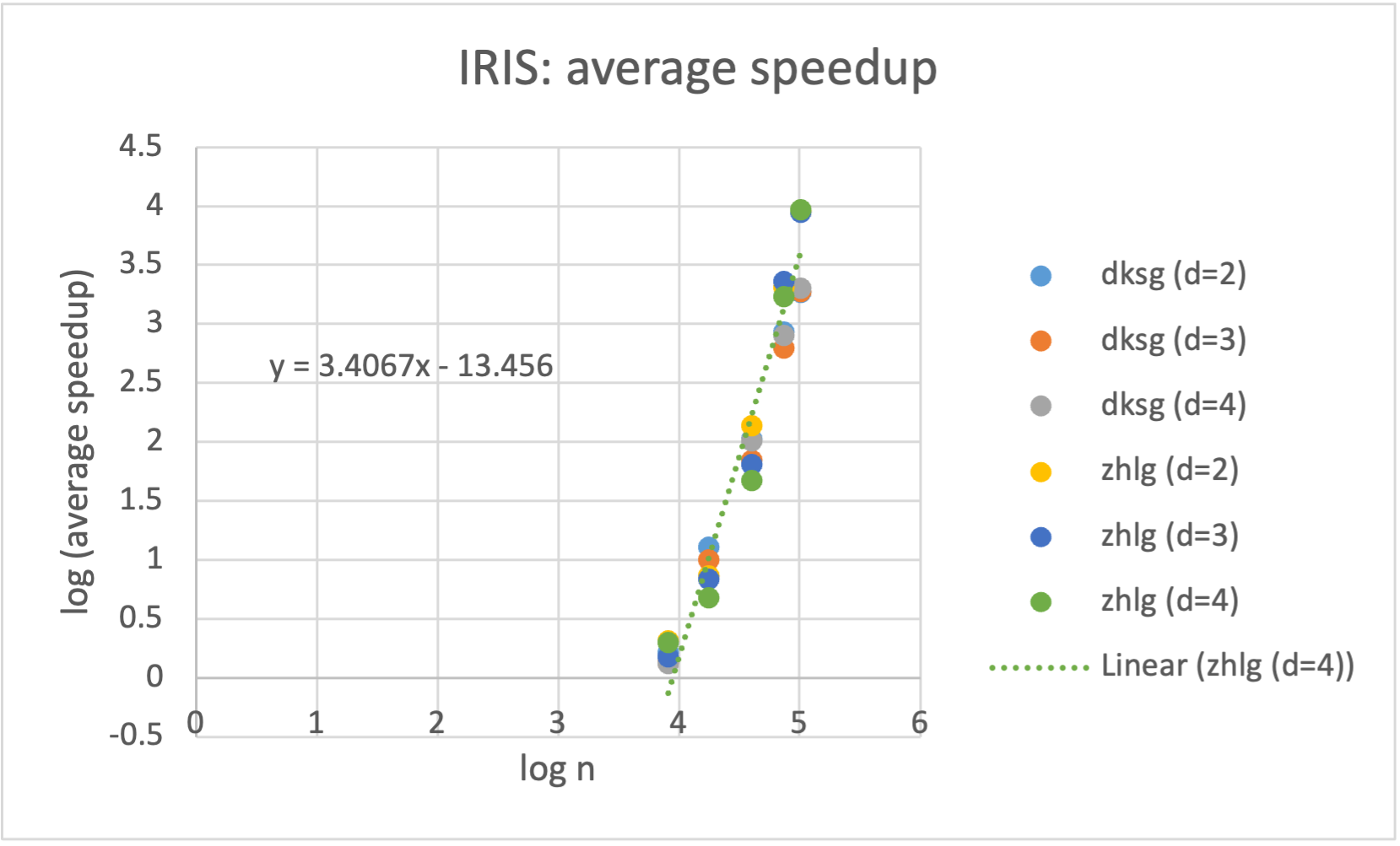} & & 
		\includegraphics[scale=0.5]{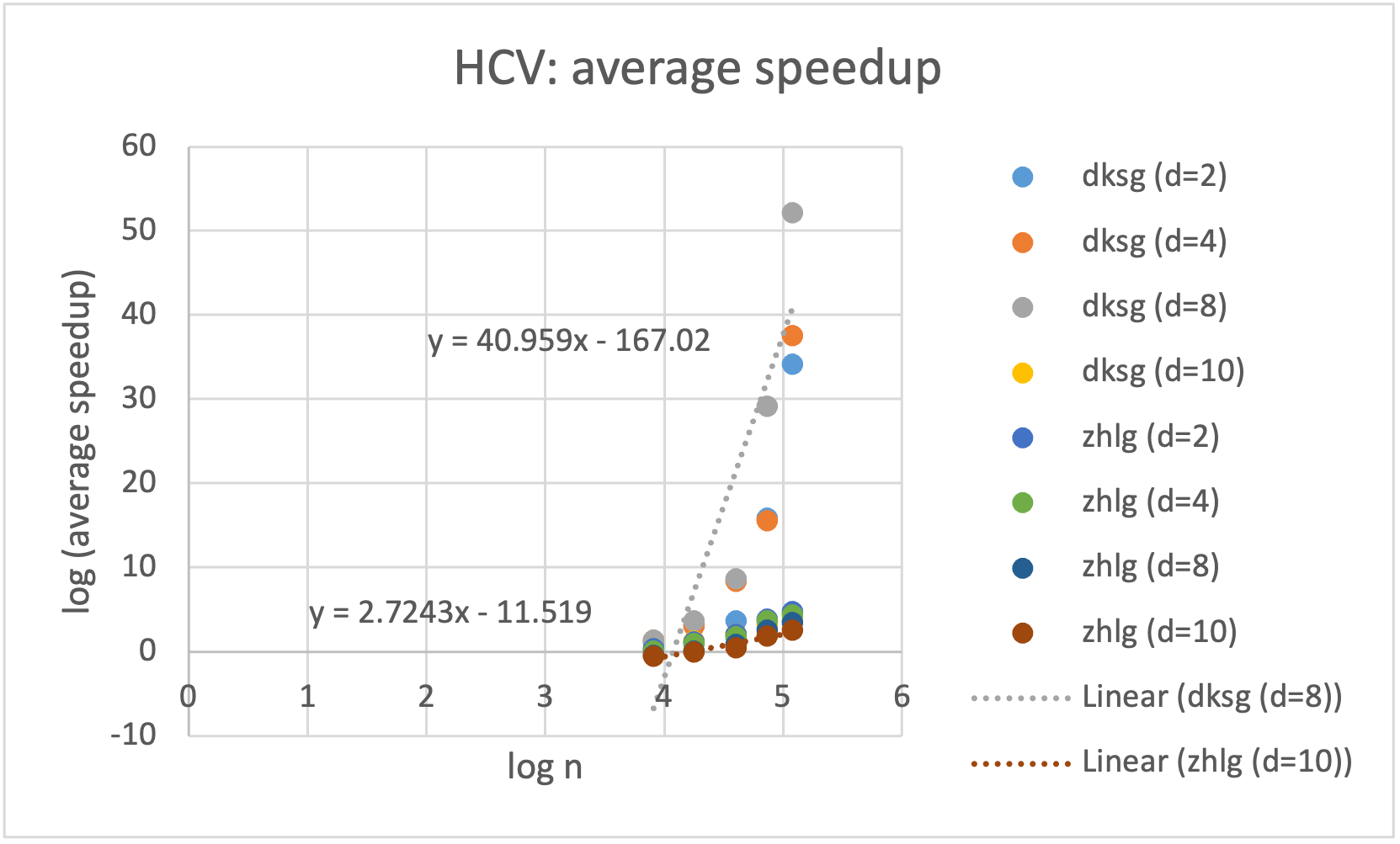}
	\end{tabular} 
	\caption{We plot the natural logarithm of the average speedup of SolveNNQ over a single call of {\sf quadprog} for DKSG and ZHLG.   The horizontal axis is $\ln n$.   A distinct color is used for data points for the same dimension $d$.}
	\label{fg:prox-0}
\end{figure}

Figure~\ref{fg:prox-0} shows four plots, one for each of the four data sets, that show how the natural logarithm of the average speedup of SolveNNQ over a single call of {\tt quadprog} increases as the value of $\ln n$ increases.  In each plot, data points for distinct dimensions are shown in distinct colors, and the legends give the color coding of the dimensions.  For both DKSG and ZHLG and for every fixed $d$, the average speedup as a function of $n$ hovers around $\Theta(n^3)$.  The average speedup data points for HCV shows a greater diversity, but the smallest speedup is still roughly proportional to $n^{2.7}$.  As a result, SolveNNQ gives a much superior performance than a single call of {\tt quadprog} even for moderate values of $n$.

To verify our assumption that the distance between  $\mat{x}_r$ and the minimum $\mat{y}_r$ in direction $\mat{n}_r$ is at least $\frac{1}{\lambda}\norm{\mat{x}_r-\mat{x}_*}$, we used box-and-whisker plots for visualizations. A box and whisker plot displays the distribution of a dataset. It shows the middle 50\% of the data as a box, with the median shown using a horizontal segment inside the box; whiskers are used to show the spread of the remaining data. We exclude two outliers of the ratios for the DKSG problem. These outliers occur when $r=1$, so the ratios $\frac{\norm{\mat{x}_1-\mat{x}_{2}}}{\norm{\mat{x}_1-\mat{x}_*}}$ describe the behavior between the first and second iterations of the algorithm. These outliers can be ignored, as the initialization does not induce any part of the active set from the gradient. The active set in the first iteration is arbitrary and distinct from the rest of the algorithm. As a result, any abnormality involving the first iteration does not impact the convergence of the algorithm, since it is just one of many iterations. As shown in Figures~\ref{fg:Rs}(a) and (b), the ratio $\frac{\norm{\mat{x}_r-\mat{x}_{r+1}}}{\norm{\mat{x}_r-\mat{x}_*}}$ is greater than $\frac{1}{100}$. Thus, our algorithm takes a descent direction that meets our assumption that $\frac{\norm{\mat{x}_r-\mat{y}}}{\norm{\mat{x}_r-\mat{x}_*}}$ is larger than or equal to a not so small constant of $\frac{1}{100}$. The observed convergence rate is even better.  As illustrated in Figure~\ref{fg:cosines-0}, the gap $f(\mat{x}_r) - f(\mat{x}_*)$ is reduced by a factor in the range $[1.5,10]$ in every iteration on average; in fact, most runs converge more rapidly than a geometric convergence.

\begin{figure}
	\centering
	\begin{tabular}{ccc}
		\includegraphics[scale=0.21]{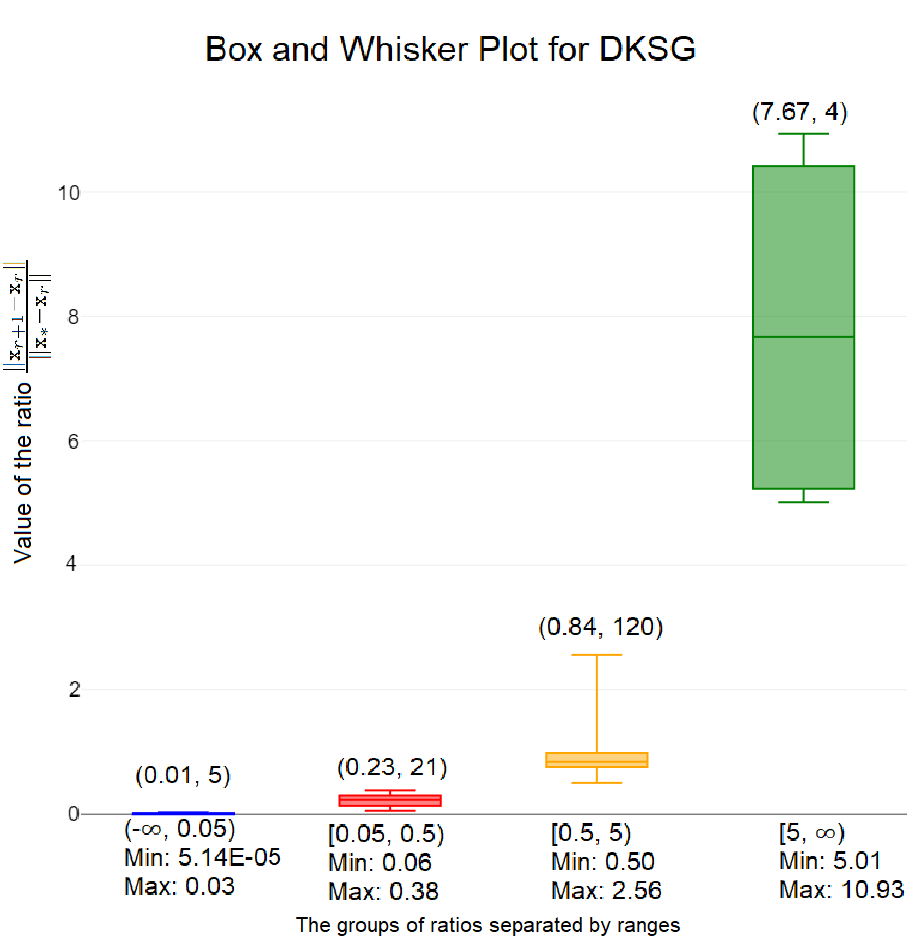} & \hspace{1.5cm} &
		\includegraphics[scale=0.21]{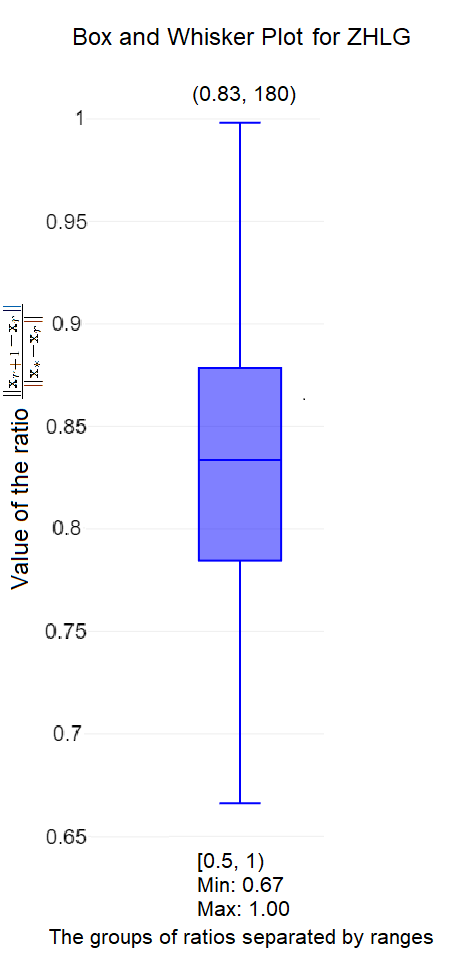} \\
		(a) &   &(b) \\ 
		
	\end{tabular}
	\caption{(a)~Plot of all ratios $\frac{\norm{\mat{x}_{r+1}-\mat{x}_{r}}}{\norm{\mat{x}_* - \mat{x}_r}}$ for DKSG. The ratios  $\frac{\norm{\mat{x}_{r+1}-\mat{x}_{r}}}{\norm{\mat{x}_* - \mat{x}_r}}$ are grouped into disjoint ranges, and these disjoint ranges are listed along the x-axis. The (median, cardinality) tuple of each group is shown above each box.  (b)~Plot of all ratios $\frac{\norm{\mat{x}_{r+1}-\mat{x}_{r}}}{\norm{\mat{x}_* - \mat{x}_r}}$ for ZHLG.}
	\label{fg:Rs}
\end{figure}

\begin{figure}
	\centering
	\begin{tabular}{cc}
		
		\includegraphics[scale=.5]{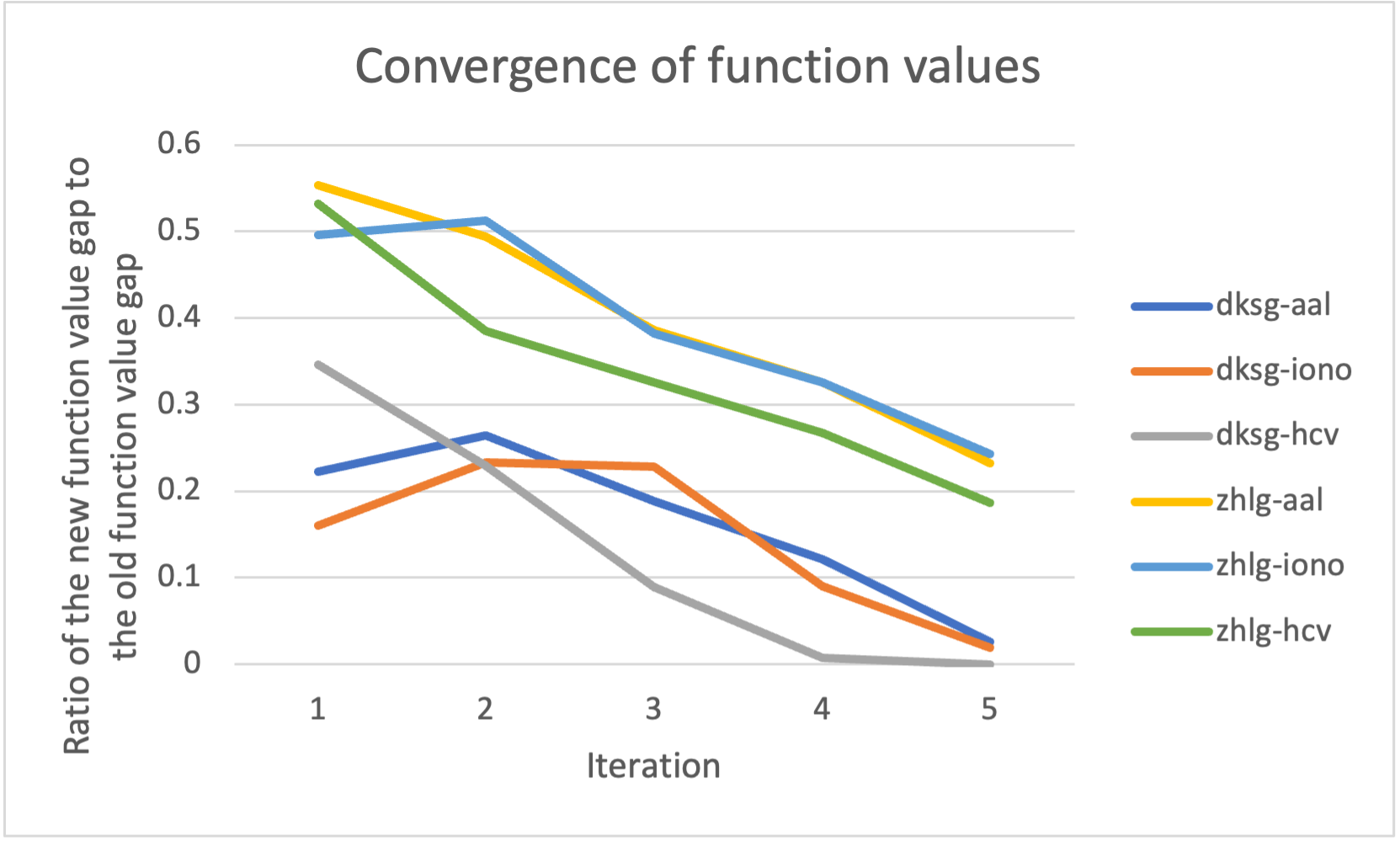} \\
		
	\end{tabular}
	\caption{Each data point is the median of ten runs. Plot of the median $\frac{f(\mat{x}_{r+1})-f(\mat{x}_*)}{f(\mat{x}_r)-f(\mat{x}_*)}$ against $r$.}
	\label{fg:cosines-0}
\end{figure}

\begin{figure}
	\centering
	\begin{tabular}{cccc}
		\includegraphics[scale=0.75]{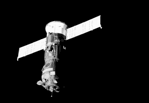} & 
		\includegraphics[scale=0.75]{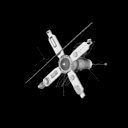} & 
		\includegraphics[scale=0.75]{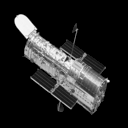} \\
		nph & sts & hst
	\end{tabular}
	\caption{Three space images.}
	\label{fg:app-space}
\end{figure}

\subsection{Image deblurring}

We experimented with four space images nph, sts, and hst from~\cite{IRTDATA, nph}.  We down-sample and sparsify by setting nearly black pixels (with intensities such as 10 or below, or 20 or below) to black.  Figure~\ref{fg:app-space} shows the four resulting space images: nph is $149\times 103$, and both sts and hst are $128\times 128$.   These images are sparse (i.e., many black pixels); therefore, they are good for the working of SolveNNQ.
%We generate blurred images using different values of~$\sigma$.

For space images, it is most relevant to use blurring due to atmospheric turbulence.  In the following, we describe a point spread function in the literature that produces such blurring effects~\cite{M67}.  Let $I$ denote a two-dimensional $r_0 \times c_0$ pixel array (an image).  We call the pixel at the $a$-th row and $b$-th column the $(a,b)$-pixel.  Let $I_{a,b}^*$ denote the intensity of the $(a,b)$-pixel in the original image.  Let $I'_{a,b}$ denote its intensity after blurring.  

The atmospheric turbulence point spread function determines how to distribute the intensity $I^*_{a,b}$ to other pixels.  At the same time, the $(a,b)$-pixel receives contributions from other pixels.  The sum of these contributions and the remaining intensity at the $(a,b)$-pixel gives $I'_{a,b}$.   The weight factor with respect to an $(a,b)$-pixel and a $(c,d)$-pixel is $K \cdot \mathrm{exp}\bigl(-\frac{(a-c)^2 + (b-d)^2}{2\sigma^2}\bigr)$ for some positive constant $K$.  It means that the contribution of the $(a,b)$-pixel to $I'_{c,d}$ is $K\cdot \mathrm{exp}\bigl(-\frac{(a-c)^2 + (b-d)^2}{2\sigma^2}\bigr) \cdot I^*_{a,b}$.  We assume the point spread function is spatially invariant which makes the relation symmetric; therefore, the contribution of the $(c,d)$-pixel to $I'_{a,b}$ is $K\cdot \mathrm{exp}\bigl(-\frac{(a-c)^2 + (b-d)^2}{2\sigma^2}\bigr) \cdot I^*_{c,d}$.  

The atmospheric turbulence point spread function needs to be truncated as its support is infinite.  Let $B$ be a $(2\sigma+1) \times (2\sigma+1)$ square centered at  the $(a,b)$-pixel.  The truncation is done so that the $(a,b)$-pixel has no contribution outside $B$, and no pixel outside $B$ contributes to the $(a,b)$-pixel.  For any $s,t \in [-\sigma,\sigma]$, the weight of the entry of $B$ at the $(a+s)$-th row and the $(b+t)$-th column is $K\cdot \mathrm{exp}\bigl(-\frac{s^2 + t^2}{2\sigma^2}\bigr)$.  We divide every entry of $B$ by the total weight of all entries in $B$ because the contributions of the $(a,b)$-pixel to other pixels should sum to 1.   The coefficient $K$ no longer appears in the normalized weights in $B$, so we do not need to worry about how to set $K$.

How do we produce the blurring effects as a multiplication of a matrix and a vector?   We first transpose the rows of $I$ to columns.  Then, we stack these columns in the row order in $I$ to form a long vector $\mat{x}$.  That is, the first row of $I$ becomes the top $c_0$ entries of $\mat{x}$ and so on.   Let $M_{a,b}$ be an array with the same row and column dimensions as $I$.  Assume for now that the square $B$ centered at the $(a,b)$-entry of $M_{a,b}$ is fully contained in $M_{a,b}$.  The entries of $M_{a,b}$ are zero if they are outside $B$; otherwise, the entries of $M_{a,b}$ are equal to the normalized weights at the corresponding entries in $B$.  Then, we concatenate the rows of $M_{a,b}$ in the row order to form a long row vector; this long row vector is the $((a-1)c_0 + b)$-th row of a matrix $\mat{A}$.  Repeating this for all $M_{a,b}$ defines the matrix $\mat{A}$.  The product of the $((a-1)c_0 + b)$-th row of $\mat{A}$ and $\mat{x}$ is a scalar, which is exactly the sum of the contributions of all pixels at the $(a,b)$-pixel, i.e., the intensity $I'_{a,b}$.  Consider the case that some entries of $B$ lie outside $M_{a,b}$.  If an entry of $B$ is outside $M_{a,b}$, we find the vertically/horizontally nearest boundary entry of $M_{a,b}$ and add to it the normalized weight at that ``outside'' entry of $B$.  This ensures that the total weight within $M_{a,b}$ still sums to 1, thus making sure that no pixel intensity is lost.  Afterwards, we linearize $M_{a,b}$ to form a row of $\mat{A}$ as before.  In all, $\mat{Ax}$ is the blurred image by the atmospheric turbulence point spread function.

Working backward, we solve an NNLS problem to deblur the image: find the non-negative $\mat{x}$ that minimizes $\norm{\mat{Ax}-\mat{b}}^2$, where $\mat{b}$ is the vector obtained by transposing the rows of the blurred image into columns, followed by stacking these columns in the row order in the blurred image.

In running SolveNNQ on the space images, we compute the gradient vector at the zero vector and select the $20\tau$ most negative coordinates.  (Recall that $\tau = 4\ln^2 \nu$ and $\nu$ is the total number of pixels in an image in this case).  We call {\tt quadprog} with these $20\tau$ most negative coordinates as the only free variables and obtain the initial solution $\mat{x}_1$.  Afterwards, SolveNNQ iterates as before.

\begin{figure}
	\centering
	{\tiny
		\begin{tabular}{cccc}
			\includegraphics[scale=0.6]{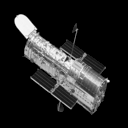} & %.75	
			\includegraphics[scale=0.6]{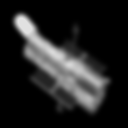} & 
			\includegraphics[scale=0.6]{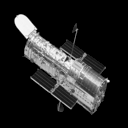} &
			\includegraphics[scale=0.6]{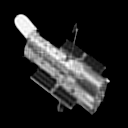} \\
			hst, original & hst, $\sigma=2.5$ & SolveNNQ,  0.002, 21.6s & FISTABT, 0.1, 11.5s \\
			\includegraphics[scale=0.6]{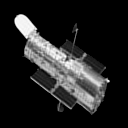} &
			\includegraphics[scale=0.6]{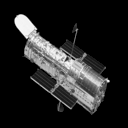} &
			\includegraphics[scale=0.6]{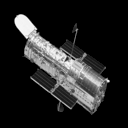} &
			\includegraphics[scale=0.6]{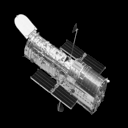} \\
			SBB, 0.08, 600s & FNNLS, $2\times 10^{-7}$, 694s  & {\tt lsqnonneg}, $7\times 10^{-12}$, 1610s & {\tt quadprog}, 0.002, 72.1s 
			%\includegraphics[scale=0.4]{walker} & 	
			%\includegraphics[scale=0.4]{blurred-walker-4} & 
			%\includegraphics[scale=0.4]{deblurred-fistabt-walker-4} &
			%\includegraphics[scale=0.4]{deblurred-sbb-walker-4} \\
			%walker: original & walker: $R=5$ & FISTABT & SBB \\
			%& &										$3.3\times 10^{-9}$\%, 2630s  & $1.6\times 10^{-5}$\%, 43.53s \\ \\
			%\includegraphics[scale=0.4]{deblurred-fnnls-walker-4} &
			%\includegraphics[scale=0.4]{deblurred-lsq-walker-4} &
			%\includegraphics[scale=0.4]{deblurred-qp-walker-4} & 
			%\includegraphics[scale=0.4]{deblurred-subset55-walker-4} \\
			%FNNLS & {\tt lsqnonneg} & {\tt quadprog} & SolveNNQ \\
			%$6\times 10^{-11}$\%, 414s  & $1.2\times 10^{-12}$\%, 748s & 0.001\%, 163s & 0.0018\%, 25.5s\\
		\end{tabular}
	}
	\caption{The upper left two images are the original and the blurred image.  Under the output of each software, the first number is the relative mean square error of the intensities of the corresponding pixels between the output and the original, and the second number is the running time.}
	\label{fg:blur-hst}
\end{figure}

\begin{table}
	\centering
	{\footnotesize
		\begin{tabular}{|l|c|r|r|r|r|r|r|}
			\hline
			& & \multicolumn{2}{c|}{SolveNNQ} & \multicolumn{2}{c|}{FISTABT} & \multicolumn{2}{c|}{SBB}  \\
			\hline%{3-14} 
			Data & $\sigma$ & rel.~err & time & rel.~err & time & rel.~err & time \\ 
			\hline
			%bcea, $\sigma=1$ & $1.3\times 10^{-9}$, 0.14s & $5.1\times 10^{-11}$, 0.07s & $8.7\times 10^{-16}$, 8.6s & $9\times 10^{-16}$, 13.3s & $4.1\times 10^{-7}$, 2.59s & $3.5\times 10^{-7}$, 0.81s \\
			%\hline
			nph & 1 & $10^{-5}$ & 3s & 0.009 & 2s & $3\times 10^{-5}$ & 171s  \\
			& 1.5 & $4\times 10^{-4}$ & 5.6s & 0.04 & 4.3s & 0.009 & 600s  \\
			& 2 & $7\times 10^{-4}$ & 8.2s & 0.06 & 6.8s & 0.02 & 600s\\
			\hline
			sts & 1 & $2\times 10^{-5}$ & 2.1s & 0.01 & 2.3s & $4\times 10^{-5}$ & 206s \\
			& 1.5 & $4\times 10^{-4}$ & 4.9s & 0.06 & 4.7s & 0.01 & 600s \\
			& 2 & $5\times 10^{-4}$ & 9.7s & 0.1 & 7.1s & 0.04 & 600s  \\
			\hline
			hst & 1 & $2\times 10^{-5}$ & 2.5s & 0.02 & 2.3s & $4\times 10^{-5}$ & 275s \\
			& 1.5 & 0.001 & 6.5s& 0.06 & 4.7s & 0.02 & 600s \\
			& 2 & 0.001 & 11.1s & 0.1 & 7.4s & 0.05 & 600s  \\
			\hline
			%walker, $R=2$ & $2\times 10^{-11}$, 350s & $1.4\times 10^{-7}$, 5.4s & $1.1\times 10^{-13}$, 144s & $3.6\times 10^{-15}$, 220s & $5.5\times 10^{-6}$, 26.9s & $5.5\times 10^{-6}$, 8.5s \\
			%\hline
			%walker, $R=3$ & $2.5\times 10^{-11}$, 1340s & $2.1\times 10^{-7}$, 16.4s & $1.9\times 10^{-13}$, 270s & $6.8\times 10^{-15}$, 447s & $1.5\times 10^{-5}$, 76.3s & $6\times 10^{-6}$, 17.6s \\
			%\hline
			%walker, $R=4$ & $3.3\times 10^{-11}$, 2630s & $1.6 \times 10^{-7}$, 43.5s & $6.1\times 10^{-13}$, 414s & $1.2\times 10^{-14}$, 748s & $1.1\times 10^{-5}$, 163s & $1.8 \times 10^{-5}$, 25.5s \\
			%\hline
	\end{tabular}}
		\centerline{
		\begin{tabular}{c}
		\end{tabular}
	}
	{\footnotesize
	\begin{tabular}{|l|c|r|r|r|r|r|r|}
			\hline
		&  & \multicolumn{2}{c|}{FNNLS} & \multicolumn{2}{c|}{{\tt lsqnonneg}} & \multicolumn{2}{c|}{{\tt quadprog}} \\
		\hline%{3-14} 
		Data & $\sigma$  & rel.~err & time & rel.~err & time & rel.~err & time \\ 
		\hline
		%bcea, $\sigma=1$ & $1.3\times 10^{-9}$, 0.14s & $5.1\times 10^{-11}$, 0.07s & $8.7\times 10^{-16}$, 8.6s & $9\times 10^{-16}$, 13.3s & $4.1\times 10^{-7}$, 2.59s & $3.5\times 10^{-7}$, 0.81s \\
		%\hline
		nph & 1  & $7\times 10^{-11}$ & 15.4s & $10^{-13}$ & 25s & $3\times 10^{-5}$ & 8.8s \\
		& 1.5  & $3\times 10^{-8}$ & 60s & $10^{-12}$ & 105s & $5\times 10^{-4}$ & 17.9s \\
		& 2  & $6\times 10^{-8}$ & 94s & $2\times 10^{-12}$ & 181s & $8\times 10^{-4}$ & 33.3s \\
		\hline
		sts & 1  & $8\times 10^{-11}$ & 14s & $9\times 10^{-14}$ & 22.7s & $3\times 10^{-5}$ & 9.4s \\
		& 1.5  & $10^{-8}$ & 48.9s & $2\times 10^{-12}$ & 87s & $4\times 10^{-4}$ & 20.4s \\
		& 2  & $3\times 10^{-8}$ & 80s & $2\times 10^{-12}$ & 156s & 0.001 & 38.6s \\
		\hline
		hst & 1 & $10^{-10}$ & 79.4s & $2\times 10^{-13}$ & 182s & $3\times 10^{-5}$ & 9.4s \\
		& 1.5 & $10^{-7}$ & 272s & $3\times 10^{-12}$ & 651s & 0.001 & 20.9s  \\
		& 2 & $3\times 10^{-7}$ & 459s & $4\times 10^{-12}$ & 1060s & 0.003 & 47.2s \\
		\hline
		%walker, $R=2$ & $2\times 10^{-11}$, 350s & $1.4\times 10^{-7}$, 5.4s & $1.1\times 10^{-13}$, 144s & $3.6\times 10^{-15}$, 220s & $5.5\times 10^{-6}$, 26.9s & $5.5\times 10^{-6}$, 8.5s \\
		%\hline
		%walker, $R=3$ & $2.5\times 10^{-11}$, 1340s & $2.1\times 10^{-7}$, 16.4s & $1.9\times 10^{-13}$, 270s & $6.8\times 10^{-15}$, 447s & $1.5\times 10^{-5}$, 76.3s & $6\times 10^{-6}$, 17.6s \\
		%\hline
		%walker, $R=4$ & $3.3\times 10^{-11}$, 2630s & $1.6 \times 10^{-7}$, 43.5s & $6.1\times 10^{-13}$, 414s & $1.2\times 10^{-14}$, 748s & $1.1\times 10^{-5}$, 163s & $1.8 \times 10^{-5}$, 25.5s \\
		%\hline
\end{tabular}}
	\caption{Experimental results for some space images.  In the column for each software, the number on the left is the relative mean square error, and the number on the right is the running time.}
	\label{tb:blur}
\end{table}

\begin{figure}
	\centering
	\begin{tabular}{cccc}
		\includegraphics[scale=0.55]{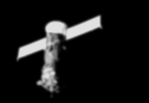} & 	
		\includegraphics[scale=0.55]{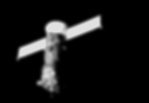} & 	
		\includegraphics[scale=0.55]{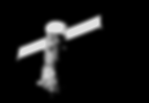} &
		\includegraphics[scale=0.55]{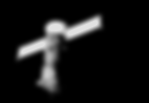} \\			
		\includegraphics[scale=0.55]{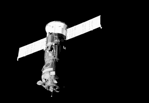} & 
		\includegraphics[scale=0.55]{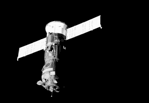} &
		\includegraphics[scale=0.55]{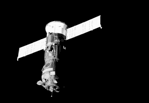} &
		\includegraphics[scale=0.55]{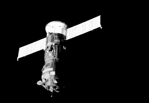} \\ 
		(a) & (b) & (c) & (d)
	\end{tabular}
	\caption{Blurring and deblurring nph: (a)~$\sigma=1$ and relative mean square error = $10^{-5}$, (b)~$\sigma=1.5$ and relative mean square error = $4\times 10^{-4}$, (c)~$\sigma=2$ and relative mean square error = $7\times 10^{-4}$, (d)~$\sigma=2.5$ and relative mean square error = $4\times 10^{-4}$.}	
	\label{fg:nph}
\end{figure}

\begin{figure}
	\centering
	\begin{tabular}{cccc}
		\includegraphics[scale=0.6]{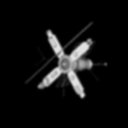} & 	
		\includegraphics[scale=0.6]{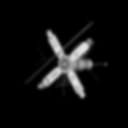} & 	
		\includegraphics[scale=0.6]{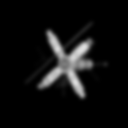} &
		\includegraphics[scale=0.6]{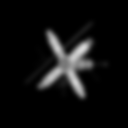} \\			
		\includegraphics[scale=0.6]{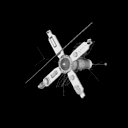} & 
		\includegraphics[scale=0.6]{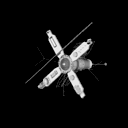} &
		\includegraphics[scale=0.6]{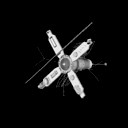} &
		\includegraphics[scale=0.6]{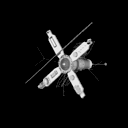} \\
		(a) & (b) & (c) & (d)
	\end{tabular}
	\caption{Blurring and deblurring sts: (a)~$\sigma=1$ and relative mean square error = $2\times 10^{-5}$, (b)~$\sigma=1.5$ and relative mean square error = $4\times 10^{-4}$, (c)~$\sigma=2$ and relative mean square error = $5 \times 10^{-4}$, (d)~$\sigma=2.5$ and relative mean square error = 0.002.}
	\label{fg:sts}
\end{figure}

\begin{figure}
	\centering
	\begin{tabular}{cccc}
		\includegraphics[scale=0.55]{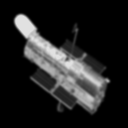} & 	
		\includegraphics[scale=0.55]{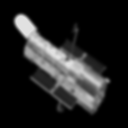} & 	
		\includegraphics[scale=0.55]{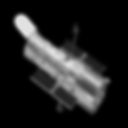} &
		\includegraphics[scale=0.55]{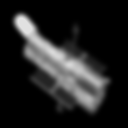} \\			
		\includegraphics[scale=0.55]{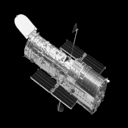} & 
		\includegraphics[scale=0.55]{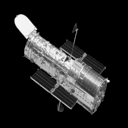} &
		\includegraphics[scale=0.55]{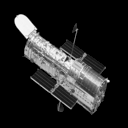} &
		\includegraphics[scale=0.55]{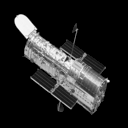} \\
		(a) & (b) & (c) & (d)
	\end{tabular}
	\caption{Blurring and deblurring hst: (a)~$\sigma=1$ and relative mean square error = $2\times 10^{-5}$, (b)~$\sigma=1.5$ and relative mean square error = 0.001, (c)~$\sigma=2$ and relative mean square error = 0.001, (d)~$\sigma=2.5$ and relative mean square error = 0.002.}
	\label{fg:hst}
\end{figure}

We compared SolveNNQ with several software that can solve NNLS, including FISTA with backtracking (FISTABT)~\cite{fistabt,BT09}, SBB~\cite{KSD13}, FNNLS~\cite{LH95}, {\tt lsqnonneg} of MATLAB, and a single call of {\tt quadprog}.  Table~\ref{tb:blur} shows the statistics of the relative mean square errors and the running times. Figure~\ref{fg:blur-hst} shows the deblurred images of hst produced by different software.   When the relative mean square error of an image is well below 0.01, it is visually non-distinguishable from the original.  In our experiments, SolveNNQ is very efficient and it produces a very good output.  
%More details and more results on thermal images are given in Appendix~\ref{app:blur}.  
Figures~\ref{fg:nph}--\ref{fg:hst} show some vertical pairs of blurred images and deblurred images recovered by SolveNNQ.  
%For each vertical pair, the value $\sigma$ used for blurring and the relative mean square error is shown.   
%When the relative mean square error is well below 0.01, the recovered image is visually non-distinguishable from the original.
We emphasize that we do not claim a solution for image deblurring as there are many issues that we do not address; we only seek to demonstrate the potential of the iterative scheme.

\begin{table}
	\centering
	{\footnotesize
		\begin{tabular}{|l|c|r|r|r|r|r|r|}
			\hline
			& & \multicolumn{2}{c|}{SolveNNQ} & \multicolumn{2}{c|}{FISTABT} & \multicolumn{2}{c|}{SBB}\\
			\hline%{3-14} 
			& $R$ & rel.~err & time & rel.~err & time & rel.~err & time  \\ 
			\hline
			%bcea, $\sigma=1$ & $1.3\times 10^{-9}$, 0.14s & $5.1\times 10^{-11}$, 0.07s & $8.7\times 10^{-16}$, 8.6s & $9\times 10^{-16}$, 13.3s & $4.1\times 10^{-7}$, 2.59s & $3.5\times 10^{-7}$, 0.81s \\
			%\hline
			walk & 2 & $6\times 10^{-6}$ & 9.6s & $2\times 10^{-11}$ & 350s & $10^{-7}$ & 5.4s  \\
			& 3 & $6\times 10^{-6}$ & 19.1s & $2\times 10^{-11}$ & 1340s & $2\times 10^{-7}$ & 16.4s\\
			& 4 & $10^{-5}$ & 33.1s& $3\times 10^{-11}$ & 2630s & $2\times 10^{-7}$ & 43.5s   \\
			\hline
			%mad & 2 & $10^{-6}$ & 1.3s & $5\times 10^{-11}$ & 441s & $9\times 10^{-8}$ & 2.3s & $10^{-13}$ & 29.6s & $4\times 10^{-15}$ & 43.9s & $3\times 10^{-6}$ & 9.7s \\
			%& 3 & $3\times 10^{-6}$ & 2.4s & $10^{-11}$ & 630s & $4\times 10^{-8}$ & 3.1s & $2\times 10^{-13}$ & 53.3s & $7\times 10^{-15}$ & 79.7s & $2\times 10^{-6}$ & 27.8s \\
			%& 4 & $3\times 10^{-6}$ & 3.9s & $t\times 10^{-9}$ & 800s & $10^{-7}$ & 10.3s & $6\times 10^{-13}$ & 82.2s & $10^{-14}$ & 132s & $4\times 10^{-6}$ & 40s \\
			%\hline
			heli & 2 & $3\times 10^{-6}$ & 1.9s & $2\times 10^{-10}$ & 800s & $3\times 10^{-7}$ & 4.9s \\
			& 3 & $2\times 10^{-6}$ & 4.6s & $2\times 10^{-10}$ & 800s & $2\times 10^{-7}$ & 8.1s  \\
			& 4 & $3 \times 10^{-6}$ & 8.2s & $9\times 10^{-8}$ & 800s & $5\times 10^{-7}$ & 36.8s\\
			\hline
			lion & 2 & $10^{-6}$ & 0.9s & $6\times 10^{-11}$ & 104s & $6 \times 10^{-8}$ & 0.9s  \\
			& 3 & $4 \times 10^{-6}$ & 1.6s & $6\times 10^{-11}$ & 185s & $6\times 10^{-8}$ & 1s\\
			& 4 & $9\times 10^{-6}$ & 3.2s & $5\times 10^{-11}$ & 800s & $2\times 10^{-7}$ & 3.7s \\
			\hline
	\end{tabular}}
	\centerline{
	\begin{tabular}{c}
	\end{tabular}
}{\footnotesize
\begin{tabular}{|l|c|r|r|r|r|r|r|}
	\hline
	&  & \multicolumn{2}{c|}{FNNLS} & \multicolumn{2}{c|}{{\tt lsqnonneg}} & \multicolumn{2}{c|}{{\tt quadprog}}  \\
	\hline%{3-14} 
	& $R$ & rel.~err & time & rel.~err & time & rel.~err & time \\ 
	\hline
	%bcea, $\sigma=1$ & $1.3\times 10^{-9}$, 0.14s & $5.1\times 10^{-11}$, 0.07s & $8.7\times 10^{-16}$, 8.6s & $9\times 10^{-16}$, 13.3s & $4.1\times 10^{-7}$, 2.59s & $3.5\times 10^{-7}$, 0.81s \\
	%\hline
	walk & 2 & $10^{-13}$ & 144s & $4\times 10^{-15}$ & 220s & $6\times 10^{-6}$ & 26.9s \\
	& 3 & $2\times 10^{-13}$ & 270s & $7\times 10^{-15}$ & 447s & $2\times 10^{-5}$ & 76.3s \\
	& 4 & $6\times 10^{-13}$ & 414s & $10^{-14}$ & 748s & $10^{-5}$ & 163s  \\
	\hline
	%mad & 2 & $10^{-6}$ & 1.3s & $5\times 10^{-11}$ & 441s & $9\times 10^{-8}$ & 2.3s & $10^{-13}$ & 29.6s & $4\times 10^{-15}$ & 43.9s & $3\times 10^{-6}$ & 9.7s \\
	%& 3 & $3\times 10^{-6}$ & 2.4s & $10^{-11}$ & 630s & $4\times 10^{-8}$ & 3.1s & $2\times 10^{-13}$ & 53.3s & $7\times 10^{-15}$ & 79.7s & $2\times 10^{-6}$ & 27.8s \\
	%& 4 & $3\times 10^{-6}$ & 3.9s & $t\times 10^{-9}$ & 800s & $10^{-7}$ & 10.3s & $6\times 10^{-13}$ & 82.2s & $10^{-14}$ & 132s & $4\times 10^{-6}$ & 40s \\
	%\hline
	heli & 2 & $2\times 10^{-13}$ & 57.2s & $5\times 10^{-15}$ & 93.8s & $4\times 10^{-6}$ & 12.3s \\
	& 3  & $5\times 10^{-13}$ & 103s & $10^{-14}$ & 187s & $6\times 10^{-6}$ & 37.2s \\
	& 4  & $10^{-12}$ & 157s & $2\times 10^{-14}$ & 289s & $8\times 10^{-6}$ & 54.7s \\
	\hline
	lion & 2  & $5\times 10^{-14}$ & 4.8s & $3\times 10^{-15}$ & 6.8s & $4\times 10^{-6}$ & 5.4s \\
	& 3 & $10^{-13}$ & 8.9s & $6\times 10^{-15}$ & 13s & $9\times 10^{-6}$ & 11.7s \\
	& 4  & $4\times 10^{-13}$ & 14.1s & $10^{-14}$ & 21.2s & $7\times 10^{-6}$ & 20.5 \\
	\hline
	\end{tabular}}
	\caption{Results for some thermal images.  In the column for each software, the number on the left is the relative mean square error, and the number on the right is the running time.}
	\label{tb:thermal}
\end{table}

\begin{figure}
	\centering
	{\tiny
		\begin{tabular}{cccc}
			\includegraphics[scale=0.1305]{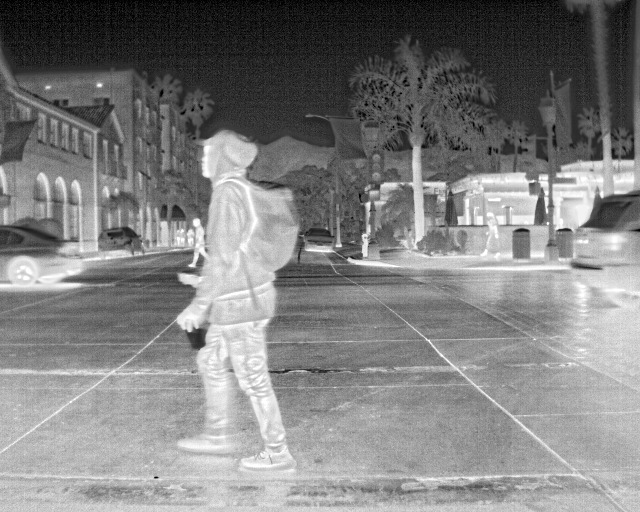} & 
			\includegraphics[scale=0.36975]{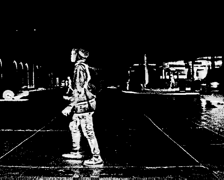} & 
			\includegraphics[scale=0.36975]{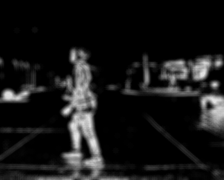} &
			\includegraphics[scale=0.36975]{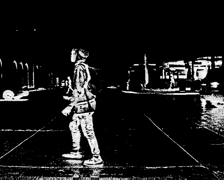} \\
			walk, thermal & walk, extracted & walk, $R=3$ & SolveNNQ, rel.~err = $6\times 10^{-6}$\\
			%\includegraphics[scale=0.245]{mad} & 
			%%includegraphics[scale=0.7]{mad_true} \\
			%mad, thermal & mad, extracted
			\includegraphics[scale=0.19575]{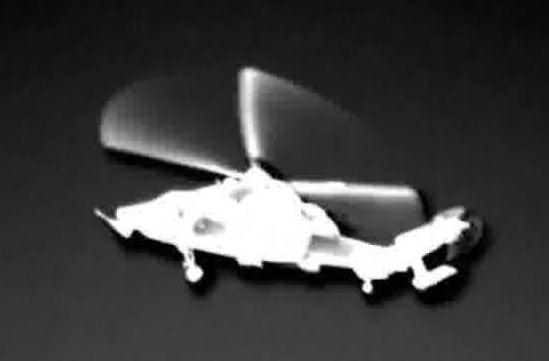} &
			\includegraphics[scale=0.4002]{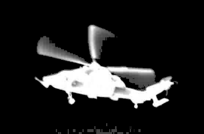} &
			\includegraphics[scale=0.4002]{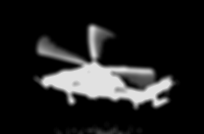} &
			\includegraphics[scale=0.4002]{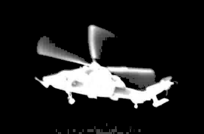} \\
			heli, thermal & heli, extracted & heli, $R=3$ & SolveNNQ, rel.~err = $2 \times 10^{-6}$ \\
			\includegraphics[scale=0.2175]{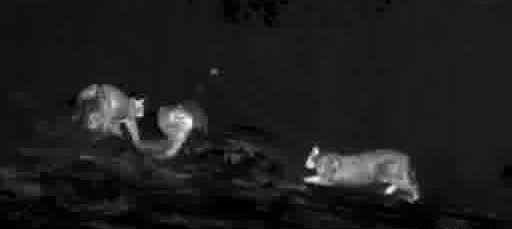} &
			\includegraphics[scale=0.435]{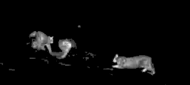} & 
			\includegraphics[scale=0.435]{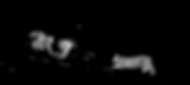} & 
			\includegraphics[scale=0.435]{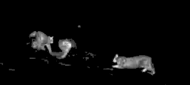} \\
			lion, thermal & lion, extracted & lion, $R=3$ & SolveNNQ, rel.~err = $7 \times 10^{-6}$ 
	\end{tabular}}
	\caption{Thermal images.}
	\label{fg:thermal}
\end{figure}

\begin{figure}
	\centering
	{\scriptsize
		\begin{tabular}{ccc}
			\includegraphics[scale=0.1925]{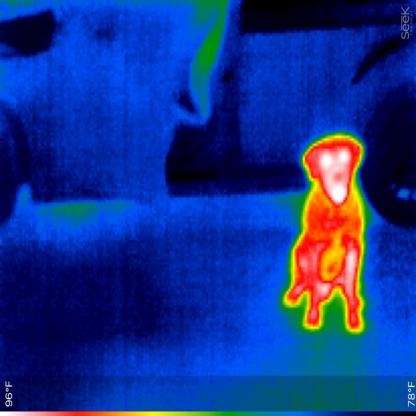} &
			\includegraphics[scale=0.55]{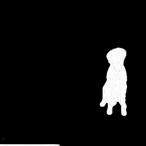} &
			\includegraphics[scale=0.55]{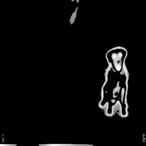} \\
		    (a) Thermal, original. & (b) Red, thresholded. & (c) Green, thresholded. \\
			\includegraphics[scale=0.55]{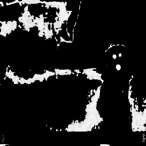} & 
			\includegraphics[scale=0.55]{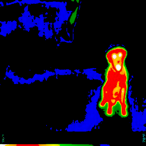} \\
		     (d) Blue, thresholded. & (e) Recombined.
	\end{tabular}}
	\caption{Thermal image of a dog, the three thresholded color panes, and their recombinations.}
	\label{fg:segment}
\end{figure}

%\begin{figure}
%	\centering
%	{\scriptsize
%		\begin{tabular}{ccccc}
%			\includegraphics[scale=0.1925]{dog} &
%			\includegraphics[scale=0.55]{dog-red-extracted} &
%			\includegraphics[scale=0.55]{dog-green-extracted} &
%			\includegraphics[scale=0.55]{dog-blue-extracted} & 
%			\includegraphics[scale=0.55]{dog-combined-extracted} \\
%			(a) Thermal, original. & (b) Red, thresholded. & (c) Green, thresholded. & (d) Blue, thresholded. & (e) Recombined.
%	\end{tabular}}
%	\caption{Thermal image of a dog, the three thresholded color panes, and their recombinations.}
%	\label{fg:segment}
%\end{figure}

\begin{table}
%	\centerline
		{\tiny
			\begin{tabular}{|c|c|c|c|c|c|c|c|c|c|c|c|c|}
				\hline
				& \multicolumn{4}{c|}{SolveNNQ} & \multicolumn{4}{c|}{FISTABT} & \multicolumn{4}{c|}{SBB} \\
				\hline%{2-13} 
				R & red & green & blue & total & red & green & blue & total & red & green & blue & total \\
				\hline
				2 & 0.5s & 0.6s & 2.6s & 3.7s & 362s & 48.1s & 160s & 570.1s & 2.1s & 0.7s & 1.5s & 4.3s \\
				\hline
				3 & 1s & 1.3s & 5.1s & 7.4s & 663s & 176s & 301s & 1140s & 2.4s & 1.3s & 2.8s & 6.5s \\
				\hline
				4 & 2s & 1.9s & 7.7s & 11.6s & 800s & 583s & 800s & 2183s & 11.5s & 5.5s & 9.8s & 26.8s \\
				\hline 
			\end{tabular}
	} 
	\centerline{
		\begin{tabular}{c}
		\end{tabular}
	}
	%\centerline
		{\tiny 
			\begin{tabular}{|c|c|c|c|c|c|c|c|c|c|c|c|c|}
				\hline
				& \multicolumn{4}{c|}{FNNLS} & \multicolumn{4}{c|}{{\tt lsqnonneg}} & \multicolumn{4}{c|}{{\tt quadprog}} \\
				\hline%{2-13} 
				R & red & green & blue & total & red & green & blue & total & red & green & blue & total \\
				\hline
				2 &  4.7s & 4.8s & 30.6s & 40.1s & 6.5s & 6.6s & 42.5s & 55.6s & 8.7s & 9.3s & 10.6s & 28.6s \\
				\hline
				3 & 8.4s & 8.1s & 54.1s & 70.6s & 12.2s & 11.3s & 79.2s & 102.7s & 23.1s & 26.6s & 30.6s & 80.3s \\
				\hline
				4 & 14s & 13.2s & 78.1s & 105.3s & 19.4s & 18.6s & 130s & 168s & 38.6s & 41.7s & 42.3s & 122.6s \\
				\hline
			\end{tabular}
	}
	\caption{Running times on the thresholded, blurred color panes of the dog image.}
	\label{tb:dog}
\end{table}

\begin{figure}
	\centering
	{\tiny
		\begin{tabular}{ccccc}
			\includegraphics[scale=0.43]{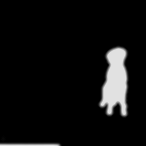} &
			\includegraphics[scale=0.43]{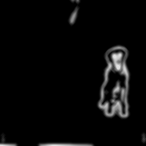} &
			\includegraphics[scale=0.43]{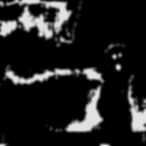} & 
			\includegraphics[scale=0.43]{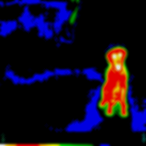} &
			\includegraphics[scale=0.43]{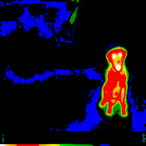} \\
			Red, $R=2$ & Green, $R=2$ & Blue, $R=2$ & Recombined, $R=2$ & SolveNNQ \\
			\includegraphics[scale=0.43]{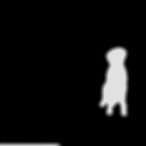} &
			\includegraphics[scale=0.43]{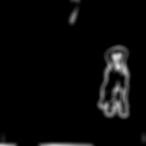} &
			\includegraphics[scale=0.43]{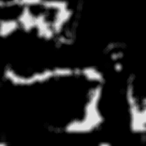} & 
			\includegraphics[scale=0.43]{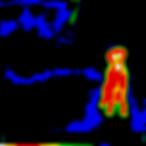} & 
			\includegraphics[scale=0.43]{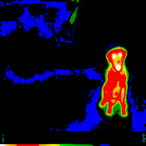} \\
			Red, $R=3$ & Green, $R=3$ & Blue, $R=3$ & Recombined, $R=3$ & SolveNNQ \\
			\includegraphics[scale=0.43]{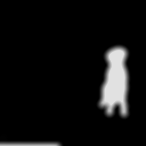} &
			\includegraphics[scale=0.43]{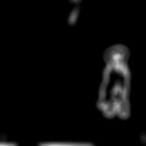} &
			\includegraphics[scale=0.43]{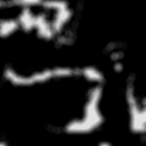} & 
			\includegraphics[scale=0.43]{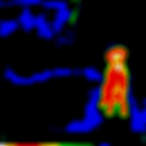} & 
			\includegraphics[scale=0.43]{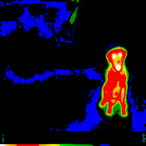} \\
			Red, $R=4$ & Green, $R=4$ & Blue, $R=4$ & Recombined, $R=4$ & SolveNNQ \\
	\end{tabular}}
	\caption{Blurred color panes, their recombinations, and the recombinations of the deblurred output of SolveNNQ.}
	\label{fg:blur}
\end{figure}

%\clearpage

We also experimented with some thermal images walk, heli, lion, and dog~\cite{walk,heli-lion,dog}; such thermal images are typically encountered in surveillance.  In each case, we extract the dominant color, down-sample, and apply thresholding to obtain a gray-scale image that omits most of the irrelevant parts.  Afterwards, we apply the uniform out-of-focus point spread function for blurring~\cite{LB91}---the atmospheric turbulence point spread function is not relevant in these cases.  There is a positive parameter $R$.  Given an $(a,b)$-pixel and a $(c,d)$-pixel, if $(a-c)^2 + (b-d)^2 \leq R^2$, the weight factor for them is $1/(\pi R^2)$; otherwise, the weight factor is zero. The larger $R$ is, the blurrier is the image.  The construction of the matrix $\mat{A}$ works as in the atmospheric turbulence point spread function.  

We use a certain number of the most negative coordinates of the gradient at the zero vector, and then we call {\tt quadprog} with these coordinates as the only free variables to obtain the initial solution $\mat{x}_1$.  Afterwards, SolveNNQ iterates as described previously.  Figure~\ref{fg:thermal} shows the thermal images, the extracted images, the blurred images, and the deblurred images produced by SolveNNQ.
%We set $N = 60\tau$ for walk, $N = 25\tau$ for heli, $N = 15\tau$ for lion.  

We did another experiment on deblurring the red, green, and blue color panes of a dog thermal image separately, followed by recombining the deblurred color panes.   Figure~\ref{fg:segment}(a) shows the original thermal image of a dog. Figures~\ref{fg:segment}(b)--(d) show the gray-scale versions of the red, green, and blue pixels which have been thresholded for sparsification.  Figure~\ref{fg:segment}(e) shows the image obtained by combining Figures~\ref{fg:segment}(b)--(d); the result is in essence an extraction of the dog in the foreground.   We blur the different color panes; Figure~\ref{fg:blur} shows these blurred color panes and their recombinations.  We deblur the blurred versions of the different color panes and then recombine the deblurred color panes.  The rightmost column in Figure~\ref{fg:blur} shows the recombinations of the deblurred output of SolveNNQ.  As in the case of the images walk, heli, and lion, the outputs of SolveNNQ, FISTABT, SBB, FNNLS, {\tt lsqnonneg}, and a single call {\tt quadprog} are visually non-distinguishable from each other.  Table~\ref{tb:dog} shows that the total running times of SolveNNQ over the three color panes are very competitive when compared with the other software.

\section{Analysis for NNLS and ZHLG}
\label{sec:zhlg}

In this section, we prove that  there exists a unit descent direction $\mat{n}_r$ in the $(r-1)$-th iteration of SolveNNQ such that $\frac{\langle \nabla f(\mat{x}_r),\mat{n}_r \rangle}{\langle \nabla f(\mat{x}_r), \mat{n}_* \rangle} \geq \frac{1}{\sqrt{2\nu\ln\nu}}$ for NNLS and ZHLG, where $\mat{n}_* = \frac{\mat{x}_*-\mat{x}_r}{\norm{\mat{x}_*-\mat{x}_r}}$, and $\nu = n$ for NNLS and $\nu = n(n-1)/2$ for ZHLG.  Combined with the assumption that the distance between $\mat{x}_r$ and the minimum in direction $\mat{n}_r$ is at least $\frac{1}{\lambda}\norm{\mat{x}_r-\mat{x}_*}$, we will show that the gap $f(\mat{x}_r) - f(\mat{x}_*)$ is decreased by a factor $e$ after every $O(\lambda\sqrt{\nu\ln\nu})$ iterations.

Let $L$ be an affine subspace in $\real^\nu$.   For every $\mat{x} \in \real^\nu$, define $\mat{x} \!\downarrow\! L$ to be the projection of $\mat{x}$ to the linear subspace parallel to $L$---the translate of $L$ that contains the origin.  Figure~\ref{fg:project} gives an illustration.  Note that $\mat{x} \!\downarrow\!L$ may not belong to $L$; it does if $L$ is a linear subspace.  The projection of $\mat{x}$ to $L$ can be implemented by multiplying $\mat{x}$ with an appropriate $\nu \times \nu$ matrix $\mat{M}$, i.e., $\mat{x} \!\downarrow\!L = \mat{Mx}$.  Thus, $(-\mat{x}) \!\downarrow \! L = \mat{M} \cdot (-\mat{x}) = -\mat{Mx} = 
-(\mat{x} \!\downarrow\!L)$, so we will write $-\mat{x} \!\downarrow\!L$ without any bracket.

Since there are only the non-negativity constraints for NNLS and ZHLG, we do not have the Lagrange multipliers $\mat{u}$ and $\tilde{\mat{u}}$.  So $(\mat{v}_r,\mat{x}_r)$ is the optimal solution of $\max_{\mat{v}} \min_{\mat{x}} g(\mat{v},\mat{x})$ that SolveNNQ solves in the $(r-1)$-th iteration under the constraints that $(\mat{v})_i \geq 0$ for all $i \not\in S_r$.
%The Lagrange multipliers $\mat{u}_r$ and $\tilde{\mat{u}}_r$ are absent because there is no constraints other than $\mat{x} \geq 0_\nu$.  By the KKT conditions, we have $\mat{x}_r \geq 0_\nu$ and $(\mat{v}_r)_i \geq 0$ for all $i \not\in S_r$.

\begin{figure}
	\centerline{\includegraphics[scale=0.65]{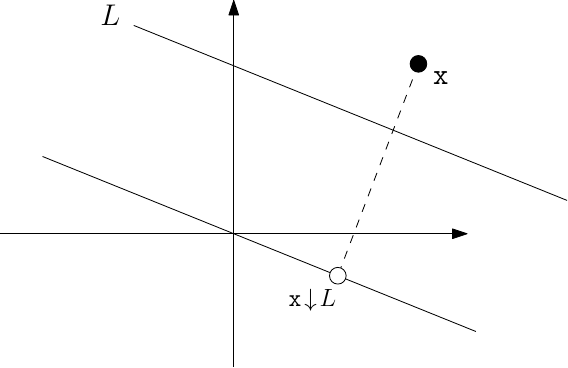}}
	\caption{In $\real^2$, a point $\mat{x}$, a line $L$, and the projection $\mat{x} \!\downarrow\! L$.}
	\label{fg:project}
\end{figure}

Recall that, in the pseudocode of SolveNNQ, $E_r$ is a sequence of indices $i \in S_r$ such that $(\mat{v}_r)_i < 0$, and $E_r$ is sorted in non-decreasing order of $(\mat{v}_r)_i$.  Define $J_r = \mathtt{span}(\{\mat{e}_i : i \in E_r\})$.   The following result follows from the definitions of $E_r$ and $J_r$.

\begin{lemma}
	\label{lem:projpro}
	For all $i \in [\nu]$, if $i \in E_r$, then $(\mat{v}_r \!\downarrow\! J_r)_i = (\mat{v}_r)_i$; otherwise, $(\mat{v}_r \!\downarrow\! J_r)_i = 0$.  Moreover, $-\mat{v}_r \!\downarrow\!J_r$ is a conical combination of $\{\mat{e}_i : i \in E_r\}$.
\end{lemma}

We show that there exists a subset of indices $i \in E_r$ such that $\langle -\mat{v}_r \!\downarrow\! J_r,\mat{e}_i \rangle$ is bounded away from zero.

\begin{lemma}
	\label{lem:tech-1}
	For any $\alpha \in (0,1]$, there exists $j$ in $E_r$ with rank at most $\alpha |E_r|$ such that for every $i \in E_r$, if $i = j$ or $i$ precedes $j$ in $E_r$, then $\langle \mat{v}_r \!\downarrow\! J_r,\mat{e}_i \rangle^2 \geq \norm{\mat{v}_r \!\downarrow\! J_r}^2 \cdot \alpha/\bigl(j\ln\nu\bigr)$.
\end{lemma}
\begin{proof}
	Consider a histogram $H_1$ of $\alpha/(i\ln \nu)$ against $i \in E_r$.   The total length of the vertical bars in $H_1$ is equal to $\sum_{i \in E_r} \alpha/(i\ln \nu) \,\, \leq \,\, (\alpha/\ln \nu) \cdot \sum_{i=1}^\nu  1/i \,\, \leq \,\, \alpha$.  Consider another histogram $H_2$ of $\langle \mat{v}_r \!\downarrow\! J_r,\mat{e}_i\rangle^2/\norm{\mat{v}_r \!\downarrow\! J_r}^2$ against $i \in E_r$.  The total length of the vertical bars in $H_2$ is equal to $\sum_{i\in E_r} \langle \mat{v}_r \!\downarrow\! J_r,\mat{e}_i\rangle^2/\norm{\mat{v}_r \!\downarrow\! J_r}^2$, which is equal to 1 because  $\mat{v}_r\!\downarrow\!J_r \in J_r$.  Recall that the indices of $E_r$ are sorted in non-decreasing order of $({\mat{v}}_r)_i$.  As $(\mat{v}_r\!\downarrow\!J_r)_i = (\mat{v}_r)_i < 0$ for all $i \in E_r$, the ordering in $E_r$ is the same as the non-increasing order of $\langle \mat{v}_r\!\downarrow\!J_r,\mat{e}_i \rangle^2/\norm{\mat{v}_r\!\downarrow\!J_r}^2 = (\mat{v}_r\!\downarrow\!J_r)_i^2/\norm{\mat{v}_r\!\downarrow\!J_r}^2$.  Therefore, the total length of the vertical bars in $H_2$ for the first $\alpha |E_r|$ indices is at least $\alpha$.  It implies that when we scan the indices in $E_r$ from left to right, we must encounter an index $j$ among the first $\alpha|E_r|$ indices such that the vertical bar in $H_2$ at $j$ is not shorter than the vertical bar in $H_1$ at $j$, i.e., $\langle \mat{v}_r\!\downarrow\!J_r,\mat{e}_i \rangle^2/\norm{\mat{v}_r\!\downarrow\!J_r}^2 \geq \alpha/(j\ln \nu)$ for every $i$ that precedes $j$ in $E_r$.  This index $j$ satisfies the lemma.
\end{proof}

Next, we boost the angle bound implied by Lemma~\ref{lem:tech-1} by showing that $-\mat{v}_r \!\downarrow\!J_r$ makes a much smaller angle with some conical combination of $\{\mat{e}_i : i \in E_r\}$.

\begin{figure}
	\centerline{\includegraphics[scale=0.65]{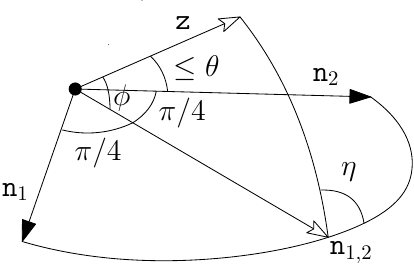}}
	\caption{The vector $\mat{n}_{1,2}$ bisects the right angle $\angle (\mat{n}_1,\mat{n}_2)$.  The angle $\eta$ is at least $\pi/2$.}
	\label{fg:nice}
\end{figure}

\begin{lemma}
	\label{lem:nice}
	Take any $c \leq 1/\sqrt{2}$.  Let $\mat{z}$ be a vector in $\real^D$ for some $D \geq 2$.  Suppose that there is a set $V$ of unit vectors in $\real^D$ such that the vectors in $V$ are mutually orthogonal, and for every $\mat{n} \in V$, $\cos \angle (\mat{n},\mat{z}) \geq c|V|^{-1/2}$.  There exists a conical combination $\mat{y}$ of the vectors in $V$ such that $\cos\angle (\mat{y},\mat{z}) \geq c/\sqrt{2}$.
\end{lemma}
\begin{proof}
	Let $\theta = \arccos\left(c |V|^{-1/2}\right)$.   If $\theta \leq \pi/3$, we can pick any vector $\mat{n} \in V$ as $\mat{y}$ because $\pi/3 \leq \arccos(c/\sqrt{2})$ for $c \leq 1/\sqrt{2}$.  Suppose not.  Let $W$ be a maximal subset of $V$ whose size is a power of 2.  Arbitrarily label the vectors in $W$ as $\mat{n}_1, \mat{n}_2, \ldots$.   Consider the unit vector $\mat{n}_{1,2} = \frac{1}{\sqrt{2}}\mat{n}_1 + \frac{1}{\sqrt{2}}\mat{n}_2$.   Let $\phi = \angle (\mat{n}_{1,2},\mat{z})$.  Refer to Figure~\ref{fg:nice}.   By assumption, $\mat{n}_1 \perp \mat{n}_2$.  Let $\eta$ be the non-acute angle between the plane $\mathrm{span}(\{\mat{n}_1,\mat{n}_2\})$ and the plane $\mathrm{span}(\{\mat{n}_{1,2},\mat{z}\})$.  By the spherical law of cosines, $\cos\theta \leq \cos\angle (\mat{n}_2,\mat{z}) = \cos\phi\cos(\pi/4) + \sin\phi\sin(\pi/4)\cos\eta$.
	Note that $\cos\eta\leq 0$ as $\eta \geq \pi/2$.  It implies that $\cos\phi \geq \sec(\pi/4) \cos\theta = \sqrt{2}\cos\theta$.
	The same analysis holds between $\mat{z}$ and the unit vector $\mat{n}_{3,4} = \frac{1}{\sqrt{2}}\mat{n}_3 + \frac{1}{\sqrt{2}}\mat{n}_4$, and so on.  In the end, we obtain $|W|/2$ vectors $\mat{n}_{2i-1,2i}$ for $i = 1, 2, \ldots,|W|/2$ such that $\angle (\mat{n}_{2i-1,2i},\mat{z}) \leq \arccos\bigl(\sqrt{2}\cos\theta\bigr)$.  Call this the first stage.   Repeat the above with the $|W|/2$ unit vectors $\mat{n}_{1,2}, \mat{n}_{3,4}, \ldots$ in the second stage and so on.  We end up with only one vector in $\log_2 |W|$ stages.  If we ever produce a vector that makes an angle at most $\pi/3$ with $\mat{z}$ before going through all $\log_2 |W|$ stages, the lemma is true.  Otherwise, we produce a vector $\mat{y}$ in the end such that $\cos \angle (\mat{y},\mat{z}) \geq \bigl(\sqrt{2}\bigr)^{\log_2 |W|}\cos\theta \geq \bigl(\sqrt{2}\bigr)^{\log_2 |V|-1}\cos\theta  \geq \sqrt{|V|/2} \cdot \cos\theta = c/\sqrt{2}$.
\end{proof}

\begin{lemma}
	\label{lem:angle-cor}
	%Suppose that we apply \emph{SolveNNQ} to a \emph{ZHLG} problem.  
	%Let $\mat{x}_r$ be the optimal solution for $f$ constrained by the active set $S_r$.
	%For any positive integer $\tau$, 
	The vectors in $\bigl\{\mat{e}_i : \text{$i$ among the first $\tau$ indices in $E_r$}\bigr\}$ have a unit conical combination $\mat{n}_r$ such that $\mat{n}_r$ is a descent direction from $\mat{x}_r$ and $\langle -\mat{v}_r\!\downarrow\!J_r,{\mat{n}}_r \rangle \geq \norm{\mat{v}_r \!\downarrow\! J_r}/\sqrt{2\nu\ln\nu}$.
\end{lemma}
\begin{proof}
	Since $-\mat{v}_r \!\downarrow\!J_r$ is a conical combination of $\bigl\{\mat{e}_i : i \in E_r\bigr\}$, $-\mat{v}_r\!\downarrow\!J_r$ makes a non-obtuse angle with any vector in $\{ \mat{e}_i : i \in E_r\}$.  Then, Lemmas~\ref{lem:tech-1} and~\ref{lem:nice} imply that the vectors in $\bigl\{\mat{e}_i : \text{$i$ among the first $\tau$ indices in $E_r$}\bigr\}$ have a unit conical combination $\mat{n}_r$ such that $\langle -\mat{v}_r\!\downarrow\!J_r,{\mat{n}}_r \rangle \geq \norm{\mat{v}_r\!\downarrow\!J_r} \cdot \sqrt{\alpha/(2\ln\nu)}$, where $\alpha = \tau/|E_r|$.  As $|E_r| \leq \nu$, we have $\alpha \geq 1/\nu$ which implies that $\langle -\mat{v}_r\!\downarrow\!J_r,{\mat{n}}_r \rangle \geq \norm{\mat{v}_r\!\downarrow\!J_r}/\sqrt{2\nu\ln\nu}$.  For any feasible solution $\mat{x}$ and any non-negative values $c_1,\ldots,c_\nu$, $\mat{x} + \sum_{i=1}^\nu c_i\mat{e}_i$ is also a feasible solution.  It follows that $\mat{x}_r + s\mat{n}_r$ is a feasible solution for all $s \geq 0$.   For NNLS and ZHLG, the primal problem has no constraint other than the non-negativity constraints $\mat{x} \geq 0_\nu$.  By the KKT conditions, we have $\partial g(\mat{x}_r)/\partial \mat{x} = \nabla f(\mat{x}_r) - \mat{v}_r = 0_\nu$, which implies that $\nabla f(\mat{x}_r) = \mat{v}_r$.  Moreover, $\langle -\mat{v}_r \!\downarrow\! J_r, \mat{n}_r \rangle = \langle -\mat{v}_r,\mat{n}_r \rangle$ because $\mat{n}_r \in J_r$ and so the component of $-\mat{v}_r$ orthogonal to $J_r$ vanishes in the inner product.  Then, the acuteness of $\angle (-\mat{v}_r\!\downarrow\!J_r,{\mat{n}}_r)$ implies that $\langle \nabla f(\mat{x}_r),\mat{n}_r \rangle < 0$.  In all, $\mat{n}_r$ is a descent direction.
\end{proof}

We use Lemma~\ref{lem:angle-cor} to establish a convergence rate under an assumption.

\begin{lemma}
	\label{lem:full}
	%Suppose that we apply \emph{SolveNNQ} to a \emph{ZHLG} problem.  
	%Let $\mat{x}_r$ be the optimal solution for $f$ constrained by the active set $S_r$.  
	Let ${\mat{n}}_r$ be a unit descent direction from $\mat{x}_r$ that satisfies Lemma~\ref{lem:angle-cor}.  Let $\mat{y}_r$ be the feasible point that minimizes $f$ on the ray from $\mat{x}_r$ in direction ${\mat{n}}_r$.
	\begin{enumerate}[{\em (i)}]
		\item $\frac{\langle \nabla f(\mat{x}_r), \, \mat{n}_r\rangle}{\langle \nabla f(\mat{x}_r), \, \mat{n}_*\rangle}\geq \frac{1}{\sqrt{2\nu\ln\nu}}$, where $\mat{n}_* = \frac{\mat{x}_*-\mat{x}_r}{\norm{\mat{x}_* - \mat{x}_r}}$.
		\item If there exists $\lambda$ such that $\norm{\mat{x}_{r} - \mat{y}_{r}} \geq \frac{1}{\lambda} \norm{\mat{x}_{r} - \mat{x}_*}$, then $\frac{f(\mat{x}_{r+1}) - f(\mat{x}_*)}{f(\mat{x}_r) - f(\mat{x}_*)} \leq 1 - \frac{1}{2\lambda\sqrt{2\nu\ln\nu}}$.
	\end{enumerate}
\end{lemma}	
\begin{proof}
	Irrespective of whether $S_{r+1}$ is computed in line~\ref{alg:exceed-end}~or~\ref{alg:free} of SolveNNQ, $S_{r+1}$ is disjoint from $\supp(\mat{x}_r) \cup \supp(\mat{n}_r)$, which makes Lemma~\ref{lem:bound2} applicable.
	
	Since $\nabla f(\mat{x}_r) = \mat{v}_r$ for an NNLS or  ZHLG problem, we have $\langle -\nabla f(\mat{x}_r),\mat{n}_r) = \langle -\mat{v}_r,\mat{n}_r \rangle = \langle -\mat{v}_r\!\downarrow\!J_r,\mat{n}_r \rangle$.  The last equality comes from the fact that $\mat{n}_r \in J_r$, so the component of $-\mat{v}_r$ orthogonal to $J_r$ vanishes in the inner product.  By Lemma~\ref{lem:angle-cor}, $\langle -\nabla f(\mat{x}_r),{\mat{n}}_r \rangle = \langle -\mat{v}_r\!\downarrow\!J_r,{\mat{n}}_r \rangle \geq \norm{\mat{v}_r\!\downarrow\!J_r} /\sqrt{2\nu\ln\nu}$.
	
	The inequality above takes care of the nominator in the bound in Lemma~\ref{lem:bound2} multiplied by $-1$.   The denominator in the bound in Lemma~\ref{lem:bound2} multiplied by $-\norm{\mat{x}_*-\mat{x}_r}$ is equal to $\langle -\nabla f(\mat{x}_r),\mat{x}_*-\mat{x}_r \rangle = \langle -\mat{v}_r,\mat{x}_*-\mat{x}_r \rangle = \langle -\mat{v}_r\!\downarrow\!J_r,\mat{x}_*-\mat{x}_r \rangle + \langle -\mat{v}_r + \mat{v}_r\!\downarrow\!J_r,\mat{x}_*-\mat{x}_r \rangle$.
	
	Recall that $(\mat{v}_r\!\downarrow\!J_r)_i$ is $(\mat{v}_r)_i$ for $i \in E_r$ and zero otherwise.  Therefore, $(-\mat{v}_r + \mat{v}_r \!\downarrow\!J_r)_i$ is zero for $i \in E_r$ and $-(\mat{v}_r)_i$ otherwise.   If $i \not\in E_r$, then $(\mat{v}_r)_i \geq 0$, which implies that $(-\mat{v}_r + \mat{v}_r\!\downarrow\!J_r)_i = -(\mat{v}_r)_i \leq 0$.  By the complementary slackness, if $(\mat{v}_r)_i > 0$, then $(\mat{x}_r)_i = 0$, which implies that $(\mat{x}_* - \mat{x}_r)_i \geq 0$ as $\mat{x}_*$ is non-negative.  Altogether, we conclude that for $i \in [\nu]$, if $i \in E_r$ or ($i \not\in E_r \, \wedge \, (\mat{v}_r)_i = 0)$, then
	$(-\mat{v}_r + \mat{v}_r\!\downarrow\!J_r)_i \cdot (\mat{x}_* - \mat{x}_r)_i = 0$; otherwise, 
	$(-\mat{v}_r + \mat{v}_r\!\downarrow\!J_r)_i \cdot (\mat{x}_* - \mat{x}_r)_i \leq 0$.
	Therefore, $\langle -\mat{v}_r + \mat{v}_r \!\downarrow\!J_r,\mat{x}_*-\mat{x}_r \rangle \leq 0$.  As a result, $\langle -\nabla f(\mat{x}_r),\mat{x}_*-\mat{x}_r \rangle \leq \langle -\mat{v}_r\!\downarrow\!J_r,\mat{x}_*-\mat{x}_r \rangle \leq \norm{\mat{x}_*-\mat{x}_r} \cdot \norm{\mat{v}_r\!\downarrow\!J_r}$.
	
	Hence, $\frac{\langle \nabla f(\mat{x}_r),\mat{n}_r\rangle}{\langle \nabla f(\mat{x}_r),\mat{n}_*\rangle}\geq \frac{1}{\sqrt{2\nu\ln\nu}}$, where $\mat{n}_* = \frac{\mat{x}_*-\mat{x}_r}{\norm{\mat{x}_* - \mat{x}_r}}$.  This completes the proof of (i).
	
	Substituting (i) into Lemma~\ref{lem:bound2} gives $\frac{f(\mat{x}_r) - f(\mat{x}_{r+1})}{f(\mat{x}_r)-f(\mat{x}_*)} \geq \frac{\norm{\mat{x}_r-\mat{y}_r}}{2\norm{\mat{x}_*-\mat{x}_r}} \cdot \frac{1}{\sqrt{2\nu\ln\nu}}$.  It then follows from the assumption of $\norm{\mat{x}_{r} - \mat{y}_{r}} \geq \norm{\mat{x}_{r} - \mat{x}_*}/\lambda$ in the lemma that $f(\mat{x}_{r+1}) - f(\mat{x}_*) = \bigl(f(\mat{x}_r) - f(\mat{x}_*)\bigr) - \bigl(f(\mat{x}_r) - f(\mat{x}_{r+1})\bigr) 
	\leq \bigl(f(\mat{x}_r) - f(\mat{x}_*) \bigr) - \frac{1}{2\lambda\sqrt{2\nu\ln \nu}} \bigl(f(\mat{x}_r) - f(\mat{x}_*) \bigr)$.
\end{proof}

Let $T(n)$ denote the time complexity of solving a convex quadratic program with $n$ variables and $O(n)$ constraints. It is known that $T(n) = O(n^3 L)$, where $L$ is bounded by the number of bits in the input~\cite{monteiro89}.  Theorem~\ref{thm:main} below gives the performance of SolveNNQ on NNLS and ZHLG.
%finds the optimum in $O(k^3 + \beta_0^3 + k\nu)$ time under some assumptions, where $k = \max_{r \geq 1}|\supp(\mat{x}_r)|$.  

\begin{theorem}
	\label{thm:main}
	Consider the application of \emph{SolveNNQ} on an \emph{NNLS} problem with $n$ constraints in $\real^d$ or a \emph{ZHLG} problem for $n$ points in $\real^d$.  Let $\nu$ be $n$ for \emph{NNLS} or $n(n-1)/2$ for \emph{ZHLG}.  The initialization of \emph{SolveNNQ} can be done in $T(\beta_0)+O(dn^2)$ time for \emph{NNLS} or $T(\beta_0)+O(n^4)$ time for \emph{ZHLG}.    Suppose that the following assumptions hold.
	\begin{itemize}
		\item Assume that the threshold $\beta_1$ on the total number of iterations is not exceeded.
		\item For all $r \geq 1$, let $\mat{n}_r$ be a unit descent direction from $\mat{x}_r$ that satisfies Lemma~\ref{lem:angle-cor}, assume that $\norm{\mat{x}_{r} - \mat{y}_{r}} \geq \frac{1}{\lambda}\norm{\mat{x}_{r}-\mat{x}_*}$, where $\lambda$ is some fixed value and $\mat{y}_r$ is the feasible point that minimizes $f$ on the ray from $\mat{x}_{r}$ in the direction $\mat{n}_r$.
	\end{itemize}
	Then, for all $r \geq 1$, $f(\mat{x}_{r+i}) - f(\mat{x}_*) \leq e^{-1} (f(\mat{x}_r) - f(\mat{x}_*))$ for some $i = O\bigl(\lambda \sqrt{\nu\log \nu}\bigr)$, and each iteration in \emph{SolveNNQ} takes $T(k + \beta_0)+O(k\nu)$ time, where $k = \max_{r \geq 1}|\supp(\mat{x}_r)|$.
\end{theorem}
\begin{proof}
	Since the threshold $\beta_1$ is not exceeded by assumption, the set $E_r$ contains at least $\tau$ indices for all $r \geq 1$.   Therefore, Lemma~\ref{lem:full} implies that the gap $f(\mat{x}_r)-f(\mat{x}_*)$ decreases by a factor $e$ in $O(\lambda\sqrt{\nu\ln\nu})$ iterations.  It remains to analyze the running time.
	
	For NNLS, $\mat{A}$ is $n \times d$ and we compute $\mat{A}^t\mat{A}$ in $O(dn^2)$ time.  For ZHLG, $\mat{A}$ is $n(n+1)/2 \times n(n-1)/2$ and $\mat{a}$ is a vector of dimension $n(n-1)/2$.  Recall that $\mat{A}^t = \bigl[ \frac{\mu}{2}\mat{U}^t \,\,\, \frac{\rho}{2}\mat{I}_{n(n-1)/2}\bigr]$, where $\mat{U} \in \real^{n\times n(n-1)/2}$ is the incidence matrix for the complete graph.  Therefore, each row of $\mat{A}^t$ has at most three non-zero entries, meaning that $\mat{A}^t\mat{A}$ can be computed in $O(n^4)$ time. 
	
	The initial solution is obtained as follows.  Sample $\beta_0$ indices from $[\nu]$.  The initial active set contains all indices in the range $[\nu]$ except for these sampled indices.  We extract the rows and columns of $\mat{A}^t\mat{A}$ and coordinates of $\mat{a}$ that correspond to these $\beta_0$ sampled indices.  Then, we call the solver in $T(\beta_0)$ time to obtain the initial solution $\mat{x}_1$.  %$O(\beta_0^3)$
	
	At the beginning of every iteration, we compute $\nabla f(\mat{x}_r)$ and update the active set.   This step can be done in $O(k\nu)$ time because $\mat{x}_r$ has at most $k$ non-zero entries.  At most $k + \beta_0$ indices are absent from the updated active set because the threshold $\beta_1$ is not exceeded.  We select the rows and columns of $\mat{A}^t\mat{A}$ and coordinates of $\mat{a}$ that correspond to the indices absent from the active set in $O(k^2+ \beta_0^2)$ time.  The subsequent call of the solver takes $T(k + \beta_0)$ time.    So each iteration runs in $T(k + \beta_0)+O(k\nu)$ time. % $O(k^3 + \beta_0^3)$ $O(k^3 + \beta_0^3 + k\nu)$
	
	%By Lemma~\ref{lem:full}, we need $O\bigl(\lambda \sqrt{\nu\log \nu/\tau}\big)$ iterations to reduce the gap $f(\mat{x}_r) - f(\mat{x}_*)$ by a factor $e$.
\end{proof}

We assume in Theorem~\ref{thm:main} that $\norm{\mat{x}_{r} - \mat{y}_{r}} \geq \frac{1}{\lambda}\norm{\mat{x}_{r}-\mat{x}_*}$ for all $r \geq 1$ for ease of presentation.  The geometric-like convergence still holds even if the inequality $\norm{\mat{x}_r-\mat{y}_r} \geq \frac{1}{\lambda} \norm{\mat{x}_r-\mat{x}_*}$ is satisfied in only a constant fraction of any sequence of consecutive iterations (the iterations in that constant fraction need not be consecutive).  

\section{Analysis for DKSG}
\label{sec:second}

Let $\mat{p}_1,\ldots,\mat{p}_n$ be the input points.  So there are $n(n-1)/2$ edges connecting them, and the problem is to determine the weights of these edges.    Suppose that we have used SolveNNQ to obtain the current solution $(\mat{u}_r,\mat{v}_r,\mat{x}_r)$ with respect to the active set $S_r$.  We are to analyze the convergence rate of SolveNNQ when computing the next solution $(\mat{u}_{r+1},\mat{v}_{r+1},\mat{x}_{r+1})$ with respect to the  next active set $S_{r+1}$.

In our analysis, we assume that the next active set $S_{r+1}$ is computed as $[n(n-1)/2] \setminus (\supp(\mat{x}_r) \cup E'_r)$, which is the same assumption that we made in the case of NNLS and ZHLG.  Recall that $E'_r$ is the subset of $E_r$ that SolveNNQ extracted so that variables in $E'_r$ will be freed in the next iteration.

To analyze the convergence rate achieved by $(\mat{u}_{r+1},\mat{v}_{r+1},\mat{x}_{r+1})$, we transform the original system to an extended system and analyze an imaginary run of SolveNNQ on this extended system.   First, we convert the inequality constraints to equality constraints by introducing $n$ slack variables.   Second, we need another large entry in each row of  the contraint matrix for a technical reason.  So we introduce another $n$ dummy variables.  We give the details of the transformation in the next section.

\subsection{Extended system}

The original variables are $(\mat{x})_j$ for $j \in [n(n-1)/2]$.  We add one slack variable per constraint in $\mat{Bx} \geq 1_n$.  This gives a vector $\mat{s} \in \real^n$ of slack variables.  We also add a vector $\mat{d} \in \real^n$ of dummy variables, one for each constraint in $\mat{Bx} \geq 1_n$.   As a result, the ambient space expands from $\real^{n(n-1)/2}$ to $\real^{n(n+3)/2}$.  We use $\mat{x}$ to denote a point in $\real^{n(n-1)/2}$.  The idea is to map $\mat{x}$ to a point $\hat{\mat{x}}$ in $\real^{n(n+3)/2}$ by appending slack and dummy variables:
\[
\hat{\mat{x}} = \begin{bmatrix}
	\mat{x} \\
	\mat{s} \\
	\mat{d}
\end{bmatrix}.
\]
%For $j \in [n(n-1)/2]$, $(\hat{\mat{x}})_j = (\mat{x})_j$.  For $j \in [n(n-1)/2+1,n(n+1)/2]$, $(\hat{\mat{x}})_j = (\mat{s})_{j-n(n-1)/2}$.  For $j \in [n(n+1)/2+1,n(n+3)/2]$, $(\hat{\mat{x}})_j = (\mat{d})_{j-n(n+1)/2}$.
%
Let $M$ be an arbitrarily large number.
%based on $n$ and the current solution $(\mat{u}_r,\mat{v}_r,\mat{x}_r)$ of the original system.  
Define the following matrix $\hat{\mat{B}} \in \real^{n \times n(n+3)/2}$:
\[
\hat{\mat{B}} =
\begin{bmatrix}
	\, \mat{B} & & -\mat{I}_n &&  M\mat{I}_n  \, 
\end{bmatrix}.
\]
The constraints for the extended system in $\real^{n(n+3)/2}$ are:
\[
\hat{\mat{B}}\hat{\mat{x}}  = (1+M)\, 1_n, \quad
\hat{\mat{x}} \geq 0_{n(n+3)/2}.
\]
That is, for each point $\mat{p}_i$, the sum of the weights of edges incident to $\mat{p}_i$, minus $(\mat{s})_i$, and plus $M(\mat{d})_i$ must be equal to $1+M$.   If we can force $\mat{d}$ to be $1_n$, we can recover the original requirement of the sum of the weights of edges incident to $\mat{p}_i$ being at least 1 because the slack variable $(\mat{s})_i$ can cancel any excess weight.

We can make $\mat{d}$ nearly $1_n$ by changing the objective function from $\mat{x}^t\mat{A}^t\mat{Ax}$ to $\mat{x}^t\mat{A}^t\mat{Ax} + \frac{1}{2}M^3\norm{\mat{d}-1_n}^2$.  Although $\mat{d}$ may not be exactly $1_n$ at an optimum, we will argue that it suffices for our purpose.   Let $\hat{f} : \real^{n(n+3)/2} \rightarrow \real$ denote the new objective function $\mat{x}^t\mat{A}^t\mat{Ax} + \frac{1}{2}M^3\norm{\mat{d}-1_n}^2$.   Note that $\hat{f}$ is quadratic.  We have 
\[
\nabla \hat{f}(\hat{\mat{x}}) = \begin{bmatrix} 
	\nabla f(\mat{x})  \\
	0_n \\
	M^3(\mat{d} - 1_n)
\end{bmatrix} =
\begin{bmatrix}
	2\mat{A}^t\mat{Ax} \\
	0_n \\
	M^3(\mat{d}-1_n)
\end{bmatrix}.
\]

%We get extra Lagrange multipliers in the dual problem.  
Let $\hat{\mat{u}}$ denote the $n$-dimensional vector of multipliers corresponding to the constraints $\hat{\mat{B}}\hat{\mat{x}} = (1+M)1_n$.  Let $\hat{\mat{v}}$ denote the $(n(n+3)/2)$-dimensional vector of multipliers corresponding to the constraints $\hat{\mat{x}} \geq 0_{n(n+3)/2}$.  

\vspace{8pt}

\newpage

\begin{lemma}
\label{lem:extend}
Let $(\mat{u}_r,\mat{v}_r,\mat{x}_r)$ be an optimal solution for the original system with respect to the active set $S_r$.  There is a feasible solution $(\hat{\mat{u}}_r,\hat{\mat{v}}_r,\hat{\mat{x}}_r)$ for the extended system with respect to the active set $S_r$ such that the following properties hold:
\begin{enumerate}[{\em (i)}]
	\item $\hat{\mat{u}}_r = \mat{u}_r$. \vspace*{2pt}
	\item $\hat{\mat{v}}_r = \begin{bmatrix} \mat{v}_r \\ \mat{u}_r \\ 0_n \end{bmatrix}$. \vspace*{4pt}
	\item $\hat{\mat{x}}_r = \begin{bmatrix} \mat{x}_r \\ \mat{s} \\ \mat{d} \end{bmatrix}$, where $\mat{s} = \mat{B}\mat{x}_r - 1_n + \mat{u}_r/M$ and $\mat{d} = 1_n + \mat{u}_r/M^2$. \vspace*{2pt}
	\item $\nabla \hat{f}(\hat{\mat{x}}_r) - \hat{\mat{B}}^t\hat{\mat{u}}_r = \hat{\mat{v}}_r$.  
	\item $\hat{f}(\hat{\mat{x}}_r) \geq f(\mat{x}_r)$ and $\hat{f}(\hat{\mat{x}}_r)$ approaches $f(\mat{x}_r)$ as $M$ increases.
\end{enumerate}
\end{lemma}
\begin{proof}
With respect to the active $S_r$, we claim that a feasible solution $(\hat{\mat{u}}_r,\hat{\mat{v}}_r,\hat{\mat{x}}_r)$ of the extended system can be obtained by solving the following linear system.   We use  $\hat{\mat{w}}_1$ to denote the first $n(n-1)/2$ coordinates of $\hat{\mat{v}}$ and $\hat{\mat{w}}_2$ the next $n$ coordinates of $\hat{\mat{v}}$.  
%Implicitly, we are setting the first $n(n-1)/2$ coordinates of $\hat{\mat{x}}$ to be $\mat{x}_r$.  
\begin{equation}
\begin{bmatrix}
	2\mat{A}^t\mat{Ax} \\
	0_n \\
	M^3(\mat{d}-1_n)
\end{bmatrix} 
-
\begin{bmatrix}
	\mat{B}^t\hat{\mat{u}} \\
	-\hat{\mat{u}} \\
	M\hat{\mat{u}}
\end{bmatrix}
=
\begin{bmatrix}
	\hat{\mat{w}}_1 \\
	\hat{\mat{w}}_2 \\
	0_n
\end{bmatrix},
\label{eq:extend-1}
\end{equation}
\cancel{
\[
\mat{B}^t\hat{\mat{u}} + \hat{\mat{w}}_1
= 
2\mat{A}^t\mat{A}\mat{x},
\quad
\hat{\mat{u}} - \hat{\mat{w}}_2 = 0_n,
\quad
M^2\mat{d} - \hat{\mat{u}} = M^21_n,
\]
}
\begin{equation}
\mat{B}\mat{x} - \mat{s} + M\mat{d}
= (1+M)1_n.
\label{eq:extend-2}
\end{equation}
Equation~\eqref{eq:extend-1} models the requirement of $\nabla \hat{f}(\hat{\mat{x}}) - \hat{\mat{B}}^t\hat{\mat{u}} = \hat{\mat{v}}$ for some particular settings of $\hat{\mat{x}}$, $\hat{\mat{u}}$, and $\hat{\mat{v}}$.  Equation~\eqref{eq:extend-2} models one of the primal feasibility constraints: $\hat{\mat{B}}\hat{\mat{x}} = (1+M)1_n$.

First, we argue that \eqref{eq:extend-1} and \eqref{eq:extend-2} can be satisfied simultaneously.  Observe that $2\mat{A}^t\mat{Ax} - \mat{B}^t\hat{\mat{u}} = \hat{\mat{w}}_1$ is a KKT condition of the original system, so it can be satisfied by setting $\mat{x} = \mat{x}_r$, $\hat{\mat{u}} = \mat{u}_r$ and $\hat{\mat{w}}_1 = \mat{v}_r$.  We set $\hat{\mat{w}}_2 = \mat{u}_r$ and $\mat{d} = 1_n + \mat{u}_r/M^2$.  These settings of $\hat{\mat{x}}$, $\hat{\mat{u}}$, $\hat{\mat{w}}_1$, $\hat{\mat{w}}_2$, and ${\mat{d}}$ satisfy~\eqref{eq:extend-1}.  To satisfy~\eqref{eq:extend-2}, we use $\mat{x}_r$ and the setting of $\mat{d}$ to set $\mat{s} = \mat{B}\mat{x}_r - 1_n + \mat{u}_r/M$.  The satisfiability of \eqref{eq:extend-1} and \eqref{eq:extend-2} is thus verified.  These settings yield the following candidate solution for the extended system and establishes the correctness of (i)--(iv):
\[
\hat{\mat{u}}_r = \mat{u}_r, \quad
\hat{\mat{v}}_r = \begin{bmatrix}
\mat{v}_r \\
\mat{u}_r \\
0_n
\end{bmatrix}, \quad
\hat{\mat{x}}_r = \begin{bmatrix}
	\mat{x}_r \\
	\mat{s} \\
	\mat{d}
\end{bmatrix}.
\]

We do not know whether $\hat{\mat{x}}_r$ is a feasible solution for the extended system yet.  We argue that $(\hat{\mat{u}}_r,\hat{\mat{v}}_r,\hat{\mat{x}}_r)$ satisfy the primal and dual feasibilities of the extended system.  Equation~\ref{eq:extend-2} explicitly models one of the primal feasibility conditions: $\hat{\mat{B}}\hat{\mat{x}} = (1+M)1_n$.  The other two feasibility conditions that need to be checked are $\hat{\mat{x}}_r \geq 0_{n(n+3)/2}$ and $(\hat{\mat{v}}_r)_j \geq 0$ for $j \not\in S_r$.  By the dual feasibility of the original system, we have $\mat{u}_r \geq 0_n$ and $(\mat{v}_r)_j \geq 0$ for $j \not\in S_r$, which implies that $(\hat{\mat{v}}_r)_j \geq 0$ for $j \not\in S_r$.  By the primal feasibility of the original system, we have $\mat{x}_r \geq 0_{n(n-1)/2}$.  Since $\mat{u}_r \geq 0_n$, we have $\mat{d} = 1_n + \mat{u}_r/M^2 \geq 1_n$.   As $\mat{B}\mat{x}_r \geq 1_n$ by the primal feasibility of the original system, we have $\mat{s} = \mat{B}\mat{x}_r - 1_n  + \mat{u}_r/M \geq 0_n$.  Hence, $\hat{\mat{x}}_r \geq 0_{n(n+3)/2}$.  

%The equations $\mat{B}^t\mat{u} + \hat{\mat{w}}_1 = 2\mat{A}^t\mat{Ax}$, $\hat{\mat{u}} - \hat{\mat{w}}_1 = 0_n$, and $M^2 \mat{d} - \hat{\mat{u}} = M^2 1_n$ as well as the setting of $(\hat{\mat{u}}_r,\hat{\mat{v}}_r,\hat{\mat{x}}_r)$ imply the correctness of (iv).  
Finally, $\hat{f}(\hat{\mat{x}}_r) = f(\mat{x}_r) + \norm{\mat{u}_r}^2/M$, which approaches $f(\mat{x}_r)$ as $M$ increases.  This proves (v).
\end{proof}

We remark that $(\hat{\mat{u}}_r,\hat{\mat{v}}_r,\hat{\mat{x}}_r)$ may not be an optimal solution for the extended system.  Indeed, one can check that $(\mat{s})_i \cdot (\mat{u}_r)_i$ may be positive for some $j \in [n]$, implying that $(\hat{\mat{u}}_r,\hat{\mat{v}}_r,\hat{\mat{x}}_r)$ may not satisfy complementary slackness.

We will analyze an imaginary run of SolveNNQ on this extended system starting from $(\hat{\mat{u}}_r,\hat{\mat{v}}_r,\hat{\mat{x}}_r)$.  Since no index in the range $[n(n-1)/2+1,n(n+3)/2]$ is put into an active set by SolveNNQ, the run on the extended system mimics what SolveNNQ does on the original system.

%Since $[n(n-1)/2+1,n(n+3)/2]$ is always disjoint from the active set, we have $(\hat{\mat{v}})_i \geq 0$ for $i \in [n(n-1)/2+1,n(n+3)/2]$ for all $i \in [n(n-1)/2+1,n(n+3)/2]$.  Notice that $\mat{d}$ is always close to $1_n$ due to the large penalty $M$, so $(\hat{\mat{v}})_i = 0$ for $i \in [n(n+1)/2+1,n(n+3)/2]$.

\subsection{A descent direction for the extended system}

Let $(\hat{\mat{u}}_r,\hat{\mat{v}}_r,\hat{\mat{x}}_r)$ be a feasible solution of the extended solution corresponding to $(\mat{u}_r,\mat{v}_r,\mat{x}_r)$ as promised by Lemma~\ref{lem:extend}.   Recall that, in the pseudocode of SolveNNQ, $E_r$ is a sequence of indices $i \in S_r$ such that $(\mat{v}_r)_i < 0$.   Without loss of generality, we assume that $(\mat{v}_r)_1 = \min_{i \in E_r} (\mat{v}_r)_i$ and that the edge with index 1 is $\mat{p}_1\mat{p}_2$.   

In an imaginary run of SolveNNQ on the extended system starting from $(\hat{\mat{u}}_r,\hat{\mat{v}}_r,\hat{\mat{x}}_r)$, we will free the variable corresponding to the first coordinate of $\hat{\mat{x}}_r$.  So we will increase the variable $(\hat{\mat{x}}_r)_1$, which means that the descent direction is some appropriate component of the vector $\mat{e}_1$ in $\real^{n(n+3)/2}$.  We introduce some definitions in order to characterize the descent direction.

\vspace{6pt}

\begin{definition}
Define the following set of indices
\[
\core(\hat{\mat{x}}_r) = \supp(\mat{x}_r) \cup \left[\frac{n(n+1)}{2}+1, \frac{n(n+3)}{2}\right] \cup \left\{a + \frac{n(n-1)}{2}: (\mat{B}\mat{x}_r)_a > 1\right\}.
\]
The role of $\core(\hat{\mat{x}}_r)$ for the extended system is analogous to the role of $\supp(\mat{x}_r)$ for the original system.
All dummy variables belong to $\core(\hat{\mat{x}}_r)$.  A slack variable is in $\core(\hat{\mat{x}}_r)$ if and only if the correponding constraint in $\mat{B}\mat{x}_r \geq 1_n$ is not tight.
\end{definition}

\vspace{6pt}

\begin{definition}
\label{df:I}
For all $a \in [n]$,  let $I_a = \bigl\{ i \in \core(\hat{\mat{x}}_r) : (\hat{\mat{B}})_{a,i} \not= 0 \bigr\}$.  So $I_a \subseteq \core(\hat{\mat{x}}_r)$ and it tells us the entries in the $a$-th constraint in $\hat{\mat{B}}$ that are affected if $(\mat{x}_r)_1$ is increased.  Note that $I_a$ always contains $n(n+1)/2+a$.
\end{definition}

\vspace{6pt}

\begin{definition}
For all $a \in [n]$, define the following quantities:
	\begin{align*}
		c_a = & \left\{\begin{array}{ccl}
			1 + M, & \quad\quad & \text{if $(\mat{B}\mat{x}_r)_a > 1,$ {\em i.e., $a + n(n-1)/2\in \core(\hat{x}_r)$;}} \\
			1 + M - (\mat{u}_r)_a/M, && \text{if $(\mat{B}\mat{x}_r)_a = 1,$ {\em i.e., $a + n(n-1)/2 \not\in \core(\hat{\mat{x}}_r)$}}.
		\end{array}\right.  \\[.25em]
		H_a = &\left\{ \hat{\mat{x}} \in \real^{n(n+3)/2} : (\hat{\mat{B}}\hat{\mat{x}})_a = c_a \right\}\, \bigcap \,
		\mathrm{span}\Bigl(\bigl\{\mat{e}_i : i \in \{1\} \cup \core(\hat{\mat{x}}_r)\bigr\}\Bigr).  
	\end{align*}
In the definition of $H_a$, since we set all variables not in $\{1\} \cup \core(\hat{\mat{x}}_r)$ to zero, the slack variables that are excluded must be subtracted from the vector $(1+M)1_n$ on the right hand side of the constraints $\hat{\mat{B}}\hat{\mat{x}} = (1+M)1_n$.  By Lemma~\ref{lem:extend}, for any $a \in [n]$, if $a + n(n-1)/2 \not\in \core(\hat{\mat{x}}_r)$, the value of the corresponding slack variable is $(\mat{u}_r)_a/M$.  This consideration justifies the definition of $c_a$.
\end{definition}

\vspace{6pt}

\begin{definition}
\label{df:h}
Let $h = \bigcap_{a \in [n]} H_a$.  It is a non-empty subspace of $\mathrm{span}\bigl(\bigl\{\mat{e}_i : i \in \{1\} \cup \core(\hat{\mat{x}}_r)\bigr)$ because it contains the projection of $\hat{\mat{x}}_r$ to $\mathrm{span}\bigl(\bigl\{\mat{e}_i : i \in \{1\} \cup \core(\hat{\mat{x}}_r) \bigr\}\bigr)$.  
%Its dimension is thus equal to $\bigl|\supp(\hat{\mat{x}}_r) \cap [n(n+1)/2]\bigr| +1 - n$.  
%We assume that the dimension of $h$ is positive which will be the case for DKSG.  
%So $h$ is non-empty and has dimension $|\supp(\mat{x}_r)| + 1 - \kappa$. 
\end{definition}

\vspace{6pt}

%\newpage

\begin{definition}
\label{df:k}
For all $a \in [n]$, define a vector $\hat{\mat{k}}_a$ in $\real^{n(n+3)/2}$ as follows:
\[
(\hat{\mat{k}}_a)_i =  \left\{\begin{array}{ccl}
	(\hat{\mat{B}})_{a,i}, & \quad\quad & \text{{\em if $i \in \{1\} \cup I_a$;}} \\[.25em]
	0, & & \text{{\em otherwise}}.
\end{array}\right.
\]
Observe that $\hat{\mat{k}}_a \perp H_a$.   Also, if $a \in \{1,2\}$, then $(\hat{\mat{k}}_a)_1 = 1$; otherwise, $(\hat{\mat{k}}_a)_1 = 0$.  It is because the column index $1$ is the index of edge $\mat{p}_1\mat{p}_2$ by assumption.
\end{definition}

\vspace{6pt}

Note that every $\hat{\mat{k}}_a$ is orthogonal to $h$.   The vectors $\{\hat{\mat{k}}_a : a \in [n]\}$ and the vectors that are parallel to $h$ span the linear subspace $\mathrm{span}\bigl(\bigl\{\mat{e}_i : i \in \{1\} \cup \core(\hat{\mat{x}}_r) \bigr\}\bigr)$.  

We want to analyze the effect of using $\mat{e}_1 \!\downarrow\! h$ as a descent direction.  The direction $\mat{e} \!\downarrow\! h$ is equal to $\mat{e}_1 - \mat{e}_1 \!\downarrow\! \mathrm{span}(\{\hat{\mat{k}}_a : a \in [n]\})$.  However, the vectors $\{\hat{\mat{k}}_a : a \in [n]\}$ are not orthonormal, so we first orthonormalize $\{\hat{\mat{k}}_a : a \in [n]\}$ using the Gram-Schmidt method~\cite{GV1996} in order to express $\mat{e}_1 \!\downarrow\! \mathrm{span}(\{\hat{\mat{k}}_a : a \in [n]\})$.  The method first generates $n$ mutually orthogonal vectors so that we can normalize each individually later.  

First, define $\check{\mat{k}}_n = \hat{\mat{k}}_n$.  For $a = n-1, n-2, \ldots, 3, 2, 1$ in this order, define
\[
\check{\mat{k}}_a = \hat{\mat{k}}_a - \sum_{b = a+1}^{n} \frac{\langle \hat{\mat{k}}_a,\check{\mat{k}}_b \rangle}{\langle \check{\mat{k}}_b, \check{\mat{k}}_b \rangle} \check{\mat{k}}_b.
\]
This process generates $\check{\mat{k}}_a$ by subtracting the component of $\hat{\mat{k}}_a$ that lies in $\mathrm{span}(\{\hat{\mat{k}}_b : b \in [a+1,n]\})$.  
\cancel{
The Gram-Schmidt process works irrespective of the processing order.  We can do the same for 2 and 1 as well, but for a technical convenience, we define
\[
\check{\mat{k}}_2 = (\hat{\mat{k}}_2 - \hat{\mat{k}}_1) - \sum_{b = 3}^{n} \frac{\langle \hat{\mat{k}}_2 - \hat{\mat{k}}_1,\check{\mat{k}}_b \rangle}{\langle \check{\mat{k}}_b, \check{\mat{k}}_b \rangle} \check{\mat{k}}_b.
\]
The above definition is equivalent to replacing $\hat{\mat{k}}_2$ by $\hat{\mat{k}}_2 - \hat{\mat{k}}_1$, which does not change the span of $\{\hat{\mat{k}}_a : a \in [n]\}$.  Finally, we define
\[
\check{\mat{k}}_1 = \hat{\mat{k}}_1 - \sum_{b = 2}^{n} \frac{\langle \hat{\mat{k}}_1,\check{\mat{k}}_b \rangle}{\langle \check{\mat{k}}_b, \check{\mat{k}}_b \rangle} \check{\mat{k}}_b.
\]
}
The set of vectors $\{\check{\mat{k}}_a : a \in [n]\}$ are mutually orthogonal.  Then, we define the following set of orthonormal vectors:
\[
\left\{\mat{k}_a = \frac{\check{\mat{k}}_a}{\norm{\check{\mat{k}}_a}} : a \in [n]\right\}.
\]

\vspace{6pt}

\begin{lemma}
\label{lem:k}
The following properties are satisfied.
\begin{enumerate}[{\em (i)}]
	\item For all $a \in [n]$, $(\check{\mat{k}}_a)_{n(n+1)/2 + a} = M$.
	\item $(\check{\mat{k}}_2)_1 = 1$, and for $a \in [3,n]$, $(\check{\mat{k}}_a)_1 = 0$.
	%\item For every $a \in [n]$ and every $i \in E_r \setminus \{1\}$, $(\check{\mat{k}}_a)_i = 0$.
	\item For all $a \in [n]$ and $i \in [n(n+1)/2]$, $-1-\frac{1}{M} \leq (\check{\mat{k}}_a)_i \leq 1+ \frac{1}{M}$ for $M \geq n^2(n-1)$.
	%\item $(\check{\mat{k}}_2)_{n(n+1)/2+1} = -m_r$ and $(\check{\mat{k}}_2)_{n(n+1)/2+2} = m_r$.
	%\item $\bigl|(\check{\mat{k}}_1)_{n(n+1)/2+1}\bigr| + \bigl|(\check{\mat{k}}_1)_{n(n+1)/2+2}\bigr|\geq m_r$.
\end{enumerate}
\end{lemma}
\begin{proof}
Observe that for $a \in [n]$, $(\hat{\mat{k}}_b)_{n(n+1)/2+a} = 0$ for all $b > a$.  This implies that $(\check{\mat{k}}_b)_{n(n+1)/2+a} = 0$ for all $b > a$ and hence $(\check{\mat{k}}_a)_{n(n+1)/2+a} = M$.  This proves (i).
	
Similarly, for $a \in [3,n]$, we have $(\hat{\mat{k}}_a)_1 = 0$, which imples that $(\check{\mat{k}}_a)_1 = 0$.  Then, coupled with the fact that $(\hat{\mat{k}}_2)_1 = 1$, we get $(\check{\mat{k}}_2)_1 = 1$.   This proves (ii).

We show (iii) inductively.  The base case is trivial as $\check{\mat{k}}_n = \hat{\mat{k}}_n$.  Consider the computation of $\check{\mat{k}}_a$.  We can argue as in proving (i) that $(\check{\mat{k}}_b)_{n(n+1)/2+a} = 0$ for all $b > a$.  Also, $(\hat{\mat{k}}_a)_{n(n+1)/2+b} = 0$ for all $b \not= a$.  Therefore, the inner product $\langle \hat{\mat{k}}_a,\check{\mat{k}}_b\rangle$ for any $b > a$ does not receive any contributions from the entries with indices in the range $[n(n+1)/2+1,n(n+3)/2]$.  By a similar reasoning, there are no contributions from the entries with indices in the range $[n(n-1)/2+1,n(n+1)/2]$.  It follows that, for any $b > a$, $\bigl|\langle \hat{\mat{k}}_a,\check{\mat{k}}_b\rangle/\langle \check{\mat{k}}_b,\check{\mat{k}}_b\rangle \bigr| \leq \sum_{i=1}^{n(n-1)/2} |(\check{\mat{k}}_b)_i|/\langle \check{\mat{k}}_b,\check{\mat{k}}_b\rangle \leq n(n-1)/\langle \check{\mat{k}}_b,\check{\mat{k}}_b\rangle$ as the largest magnitude of the first $n(n-1)/2$ entries of $\hat{\mat{k}}_a$ is 1, and the largest magnitude of the first $n(n-1)/2$ entries of $\check{\mat{k}}_b$ is at most $1 + \frac{1}{M} \leq 2$ by induction assumption.  By~(i), $\langle \check{\mat{k}}_b,\check{\mat{k}}_b\rangle \geq M^2$.  As a result, the largest magnitude of the first $n(n+1)/2$ entries of $\sum_{b=a+1}^n \check{\mat{k}}_b \cdot \langle \hat{\mat{k}}_a,\check{\mat{k}}_b\rangle/\langle \check{\mat{k}}_b,\check{\mat{k}}_b\rangle$ is at most $n \cdot (1+1/M) \cdot n(n-1)/M^2 \leq 1/M$ by induction assumption and the assumption that $M \geq n^2(n-1)$. Therefore, among the first $n(n+1)/2$ entries of $\check{\mat{k}}_a$, the smallest entry is at least $-1-1/M$ and the largest entry is at most $1 + 1/M$.   This proves (iii).
\cancel{
Since $(\hat{\mat{k}}_1)_{n(n+1)/2+2}$ and $(\hat{\mat{k}}_2)_{n(n+1)/2+1}$ are zeros and $(\hat{\mat{k}}_1)_{n(n+1)/2+1}$ and $(\hat{\mat{k}}_2)_{n(n+1)/2+2}$ are equal to $m_r$, the vector $\hat{\mat{k}}_2-\hat{\mat{k}}_1$ has an entry $-m_r$ at position $n(n+1)/2+1$ and an entry $m_r$ at position $n(n+1)/2+2$.  Observe that for all $a \in [3,n]$, $\check{\mat{k}}_a$ has zeros at positions $n(n+1)/2+1$ and $n(n+1)/2+2$.  It follows that $\check{\mat{k}}_2$ has an entry $-m_r$ at position $n(n+1)/2+1$ and an entry $n^2$ at position $n(n+1)/2+2$.  This proves (iv).

The entries of $\check{\mat{k}}_1$ at positions $n(n+1)/2+1$ and $n(n+2)/2+2$ are  equal to $m_r + m_r\langle \hat{\mat{k}}_1,\check{\mat{k}}_2 \rangle/\langle \check{\mat{k}}_2,\check{\mat{k}}_2 \rangle$ and $-m_r\langle \hat{\mat{k}}_1,\check{\mat{k}}_2 \rangle/\langle \check{\mat{k}}_2,\check{\mat{k}}_2 \rangle$, respectively.  It is straightforward to verify that the sum of the magnitudes of these entries is at least $m_r$.
}
\end{proof}

%Observe that $\mathrm{span}(\{\mat{k}_a : a \in [\kappa]\}) \subseteq \mathrm{span}(\{\mat{e}_i : i \in \{1\} \cup \supp(\mat{x}_r)\})$.  As $h \subseteq \mathrm{span}(\{\mat{e}_i : i \in \{1\} \cup \supp(\mat{x}_r)\})$ and $h$ has dimension $|\supp(\mat{x}_r)| + 1 - \kappa$, the subspace $\mathrm{span}(\{\mat{k}_a : a \in [\kappa]\})$ is the orthogonal complement of $h$ in $\mathrm{span}(\{\mat{e}_i : i \in \{1\} \cup \supp(\mat{x}_r)\})$.  

\vspace{6pt}

Let $\hat{\mat{n}}_r = \mat{e}_1 \!\downarrow\! h/\norm{\mat{e}_1 \!\downarrow\! h}$.  We analyze the value of $\langle \nabla \hat{f}(\hat{\mat{x}}_r),\hat{\mat{n}}_r \rangle$ and show that $\hat{\mat{n}}_r$ is a descent direction.

\vspace{6pt}

\begin{lemma}
\label{lem:descent}
For a large enough $M$, $\langle \nabla \hat{f}(\hat{\mat{x}}_r),\hat{\mat{n}}_r \rangle \leq \frac{2}{3}(\mat{v}_r)_1 < 0$.
\end{lemma}
\begin{proof}
We first obtain a simple expression for $\mat{e}_1 \!\downarrow\! h$.  Since the vectors $\{\mat{k}_a : a \in [n]\}$ are orthonormal, we have $\mat{e}_1 \!\downarrow\! h = \mat{e}_1 - \sum_{a=1}^n \langle \mat{e}_1,\mat{k}_a \rangle \mat{k}_a$.  The only non-zero entry of $\mat{e}_1$ is a 1 at position 1.  By Lemma~\ref{lem:k}, $(\mat{k}_a)_1 = 0$ for $a \in [3,n]$.   Note that $\langle \mat{e}_1,\mat{k}_1\rangle = (\check{\mat{k}}_1)_1/\norm{\check{\mat{k}}_1}$ and by Lemma~\ref{lem:k}(ii), $\langle \mat{e}_1,\mat{k}_2\rangle = 1/\norm{\check{\mat{k}}_2}$.  Therefore,
\begin{equation}
\mat{e}_1 \!\downarrow\! h = \mat{e}_1 - \frac{(\check{\mat{k}}_1)_1}{\norm{\check{\mat{k}}_1}} \mat{k}_1 - \frac{1}{\norm{\check{\mat{k}}_2}} \mat{k}_2.
\label{eq:ana-01}
\end{equation}

Let $K_r  = \mathrm{span}\bigl(\bigl\{\mat{e}_i : i \in E_r  \cup \core(\hat{\mat{x}}_r) \bigr\}\bigr)$.  Therefore, for all $\mat{y} \in K_r$, if $i \not\in E_r \cup \core(\hat{\mat{x}}_r)$, then $(\mat{y})_i = 0$.  We claim that for a sufficiently large $M$,
\begin{equation}
	\langle \hat{\mat{v}}_r\!\downarrow\!K_r, \mat{e}_1 \!\downarrow\! h \rangle \leq \frac{2}{3}(\mat{v}_r)_1.
	\label{eq:ana-2}
\end{equation}
The proof goes as follows.  By \eqref{eq:ana-01}, 
\[
\langle \hat{\mat{v}}_r\!\downarrow\!K_r, \mat{e}_1 \!\downarrow\! h \rangle = \langle \hat{\mat{v}}_r\!\downarrow\!K_r, \mat{e}_1 \rangle - \frac{(\check{\mat{k}}_1)_1}{\norm{\check{\mat{k}}_1}}\langle \hat{\mat{v}}_r\!\downarrow\!K_r, \mat{k}_1 \rangle
- \frac{1}{\norm{\check{\mat{k}}_2}}\langle \hat{\mat{v}}_r\!\downarrow\!K_r, \mat{k}_2 \rangle.  
\]
Clearly, $\langle \hat{\mat{v}}_r\!\downarrow\!K_r, \mat{e}_1 \rangle = (\hat{\mat{v}}_r)_1$ which is equal to $(\mat{v}_r)_1$ by Lemma~\ref{lem:extend}.  Recall that $(\mat{v}_r)_1$ is the most negative entry of $\mat{v}_r$ by assumption.  For $i > n(n+1)/2$, if $i \not\in E_r \cup \core(\hat{\mat{x}}_r)$, then $(\hat{\mat{v}}_r \!\downarrow\! K_r)_i = 0$ by the defintion of $K_r$; otherwise,  $(\hat{\mat{v}}_r \!\downarrow\! K_r)_i = (\hat{\mat{v}}_r)_i$ which is zero by Lemma~\ref{lem:extend}.   Take any $i \leq n(n+1)/2$.  By Lemma~\ref{lem:k}(i) and (iii), $|(\mat{k}_1)_i|$ and $|(\mat{k}_2)_i|$ are at most $1/M + 1/M^2$.  Therefore, $(\hat{\mat{v}}_r\!\downarrow\!K_r)_i \cdot (\mat{k}_1)_i$ and $(\hat{\mat{v}}_r\!\downarrow\!K_r)_i \cdot (\mat{k}_2)_i$ have negligible magnitude for a large enough $M$.  It follows that $\langle \hat{\mat{v}}_r\!\downarrow\!K_r, \mat{e}_1 \!\downarrow\! h \rangle$ is no more than $({\mat{v}}_r)_1$ plus a non-negative number of negligible magnitude in the worst case.  The result is thus less than $\frac{2}{3}({\mat{v}}_r)_1$.  This completes the proof of \eqref{eq:ana-2}.

For all $a \in [n]$, the $a$-th constraint in the extended system is a hyperplane in $\real^{n(n+3)/2}$ for which the $a$-th row $\hat{\mat{B}}_{a,*}$ of $\hat{\mat{B}}$ is a normal vector.  As a result, for any direction $\mat{w}$ from the feasible solution $\hat{\mat{x}}_r$ that is parallel to the hyperplanes encoded by $\hat{\mat{B}}\hat{\mat{x}} = (1+M)1_n$, the vector $\hat{\mat{B}}^t\hat{\mat{u}}_r$ is orthogonal to $\mat{w}$, which gives 
\begin{equation}
	\langle \nabla \hat{f}(\hat{\mat{x}}_r),\mat{w} \rangle = \langle \nabla \hat{f}(\hat{\mat{x}}_r) - \hat{\mat{B}}^t\hat{\mat{u}}_r,\mat{w} \rangle  = \langle \hat{\mat{v}}_r,\mat{w} \rangle.  \label{eq:ana-0}
\end{equation}
The last step uses the relation of $\nabla \hat{f}(\hat{\mat{x}}_r) - \hat{\mat{B}}^t\hat{\mat{u}}_r = \hat{\mat{v}}_r$ from Lemma~\ref{lem:extend}(iv).

As we move from $\hat{\mat{x}}_r$ in direction $\mat{e}_1\!\downarrow\!h$, our projection in $\mat{e}_1$ moves in the positive direction, and so $(\hat{\mat{x}})_1$ increases.  During the movement, for any $i \not\in \{1\}\cup\core(\hat{\mat{x}}_r)$, the coordinate $(\hat{\mat{x}})_i$ remains unchanged by the definition of $h$.  Specifically, if $i \leq n(n-1)/2$, then $(\hat{\mat{x}}_r)_i$ remains at zero; if $i \in [n(n-1)/2+1,n(n+1)/2]$, $(\hat{\mat{x}}_r)_i$ is fixed at $(\hat{\mat{u}}_r)_a/M$, where $a = i - n(n-1)/2$.  Some coordinates among $\bigl\{(\hat{\mat{x}})_i : i \in \core(\hat{\mat{x}}_r) \bigr\}$, which are positive, may decrease.  The feasibility constraints are thus preserved for some extent of the movement.  By~\eqref{eq:ana-0}, $\langle \nabla \hat{f}(\hat{\mat{x}}_r),\mat{e}_1\!\downarrow\!h\rangle = \langle \hat{\mat{v}}_r,\mat{e}_1\!\downarrow\!h\rangle$, which is equal to $\langle \hat{\mat{v}}_r \!\downarrow\!K_r,\mat{e}_1\!\downarrow\!h\rangle$ as $\mat{e}_1\!\downarrow\!h \in K_r$.  By \eqref{eq:ana-2}, $\langle \nabla \hat{f}(\hat{\mat{x}}_r),\mat{e}_1\!\downarrow\!h\rangle \leq \frac{2}{3}(\mat{v}_r)_1 < 0$.  Clearly, $\norm{\mat{e}_1 \!\downarrow\! h} \leq 1$.  So $\langle \nabla \hat{f}(\hat{\mat{x}}_r),\hat{\mat{n}}_r \rangle \leq \langle \nabla \hat{f}(\hat{\mat{x}}_r),\mat{e}_1\!\downarrow\!h\rangle \leq \frac{2}{3}(\mat{v}_r)_1 < 0$.
\end{proof}

\subsection{Convergence}

%We discuss a conditional convergence result for DKSG.  
%Consider the original system.  From Lemma~\ref{lem:bound2}, the natural goal is to identify a descent direction $\mat{n}_r$ from the current solution $\mat{x}_r$ and prove lower bounds for $\norm{\mat{y} - \mat{x}_r}/{\norm{\mat{x}_r - \mat{x}_*}}$ and $\langle \nabla f(\mat{x}_r), \mat{n}_r \rangle/{{\langle \nabla f(\mat{x}_r), \mat{n}_* \rangle}}$, where $\mat{y}$ is the feasible solution that minimizes $f$ in direction $\mat{n}_r$ from $\mat{x}_r$, $\mat{x}_*$ is the optimal DKSG solution, and $\mat{n}_* = (\mat{x}_* - \mat{x}_r)/\norm{\mat{x}_* - \mat{x}_r}$.  
%
Our analysis takes a detour via the extended system.  Since the extended system is also a non-negative quadratic programming problem, Lemma~\ref{lem:bound2} is applicable.  The natural goal is thus to use the descent direction $\hat{\mat{n}}_r$ from $\hat{\mat{x}}_r$ and prove lower bounds for $\norm{\hat{\mat{y}} - \hat{\mat{x}}_r}/{\norm{\hat{\mat{x}}_r - \hat{\mat{x}}_*}}$ and $\langle \nabla \hat{f}(\hat{\mat{x}}_r), \hat{\mat{n}}_r \rangle/{{\langle \nabla \hat{f}(\hat{\mat{x}}_r), \hat{\mat{n}}_* \rangle}}$, where $\hat{\mat{y}}$ is the feasible point that minimizes $\hat{f}$ in direction $\hat{\mat{n}}_r$ from $\hat{\mat{x}}_r$, $\hat{\mat{x}}_*$ is the optimal solution of the extended sytsem, and $\hat{\mat{n}}_* = (\hat{\mat{x}}_* - \hat{\mat{x}}_r)/\norm{\hat{\mat{x}}_* - \hat{\mat{x}}_r}$.  

We prove a lower bound for $\langle \nabla \hat{f}(\hat{\mat{x}}_r), \hat{\mat{n}}_r \rangle/{{\langle \nabla \hat{f}(\hat{\mat{x}}_r), \hat{\mat{n}}_* \rangle}}$. As in the case of NNLS and ZHLG, we do not know a lower bound for $\norm{\hat{\mat{y}} - \hat{\mat{x}}_r}/{\norm{\hat{\mat{x}}_r - \hat{\mat{x}}_*}}$.  Therefore, we derive a convergence result based the assumption that $\norm{\hat{\mat{y}} - \hat{\mat{x}}_r}/{\norm{\hat{\mat{x}}_r - \hat{\mat{x}}_*}}$ is bounded from below.
 
Let $S$ be an active set that is a subset of $[n(n-1)/2]$ disjoint from $\{1\} \cup \core(\hat{\mat{x}}_r)$.  So $\hat{\mat{n}}_r$ is a descent direction from $\hat{\mat{x}}_r$ with respect to $S$.  Let $\hat{\mat{x}}$ be any feasible solution of the extended system with respect to $S$ such that $\hat{f}(\hat{\mat{x}}) < \hat{f}(\hat{\mat{x}}_r)$.  Let $\hat{\mat{n}}$ be the unit vector from $\hat{\mat{x}}_r$ towards $\hat{\mat{x}}$.  We bound the ratio $\langle \nabla \hat{f}(\hat{\mat{x}}_r),\hat{\mat{n}}_r \rangle/ \langle \nabla \hat{f}(\hat{\mat{x}}_r), \hat{\mat{n}} \rangle$ from below.  Note that $\hat{\mat{n}}$ is a descent direction from $\hat{\mat{x}}_r$ with respect to $S$.  Therefore, both $\langle \nabla \hat{f}(\hat{\mat{x}}_r),\hat{\mat{n}}_r \rangle$ and $\langle \nabla \hat{f}(\hat{\mat{x}}_r), \hat{\mat{n}} \rangle$ are negative, and the ratio is thus positive.

\vspace{6pt}

\begin{lemma}
\label{lem:ana-ratio}
 Let $\hat{\mat{x}}$ be any feasible solution of the extended system with respect to an active set $S$ that is a subset of $[n(n-1)/2]$ disjoint from $\{1\} \cup \core(\hat{\mat{x}}_r)$.   Assume that $\hat{f}(\hat{\mat{x}}) < \hat{f}(\hat{\mat{x}}_r)$.  Let $\hat{\mat{n}} = (\hat{\mat{x}} - \hat{\mat{x}}_r)/\norm{\hat{\mat{x}} - \hat{\mat{x}}_r}$.   For a sufficiently large $M$,
 \[
 \frac{\langle \nabla \hat{f}(\hat{\mat{x}}_r),\hat{\mat{n}}_r \rangle}{\langle \nabla \hat{f}(\hat{\mat{x}}_r), \hat{\mat{n}} \rangle} \geq \frac{1}{\sqrt{2n(n-1)}}.
 \]
\end{lemma}
\begin{proof}
As in the proof of Lemma~\ref{lem:descent}, let $K_r  = \mathrm{span}\bigl(\bigl\{\mat{e}_i : i \in E_r  \cup \core(\hat{\mat{x}}_r) \bigr\}\bigr)$.  Therefore, for all $\mat{y} \in K_r$, if $i \not\in E_r \cup \core(\hat{\mat{x}}_r)$, then $(\mat{y})_i = 0$.   We claim that there exists a non-negative $\delta \geq 0$ that approaches zero as $M$ increases such that
\begin{equation}
\langle -\hat{\mat{v}}_r,\hat{\mat{x}} -\hat{\mat{x}}_r \rangle \leq \langle -\hat{\mat{v}}_r\!\downarrow\!K_r, \hat{\mat{x}} - \hat{\mat{x}}_r \rangle + \delta.
\label{eq:ana-ratio}
\end{equation}
Since $\langle -\hat{\mat{v}}_r,\hat{\mat{x}} -\hat{\mat{x}}_r \rangle =  \langle -\hat{\mat{v}}_r\!\downarrow\!K_r, \hat{\mat{x}} - \hat{\mat{x}}_r \rangle + \langle -\hat{\mat{v}}_r + \hat{\mat{v}}_r\!\downarrow\!K_r, \hat{\mat{x}} - \hat{\mat{x}}_r  \rangle$, it suffices to prove that $\langle -\hat{\mat{v}}_r + \hat{\mat{v}}_r\!\downarrow\!K_r, \hat{\mat{x}} - \hat{\mat{x}}_r \rangle \leq \delta$.  

If $i \in \core(\hat{\mat{x}}_r) \cup E_r$, then $(\hat{\mat{v}}_r\!\downarrow\!K_r)_i = (\hat{\mat{v}}_r)_i$;  otherwise, $(\hat{\mat{v}}_r\!\downarrow\!K_r)_i = 0$.  Therefore, if $i \in \core(\hat{\mat{x}}_r) \cup E_r$, then $(-\hat{\mat{v}}_r + \hat{\mat{v}}_r\!\downarrow\!K_r)_i = 0$; otherwise, by Lemma~\ref{lem:extend}, $(-\hat{\mat{v}}_r + \hat{\mat{v}}_r\!\downarrow\!K_r)_i = -(\hat{\mat{v}}_r)_i \leq 0$.   
	%So it suffices to consider the indices $i \not\in  \supp(\hat{\mat{x}}_r) \cup E_r$.
	%
	%Consider any 

Take any $i \not\in  \core(\hat{\mat{x}}_r) \cup E_r$.  If $i \in [n(n-1)/2]$, then $i \not\in \supp(\mat{x}_r)$ and $(\hat{\mat{x}}_r)_i = (\mat{x}_r)_i = 0$ by Lemma~\ref{lem:extend}, which implies that $(\hat{\mat{x}} - \hat{\mat{x}}_r)_i \geq 0$ because $\hat{\mat{x}} \geq 0_{n(n+3)/2}$ by the primal feasibility constraint.  If $i \not\in [n(n-1)/2]$, then $i \in [n(n-1)/2+1,n(n+1)/2]$, so $(\hat{\mat{x}}_r)_i = (\mat{u}_r)_a/M$, where $a = i - n(n-1)/2$, which implies that $(\hat{\mat{x}} - \hat{\mat{x}}_r)_i \geq -(\mat{u}_r)_a/M$.  

We conclude that either $(-\hat{\mat{v}}_r + \hat{\mat{v}}_r\!\downarrow\!K_r)_i  \cdot (\hat{\mat{x}} - \hat{\mat{x}}_r)_i$ is negative, or it is at most some negligible postive value for a large enough $M$.  As a result, $\langle -\hat{\mat{v}}_r + \hat{\mat{v}}_r\!\downarrow\!K_r, \hat{\mat{x}} - \hat{\mat{x}}_r \rangle \leq \delta$.  This completes the proof of \eqref{eq:ana-ratio}.

By \eqref{eq:ana-0}, Lemma~\ref{lem:descent}, and~\eqref{eq:ana-ratio}, we have 
\[
\frac{\langle \nabla  \hat{f}(\hat{\mat{x}}_r), \, \hat{\mat{n}}_r\rangle}{\langle \nabla \hat{f}(\hat{\mat{x}}_r), \, \hat{\mat{n}} \rangle} =
\frac{-\langle \nabla \hat{f}(\hat{\mat{x}}_r), \, \hat{\mat{n}}_r\rangle}{-\langle \hat{\mat{v}}_r , \hat{\mat{n}} \rangle} \geq 
\frac{-2(\mat{v}_r)_1}{-3\langle \hat{\mat{v}}_r , \hat{\mat{n}} \rangle} \geq
\frac{-2(\mat{v}_r)_1}{-3\langle \hat{\mat{v}}_r \!\downarrow\! K_r , \hat{\mat{n}} \rangle + 3\delta/\norm{\hat{\mat{x}}-\hat{\mat{x}}_r}}.
\]
Observe that $-\langle \hat{\mat{v}}_r \!\downarrow\! K_r , \hat{\mat{n}} \rangle$ is positive because, by~\eqref{eq:ana-0} and~\eqref{eq:ana-ratio}, $-\langle \hat{\mat{v}}_r \!\downarrow\! K_r , \hat{\mat{n}} \rangle \geq -\langle \hat{\mat{v}}_r , \hat{\mat{n}} \rangle - \delta/\norm{\hat{\mat{x}}-\hat{\mat{x}}_r} = -\langle \nabla \hat{f}(\hat{\mat{x}}_r), \, \hat{\mat{n}} \rangle - \delta/\norm{\hat{\mat{x}}-\hat{\mat{x}}_r}$ which is positive for a large enough $M$.  As $-\langle \hat{\mat{v}}_r \!\downarrow\! K_r , \hat{\mat{n}} \rangle = \langle \hat{\mat{v}}_r \!\downarrow\! K_r , -\hat{\mat{n}} \rangle$ and $-\hat{\mat{n}}$ is a unit vector, we have $\langle \hat{\mat{v}}_r \!\downarrow\! K_r , -\hat{\mat{n}} \rangle \leq \norm{\hat{\mat{v}}_r \!\downarrow\! K_r}$.  We conclude that for a large enough $M$, we have
\begin{equation}
\frac{\langle \nabla \hat{f}(\hat{\mat{x}}_r), \, \hat{\mat{n}}_r\rangle}{\langle \nabla \hat{f}(\hat{\mat{x}}_r), \, \hat{\mat{n}} \rangle} \geq 
\frac{-(\mat{v}_r)_1}{2\norm{\hat{\mat{v}}_r \!\downarrow\! K_r}}.
\label{eq:ana-ratio-2}
\end{equation}
We analyze $\norm{\hat{\mat{v}}_r \!\downarrow\! K_r}$.  Take any $i \in \{1\}\cup \core(\hat{\mat{x}}_r)$.  There are three cases.
\begin{itemize}
	\item $i \in [n(n-1)/2]$.  Then $i \in E_r \cup \supp(\mat{x}_r)$ and $(\hat{\mat{v}}_r)_i = (\mat{v}_r)_i$ by Lemma~\ref{lem:extend}, which means that $(\hat{\mat{v}}_r)_i = (\mat{v}_r)_i$ contributes to $\norm{\hat{\mat{v}}_r \!\downarrow\! K_r}$ if and only if $i \in E_r$.
	\item $i \in [n(n-1)/2+1,n(n+1)/2]$.  Then, $i$ must belong to $\core(\hat{\mat{x}}_r)$, which implies that $(\mat{B}\mat{x}_r)_a > 1$, where $a = i - n(n-1)/2$.  By the complementary slackness of the original system, $(\mat{u}_r)_a = 0$ which implies that $(\hat{\mat{v}}_r)_i = (\mat{u}_r)_a = 0$ by Lemma~\ref{lem:extend}.
	\item $i > n(n+1)/2$.  Then, $(\hat{\mat{v}}_r)_i = 0$ by Lemma~\ref{lem:extend}.
\end{itemize}
We conclude that $\norm{\hat{\mat{v}}_r \!\downarrow\! K_r}^2  = \sum_{i \in E_r} (\hat{\mat{v}}_r)_i^2 \leq n(n-1)/2 \cdot (\hat{\mat{v}}_r)_1^2$.  Substituting into \eqref{eq:ana-ratio-2} gives the lemma.
\end{proof}

Let $S_{r+1}$ be the active set used by SolveNNQ for the next iteration on the original system.  Note that $S_{r+1}$ is a subset of $[n(n-1)/2]$ disjoint from $\{1\} \cup \supp(\mat{x}_r)$.  It follows that $S_{r+1}$ is disjoint from $\{1\} \cup \core(\hat{\mat{x}}_r)$ as well.  Let $\hat{\mat{y}}_r$ be the point that minimizes $\hat{f}$ in the direction $\hat{\mat{n}}_r$ from $\hat{\mat{x}}_r$.  Let $\hat{\mat{x}}_*$ be the optimal solution of the extended system (with an empty active set).  Recall $\mat{x}_{r+1}$ is the solution of the original system with respect to $S_{r+1}$ computed by SolveNNQ.  Let $\mat{x}_*$ be the optimal DKSG solution.  We prove that if $\norm{\hat{\mat{y}}_{r} - \hat{\mat{x}}_r} \geq \frac{1}{\lambda}\norm{\hat{\mat{x}}_r - \hat{\mat{x}}_*}$, then $f(\mat{x}_{r+1})$ is closer to $f(\mat{x}_*)$ than $f(\mat{x}_r)$ by a factor of $1 - \Theta(\frac{1}{\lambda\sqrt{n(n-1)}})$.

\vspace{6pt}

%Let $\mat{x}_*$ be the optimal solution of the original system with respect to $S_{r+1}$.  Let $\mat{x}_{r+1}$ the solution that SolveNNQ produces with respect to the active set $S_{r+1}$.  We prove that if $\norm{\mat{x}_{r+1} - \mat{x}_r} \geq \frac{1}{\lambda}\norm{\mat{x}_r - \mat{x}_*}$, then $\frac{f(\mat{x}_{r+1}) - f(\mat{x}_*)}{f(\mat{x}_r) - f(\mat{x}_*)} \leq 1 - \frac{1}{9\lambda\sqrt{n(n-1)}}$.

\begin{lemma}
\label{lem:converge}
%Let $S_{r+1}$ be the active set used by SolveNNQ for the next iteration.    
If $\norm{\hat{\mat{y}}_r  - \hat{\mat{x}}_r} \geq \frac{1}{\lambda}\norm{\hat{\mat{x}}_r - \hat{\mat{x}}_*}$ where $\lambda$ is a fixed value independent of $M$, then for a large enough $M$,
\[
\dfrac{f(\mat{x}_{r+1}) - f(\mat{x}_*)}{f(\mat{x}_r) - f(\mat{x}_*)} \leq 1 - \dfrac{1}{3\lambda\sqrt{2n(n-1)}}.
\]
\end{lemma}
\begin{proof}
	Let $\hat{\mat{z}}_{r+1}$ be the optimal solution of the extended system with respect to the active set $S_{r+1}$.  By Lemmas~\ref{lem:bound2} and~\ref{lem:ana-ratio} and the assumption that $\norm{\hat{\mat{y}}_r  - \hat{\mat{x}}_r} \geq \frac{1}{\lambda}\norm{\hat{\mat{x}}_r - \hat{\mat{x}}_*}$, we have 
	\[
	\frac{\hat{f}(\hat{\mat{x}}_r) - \hat{f}(\hat{\mat{z}}_{r+1})}{\hat{f}(\hat{\mat{x}}_r) - \hat{f}(\hat{\mat{x}}_*)} \geq \frac{1}{2\lambda\sqrt{2n(n-1)}}.
	\]
	
	We obtain a feasible solution $\mat{z}_{r+1} \in \real^{n(n-1)/2}$ of the original system from $\hat{\mat{z}}_{r+1}$ as follows.  Let $\tilde{\mat{z}}_{r+1}$ be the projection of $\hat{\mat{z}}_{r+1}$ to the first $n(n-1)/2$ coordinates.  First, we set $\mat{z}_{r+1} = \tilde{\mat{z}}_{r+1}$.  Second, for each $j > n(n+1)/2$, if $(\hat{\mat{z}}_{r+1})_j > 1$, then let $a = j - n(n+1)/2$, we pick an arbitrary $i \in \supp(\mat{z}_{r+1})$ such that $(\mat{B})_{a,i} = 1$, and we increase $(\mat{z}_{r+1})_i$ by $M(\hat{\mat{z}}_{r+1})_j - M$.  This completes the determination of $\mat{z}_{r+1}$.  Since we move the excess of the dummy variables over 1 to the positive coordinates of $\mat{z}_{r+1}$, one can verify that $\mat{B}\mat{z}_{r+1} \geq 1_n$.  So $\mat{z}_{r+1}$ is a feasible solution of the original system with respect to $S_{r+1}$.   Note that for any $j > n(n+1)/2$, $M(\hat{\mat{z}}_{r+1})_j - M$ must approach zero as $M$ increases.  Otherwise, as $\hat{f}(\hat{\mat{z}}_{r+1})$ consists of the term $\frac{1}{2}M^3 \cdot ((\hat{\mat{z}}_{r+1})_j - 1)^2 = \frac{1}{2}M \cdot (M(\hat{\mat{z}}_{r+1})_j - M)^2$, $\hat{f}(\hat{\mat{z}}_{r+1})$ would tend to $\infty$ as $M$ increases, contradicting the fact that $\hat{f}(\hat{\mat{z}}_{r+1}) \leq \hat{f}(\hat{\mat{x}}_{r})$ which is bounded as $M$ increases.  Therefore, as $M$ increases, $f(\mat{z}_{r+1}) = \mat{z}_{r+1}^t\mat{A}^t\mat{A}\mat{z}_{r+1}^{}$ tends to $\tilde{\mat{z}}_{r+1}^t\mat{A}^t\mat{A}\tilde{\mat{z}}_{r+1}^{}$ which is the first term in $\hat{f}(\hat{\mat{z}}_{r+1})$.  The second term of $\hat{f}(\hat{\mat{z}}_{r+1})$ is non-negative, which means that $f(\mat{z}_{r+1}) \geq f(\mat{x}_{r+1}) \geq \hat{f}(\hat{\mat{z}}_{r+1}) \geq \tilde{\mat{z}}_{r+1}^t\mat{A}^t\mat{A}\tilde{\mat{z}}_{r+1}^{}$.  We conclude that $\hat{f}(\hat{\mat{z}}_{r+1})$ tends to $f(\mat{x}_{r+1})$ as $M$ increases.
	
	By a similar reasoning, we can also show that $\hat{f}(\hat{\mat{x}}_*)$ tends to $f(\mat{x}_*)$ as $M$ increases.   By Lemma~\ref{lem:extend}(v), $\hat{f}(\hat{\mat{x}}_r)$ tends to $f(\mat{x}_r)$ as $M$ increases.  Hence, for a large enough $M$,
	\[
	\frac{f(\mat{x}_{r}) - f(\mat{x}_{r+1})}{f(\mat{x}_r) - f(\mat{x}_*)}  \geq
	\frac{2}{3} \cdot \frac{\hat{f}(\hat{\mat{x}}_{r}) - \hat{f}(\hat{\mat{z}}_{r+1})}{\hat{f}(\hat{\mat{x}}_r) - \hat{f}(\hat{\mat{x}}_*)} \geq \frac{1}{3\lambda\sqrt{2n(n-1)}}.
	\]
	Hence, $\dfrac{f(\mat{x}_{r+1}) - f(\mat{x}_*)}{f(\mat{x}_r) - f(\mat{x}_*)}  \leq 1 - \dfrac{1}{3\lambda\sqrt{2n(n-1)}}$.
	\end{proof}

	\vspace{6pt}
	
	As set in the definition of $\hat{\mat{x}}_r$ and as argued in the proof of Lemma~\ref{lem:converge}, the contributions of $M$ to the coordinates of $\hat{\mat{x}}_r$, $\hat{\mat{y}}_r$, and $\hat{\mat{x}}_*$ become negligible as $M$ increases.  Therefore, whether there exists a fixed $\lambda$ in Lemma~\ref{lem:converge} such that $\norm{\hat{\mat{y}}_r  - \hat{\mat{x}}_r} \geq \frac{1}{\lambda}\norm{\hat{\mat{x}}_r - \hat{\mat{x}}_*}$ is unaffected by $M$ when $M$ is large enough.  Still, we have to leave the condition $\norm{\hat{\mat{y}}_r  - \hat{\mat{x}}_r} \geq \frac{1}{\lambda}\norm{\hat{\mat{x}}_r - \hat{\mat{x}}_*}$ as an assumption to obtain the convergence result.
	
	\vspace{6pt}
	
	\begin{theorem}
		%\label{thm:main}
		Consider the application of \emph{SolveNNQ} on the \emph{DKSG} problem with $n$ points in $\real^d$.  The initialization of \emph{SolveNNQ} can be done in $O(dn^4) + T(n+\beta_0)$ time.    Suppose that the following assumptions hold.
		\begin{itemize}
			\item The threshold $\beta_1$ on the total number of iterations is not exceeded.
			\item For all $r \geq 1$, $\norm{\hat{\mat{x}}_{r} - \hat{\mat{y}}_{r}} \geq \frac{1}{\lambda}\norm{\hat{\mat{x}}_{r}-\hat{\mat{x}}_*}$, where $\hat{\mat{x}}_r$ is defined as in Lemma~\ref{lem:extend}, $\hat{\mat{n}}_r$ is a unit descent direction from $\hat{\mat{x}}_r$ that satisfies Lemma~\ref{lem:descent}, $\hat{\mat{y}}_r$ is the feasible point of the extended system that minimizes $\hat{f}$ in the direction $\hat{\mat{n}}_r$ from $\hat{\mat{x}}_{r}$, and $\hat{\mat{x}}_*$ is the optimal solution of the extended system.
		\end{itemize}
		Then, for all $r \geq 1$, $f(\mat{x}_{r+i}) - f(\mat{x}_*) \leq e^{-1} (f(\mat{x}_r) - f(\mat{x}_*))$ for some $i = O\bigl(\lambda n\bigr)$, and each iteration in \emph{SolveNNQ} takes $T(k + \beta_0)$ time, where $k = \max_{r \geq 1}|\supp(\mat{x}_r)|$.
	\end{theorem}
\begin{proof}
	By Lemma~\ref{lem:converge}, the gap $f(\mat{x}_r) - f(\mat{x}_*)$ decreases by a factor $e$ in $O(\lambda n)$ iterations.  Thus, 
	$f(\mat{x}_{r+i}) - f(\mat{x}_*) \leq e^{-1} (f(\mat{x}_r) - f(\mat{x}_*))$ for some $i = O\bigl(\lambda n\bigr)$.  It remains to analyze the running times.
	
We first compute $\mat{A}^t\mat{A}$.  Recall that $\mat{A}$ is $dn \times n(n-1)/2$, and $\mat{A}$ resembles the incidence matrix for the complete graph in the sense that every non-zero entry of the incidence matrix gives rise to $d$ non-zero entries in the same column in $\mat{A}$.   There are at most two non-zero entries in each column of the incidence matrix.  Therefore, computing $\mat{A}^t\mat{A}$ takes $O(dn^4)$ time.  We form an initial graph as follows.  Randomly include $\beta_0$ edges by sampling indices from $\bigl[n(n-1)/2\bigr]$.   Also, pick a node and connect it to all other $n-1$ nodes in order to guarantee that the initial graph admits a feasible solution for the DKSG problem.  There are at most $n+\beta_0-1$ edges.  The initial active set contains all indices in the range $\bigl[n(n-1)/2\bigr]$ except for the indices of the edges of the initial graph.  We extract in $O(n^2 + \beta_0^2)$ time the rows and columns of $\mat{A}^t\mat{A}$ that correspond to the edges of the initial graph.  Then, we call the solver in $T(n+\beta_0)$ time to obtain the initial solution $(\mat{u}_1,\mat{x}_1)$.  Therefore, the initialization time is $O(dn^4)+T(n+\beta_0)$.% $O(n^3+\beta_0^3) $ %O(dn^4+\beta_0^3)

At the beginning of every iteration, we compute $\nabla f(\mat{x}_r)$ and update the active set.  This can be done in $O(kn^2)$ time because $\mat{x}_r$ has at most $k$ non-zero entries.   Since the threshold $\beta_1$ is not exceeded, at most $k + \beta_0$ indices are absent from the updated active set.  We extract the corresponding rows and columns of $\mat{A}^t\mat{A}$ in $O(k^2 + \beta_0^2)$ time.  The subsequent call of the solver takes $T(k+\beta_0)$ time, which dominates the other processing steps that take $O(kn^2 + k^2 + \beta_0^2)$ time.  So each iteration runs in $T(k+\beta_0)$ time.
%By Lemma~\ref{lem:ana-3} with $\gamma = \Theta(n)$, we need $O\bigl(\lambda n\big)$ iterations to reduce $f(\mat{x}_r) - f(\mat{x}_*)$ by a factor $e$.$O(k^3+\beta_0^3)$$O(k^3+\beta_0^3+kn^2)$
\end{proof}

\section{Conclusion}

We presented a method that solves non-negative convex quadratic programming problems by calling a numerical solver iteratively.  This method is very efficient when the intermediate and final solutions have a lot of zero coordinates.  We experimented with two proximity graph problems and the non-negative least square problem in image deblurring.  The method is often several times faster than a single call of {\tt quadprog} of MATLAB.  From our computational experience, we obtain a speedup when the percentage of positive coordinates in the final solution drops below 50\%, and the speedup increases significantly as this percentage decreases.  We proved that the iterative method converges efficiently under the assumption.  We checked the assumption in our experiments on proximity graph and confirmed that the assumption hold.  A possible research problem is proving a better convergence result under some weaker assumptions.  The runtime of our method is sensitive to solution sparsity. As far as we know, we are the first to introduce a parameter $\tau$ explicitly to control the size of the intermediate problems. Another possible line of research is to explore the choice of $\tau$ given some prior knowledge of the solution sparsity.

\bibliography{graphref}
%\bibliographystyle{plain}

\end{document}